\documentclass[a4paper,11pt,fleqn]{article}
\usepackage[dvipsnames,svgnames,table,cmyk]{xcolor}

\usepackage{amsmath,amssymb,amsthm}
\usepackage{mathtools}
\usepackage{mathrsfs}
\usepackage{amsbsy}
\usepackage{bm}
\usepackage{cancel}
\usepackage{slashed}
\usepackage{multirow}

\usepackage[T1]{fontenc}
\usepackage[expansion=false]{microtype}

\usepackage[most,skins]{tcolorbox}
\usepackage{mdframed}

\newtcolorbox{resultbox}[1][]{
  enhanced,colback=blue!5,colframe=blue!50!black,
  fonttitle=\bfseries,title={#1},breakable}
\newtcolorbox{critbox}[1][]{
  enhanced,colback=red!5,colframe=red!50!black,
  fonttitle=\bfseries,title={#1},breakable}
\newtcolorbox{warnbox}[1][]{
  enhanced,colback=orange!8,colframe=orange!60!black,
  fonttitle=\bfseries,title={#1},breakable}

\usepackage{tikz}
\usetikzlibrary{calc, decorations.markings, arrows.meta,
  positioning, shapes.geometric}
\usepackage{tikz-cd}

\usepackage[
  a4paper,
  textwidth=17.2cm,
  textheight=21.9cm,
  left=1.5cm,
  right=1.5cm,
  top=1.5cm
]{geometry}

\usepackage{graphicx}
\usepackage{subfigure}
\usepackage[final]{pdfpages}

\usepackage{booktabs}
\usepackage{dcolumn}

\usepackage[numbers,sort&compress]{natbib}

\usepackage{verbatim}
\usepackage{appendix}

\theoremstyle{plain}
\newtheorem{theorem}{Theorem}[section]
\newtheorem{lemma}[theorem]{Lemma}
\newtheorem{corollary}[theorem]{Corollary}
\newtheorem{proposition}[theorem]{Proposition}
\theoremstyle{definition}
\newtheorem{definition}[theorem]{Definition}

\newcommand{\HomG}{\mathrm{Hom}_{G_2}}
\newcommand{\be}{\begin{eqnarray}}
\newcommand{\ee}{\end{eqnarray}}

\newcommand{\vp}{\varphi}
\newcommand{\Om}{\Omega}
\newcommand{\ps}{\psi}
\newcommand{\ta}{\tau}
\newcommand{\Nf}{{\cal B}_{\rm KK}(r)}
\newcommand{\vthree}{\mathrm{vol}_3}
\newcommand{\er}{e^r}
\newcommand{\Pint}{\Pi_{\mathrm{int}}}
\newcommand{\gone}{g_{11d}}
\newcommand{\gthree}{g_{3d}}
\newcommand{\gseven}{g_{7d}}
\newcommand{\grr}{g_{rr}}
\newcommand{\kap}{\kappa_{11}}
\newcommand{\Leff}{\mathcal{L}_{\mathrm{eff}}}
\newcommand{\iot}[1]{\iota_{#1}}
\newcommand{\str}{\star_7}
\newcommand{\sEl}{\star_{11}}
\newcommand{\Ical}{\mathcal{I}_7}

\definecolor{db5}{cmyk}{0.5,0.5,0,0.5}
\definecolor{mauve}{cmyk}{0.3,0.7,0.1,0.3}
\definecolor{palemauve}{cmyk}{0.3,0.7,0.1,0.0}
\definecolor{pb}{cmyk}{0.4,0.1,0,0.1}
\definecolor{pgreen}{cmyk}{0.4,0.0,0.3,0.0}
\definecolor{pink}{cmyk}{0.0,0.5,0.3,0.0}

\catcode`@=11
\def\seceqaa{\@addtoreset{equation}{section}
\def\theequation{A\arabic{equation}}}
\def\seceqbb{\@addtoreset{equation}{section}
\def\theequation{B\arabic{equation}}}
\def\seceqcc{\@addtoreset{equation}{section}
\def\theequation{C\arabic{equation}}}
\def\seceqdd{\@addtoreset{equation}{section}
\def\theequation{D\arabic{equation}}}
\def\seceqee{\@addtoreset{equation}{section}
\def\theequation{E\arabic{equation}}}
\def\seceqff{\@addtoreset{equation}{section}
\def\theequation{F\arabic{equation}}}
\def\seceqgg{\@addtoreset{equation}{section}
\def\theequation{G\arabic{equation}}}
\def\seceqhh{\@addtoreset{equation}{section}
\def\theequation{H\arabic{equation}}}
\def\seceqii{\@addtoreset{equation}{section}
\def\theequation{I\arabic{equation}}}
\def\seceqjj{\@addtoreset{equation}{section}
\def\theequation{J\arabic{equation}}}
\catcode`@=12
\newcommand{\rh}{r_h}
\newcommand{\anu}{\alpha_1}
\newcommand{\ade}{\alpha_2}
\newcommand{\ath}{\alpha_3}

\usepackage[colorlinks=true, linktocpage=true,
  linkcolor=black, citecolor=black, urlcolor=black]{hyperref}
\usepackage{url}
\newcommand{\logN}{\log N}
\newcommand{\gsNf}{g_s N_f}

\newcommand{\logrh}{\log r_h}
\newcommand{\logr}{\log r}
\begin{document}
\title{Contact Structure-induced Symplectic Nearly Half-Flat $SU(3)$ Structure in Finite-$N$ Thermal ${\cal M}$-Theory}
\author{Aalok Misra\\
Department of Physics, Indian Institute of Technology Roorkee \\
	~~Roorkee - 247667, Uttarakhand, India\\
	~~{\tt aalok.misra@ph.iitr.ac.in}}
\date{}
\maketitle
\noindent

\begin{abstract}
The ${\cal M}$-theory dual of thermal QCD-like theories constructed in~\cite{MQGP,OR4}
involves, as shown in~\cite{ACMS}, a closed $G_2$-structure seven-fold $\mathcal{M}_7$
supporting Contact 3-Structures~(C3S) that induce a ``transverse'' $SU(3)$-structure. We prove that this induced structure belongs to the symplectic nearly half-flat~(SHF)
class: $d\omega_\Phi=0$ exactly, while $d\Omega_+$ is suppressed as
$\mathcal{O}(e^{-N^{1/3}/\mathcal{O}(1)})$ in the infrared "MQGP limit",
leaving $W_\Phi^{SU(3)}=W_2^-$ as the sole non-vanishing $SU(3)$-structure
torsion class and completing the $G$-structure classification of the MQGP dual
initiated in~\cite{NPB,ACMS}.
As a physical application, Kaluza--Klein reduction of the eleven-dimensional Hodge dual of the differential of the $G_2$-structure positive three-form on $\mathcal{M}_7$, organised through the $G_2$/SHF torsion data, yields a radial
``magnetic field'' $\mathcal{B}_{\rm KK}(r)$ whose squared norm factorises exactly into external and internal factors.  $\mathcal{B}_{\rm KK}(r)$ satisfies a Bessel-type ODE with a regular singular point at the horizon; normalisability and infrared regularity select $J_2$, with harmonicity of the aforementioned internal factor recovered from the large-$N$ suppression. $\mathcal{B}_{\rm KK}$ is related to the magnetic field of the type IIA flavour $D6$-branes via a non-linear power law, furnishing a top-down holographic mechanism for generating weak magnetic fields in finite-temperature string theory.
\end{abstract}


\newpage
\tableofcontents

\section{Introduction}

The gauge/gravity duality, first proposed by Maldacena in 1997~\cite{Maldacena:1997re} and
subsequently elaborated by Gubser, Klebanov and Polyakov~\cite{Gubser:1998bc} and by
Witten~\cite{Witten:1998zw}, has provided the most powerful non-perturbative handle on
strongly coupled quantum field theories available within string theory. In its canonical form
the correspondence equates type IIB superstring theory on $\mathrm{AdS}_5\times S^5$ with
$\mathcal{N}=4$ supersymmetric Yang--Mills theory in the planar limit, and it has since been
generalized to a broad class of non-conformal, flavoured and finite-temperature backgrounds.
Of particular phenomenological urgency is the holographic description of Quantum Chromodynamics
(QCD) at finite temperature, motivated by the experimental discovery at the Relativistic
Heavy-Ion Collider (RHIC) and at the Large Hadron Collider (LHC) that the deconfined QCD
plasma formed in heavy-ion collisions behaves, in the temperature window $T_c<T\lesssim 2T_c$,
as a strongly coupled, near-perfect liquid --- the so-called strongly coupled Quark--Gluon
Plasma (sQGP)~\cite{Shuryak:2006se,Muller:2008zz} --- with a shear-viscosity-to-entropy-density
ratio $\eta/s$ close to the holographic lower bound $1/4\pi$~\cite{Kovtun:2004de}.

A central challenge in holographic QCD is the construction of a top-down gravity dual that
faithfully captures the dynamics of thermal QCD-like theories --- UV-conformal, IR-confining
and populated with fundamental-representation quarks --- at finite gauge coupling. Standard
AdS/CFT computations operate in the limit $g_s\to 0$ with the 't~Hooft coupling
$\lambda\equiv g_s N_c\to\infty$, suppressing both string-loop and $\alpha'$ corrections. In
the sQGP, however, the gauge coupling is finite; a holographic treatment faithful to this
regime therefore requires an ${\cal M}$-theory perspective in which $g_s\sim\mathcal{O}(1)$. The UV-complete type IIB holographic dual of large-$N$ thermal QCD-like theories
was constructed in \cite{Mia:2009wj}, incorporating $N_f$ flavor
D7-branes via the Ouyang embedding~\cite{Ouyang} on a fluxed resolved warped deformed conifold.
The Klebanov--Strassler (KS) solution~\cite{Klebanov:2000hb}, dual to the cascading
$SU(N+M)\times SU(N)$ gauge theory that confines in the IR after a sequence of Seiberg-duality
transformations, provides the infrared backbone of this construction. The delocalized
Strominger--Yau--Zaslow (SYZ) mirror~\cite{Strominger:1996it} of this type IIB background, obtained via three T-dualities along a special Lagrangian $T^3$ fibred over a large base, was
carried out by the author and collaborators~\cite{MQGP,OR4}, yielding a type IIA background
subsequently uplifted to eleven-dimensional ${\cal M}$-theory.

The ${\cal M}$-theory uplift is valid in the intermediate-$N$/coupling MQGP limit~\cite{MQGP}, \cite{ACMS} --- a combined
intermediate-string-coupling limit defined by $g_s\sim\frac{1}{{\cal O}(1)}$,
$N>1$ and $N_f, M\sim\mathcal{O}(1): g_s M^2<1, g_s N_f<1, \frac{\left(g_sM^2\right)^{m_1}\left(g_s N_f\right)^{m_2}}{N}<1, m_{1,2}\in\mathbb{Z}_+$, with $N_c$ given by the sum of the effective color $D3$-branes and the fractional $D3$-branes ($D5$-branes wrapping a vanising $S^2$), reducing to $3$ in the infrared after the end of the Seiberg-duality cascade. In this limit all flavour $D6$-branes of the type IIA SYZ mirror are uplifted to ${\cal M}$-theoretic Kaluza--Klein (KK) monopoles, geometrizing the matter content entirely into fluxes and geometry, and one obtains ${\cal M}$-theory on a seven-manifold $\mathcal{M}_7$ equipped with a
$G_2$~structure. This constitutes the first genuinely finite-coupling, top-down holographic
framework for sQGP phenomenology. Several key results have been derived within this
framework: NLO-in-$N$ corrections to shear viscosity and transport coefficients~\cite{Sil:2016jet};
glueball spectroscopy at finite gauge coupling~\cite{Sil:2019vbo};
vector and scalar meson spectra in close agreement with Particle Data Group
values~\cite{Misra:2017iig}; and $\mathcal{O}(l_p^6)$ corrections to the MQGP background
arising from $\mathcal{O}(R^4)$ terms in eleven-dimensional
supergravity ~\cite{OR4}. A systematic $G$-structure torsion-class classification of the
underlying non-supersymmetric geometries, particularly in the infrared and non-conformal sector,
was initiated in~\cite{NPB} and extended in~\cite{ACMS}.

The $G$-structure classification of non-supersymmetric string/${\cal M}$-theory backgrounds is
mathematically foundational. For supersymmetric compactifications, preserved spinors enforce
torsion conditions encoded in the $G$-structure classes. For a six-manifold with SU(3)
structure, the intrinsic torsion decomposes into five classes $W_1,W_2,W_3,W_4,W_5$,
catalogued by Chiossi and Salamon~\cite{Chiossi:2002}; this decomposition was related to
$\mathcal{N}=1$ supersymmetry conditions in non-K\"{a}hler string backgrounds by Cardoso
{\it et al.}~\cite{Cardoso:2002hd}, with a comprehensive review of flux compactifications
and $G$-structures provided in~\cite{Grana:2005jc}. The vanishing or non-vanishing of subsets
of these classes determines whether the geometry is Calabi--Yau (all vanish), complex
non-K\"{a}hler ($W_3\neq 0$), nearly K\"{a}hler ($W_1$ non-vanishing only), half-flat
($d\,\mathrm{Re}\,\Omega=0$, $d(J\wedge J)=0$), or symplectic half-flat ($dJ=0$,
$d\,\mathrm{Re}\,\Omega=0$). Each class carries distinct consequences for supersymmetry breaking, moduli stabilization and the spectrum of four-dimensional fields obtained by
Kaluza--Klein reduction. For non-supersymmetric thermal backgrounds the torsion classes are
generically non-vanishing; determining them is therefore essential for the
geometric consistency of the holographic vacuum and for extracting reliable physical predictions.

Manifolds of $G_2$ holonomy, and more generally $G_2$-structure seven-manifolds, are the
natural internal spaces for ${\cal M}$-theory compactifications to four dimensions~\cite{Acharya:2004qe,
Papadopoulos:1995da}. A torsion-free $G_2$~structure, defined by a closed and co-closed
associative three-form $\varphi$, yields a Ricci-flat internal space and a four-dimensional
$\mathcal{N}=1$ supersymmetric effective theory~\cite{Acharya:2004qe,Papadopoulos:1995da}.
The presence of background fluxes --- unavoidable in the thermal QCD duals under consideration
--- promotes the holonomy to a $G_2$~structure with intrinsic torsion; ${\cal M}$-theory
compactifications on such manifolds with $G_2$~structure have been studied
in~\cite{Lukas:2004ip,House:2004pm,Dall'Agata:2005ff}. The ${\cal M}$-theory seven-fold $\mathcal{M}_7$
of the MQGP construction was shown in~\cite{ACMS} to be a closed seven-fold supporting a
$G_2$~structure that induces Contact 3-Structures (C3S). Concretely, at a fixed conifold radial coordinate, $\mathcal{M}_7$ is the ${\cal M}$-theory circle $S^1_{\mathcal{M}}$ times a non-K\"{a}hler six-fold
$\mathcal{M}_6$, which is itself a thermal-circle fibration over a non-Einsteinian deformation
of the Sasaki--Einstein $T^{1,1}$. Reference~\cite{ACMS}
derived explicit Almost Contact (3) Metric Structures and a three-tuple of transverse SU(3)
structures induced on $\mathcal{M}_6$, and furthermore established a genuine Contact Structure
on $\mathcal{M}_7$ --- a stronger datum than an almost contact structure, relevant in the
context of~\cite{Arikan:2011zm,Arikan:2012gt,Todd:2015apo,Sasaki,Kuo-AC3S,3-Sasakian-geometry,
Friedrich:1997} --- with non-trivial topological consequences for the fibration. This geometric
infrastructure provides the starting point for the present analysis.

The symplectic half-flat (SHF) SU(3) structure class~\cite{Chiossi:2002,Conti:2006} are defined
by $dJ=0$ and $d\,\mathrm{Re}\,\Omega=0$. SHF manifolds, via the Hitchin flow~\cite{Hitchin:2001rw}, can be evolved over a transverse direction to
construct a seven-manifold with $G_2$ holonomy, and conversely the transverse geometry of a
$G_2$-structure manifold fibred over a line generically inherits a half-flat SU(3)
structure~\cite{Chiossi:2002}. Symplectic half-flat structures arise as type IIA supersymmetric
compactification backgrounds~\cite{Conti:2006}, appear as SYZ mirror partners of half-flat
manifolds, and have been fully classified in the non-compact homogeneous
setting~\cite{Podesta:2017}, where the Ricci tensor is shown to be Hermitian with respect to
the induced almost complex structure in every example. The mathematical analysis of these
structures connects to the broader context of $G_2$-structure moduli studied
in~\cite{Ossa et al[2013]}. In the Chiossi--Salamon decomposition~\cite{Chiossi:2002}, the SHF
condition restricts the torsion to the sub-classes $W_1^+$ and $W_2^+$ (with imaginary parts
forced to vanish), so that the geometry is simultaneously symplectic ($dJ=0$) and satisfies the
real part of the half-flat condition. 

A further Physics motivation for the present work comes from the phenomenology of magnetic fields in heavy-ion collisions. Peripheral collisions at RHIC and LHC generate transient magnetic fields
of unprecedented magnitude --- $|eB|\sim m_\pi^2\sim 10^{18}$~Gauss during the early collision
stage~\cite{Kharzeev:2007jp} --- that drive anomaly-induced transport phenomena. 
Holographic models of magnetized QCD plasma have been studied in frameworks using the Sakai--Sugimoto model~\cite{SS},~\cite{Rebhan:2010bd};
a derivation of the relevant magnetic field from a fully consistent top-down ${\cal M}$-theory
compactification --- respecting all eleven-dimensional supergravity equations of motion  --- has been absent. The KK reduction of the MQGP ${\cal M}$-theory dual on
$\mathcal{M}_7$ provides a systematic mechanism to generate this field from purely geometric
data, controlled by the $G_2$~torsion of $\mathcal{M}_7$ and the SHF structure of the
transverse $\mathcal{M}_6$.

The present paper establishes two main results. First, we show that the transverse SU(3)
structure  induced from the Contact 3-Structures of $\mathcal{M}_7$, as
constructed in~\cite{ACMS}, is symplectic half-flat up to exponential-in-$N$-suppressed corrections in the MQGP limit (\ref{eq:MQGP}) and in the infrared. The deviation from exact SHF is
governed by the ratio $e^{-N^{1/3}/\kappa_{r_h}}, \kappa_{r_h}\sim{\cal O}(1)$  --- so that the geometry is precisely termed \emph{nearly} symplectic
half-flat. This torsion-class identification completes the $G$-structure classification program
for the MQGP ${\cal M}$-theory dual initiated in~\cite{NPB,ACMS}, and reveals that the
$G_2$-structure seven-fold $\mathcal{M}_7$ encodes, through its Contact 3-Structure, a
distinguished and physically motivated SHF geometry on the ``transverse'' $SU(3)$. Second, as a
concrete application, we derive a weak magnetic field from the KK reduction of the
eleven-dimensional dual compactified on $\mathcal{M}_7$, demonstrating explicitly how the SHF
torsion structure governs the flux decomposition and the residual four-dimensional
electromagnetic content, thereby providing a top-down holographic mechanism for the generation
of magnetic fields relevant to sQGP phenomenology. The remainder of the paper is organized as
follows: Section~2 reviews the MQGP setup and $G_2$~geometry of $\mathcal{M}_7$. Section~3 briefly reviews Contact 3-Structures from a $G_2$-structure seven-fold. Section~4
establishes the SHF character of the transverse SU(3) structure. Section~5 presents the KK
reduction and magnetic field derivation. Section~6 concludes. There are two appendices - appendix \ref{Omegas} lists out the non-vanishing generalized spin connection components (quoting from \cite{ACMS}) and appendix \ref{HomG27270} derives ${\rm Hom}_{G_2}({\bf 27}, {\bf 7})=0$.

These results may be consolidated into a single statement, which constitutes the main
theorem of this paper; in what follows $r/r_h$ represents the dimensionless $\frac{r/r_h}{{\cal R}_{D5/\overline{D5}}}, {\cal R}_{D5/\overline{D5}}$ being the $D5-\overline{D5}$ separation of the parent type IIB dual \cite{metrics} of thermal QCD-like theories (the equivalence class of theories that are UV-conformal, IR-confining and have quarks in the fundamental representation of "color" and "flavor"). It combines the torsion-class identification of
Proposition~\ref{proposition} with the chain of results
Propositions~\ref{prop:normfact}, \ref{prop:ODE}, governing the
descended weak magnetic field.

\begin{theorem}[Main result]
\label{thm:mainresult}
Let $\mathcal{M}_7$ be the $G_2$-structure seven-fold that at a fixed Infra-Red valued (i.e. small) conifold radial coordinate is a product of the ${\cal M}$-theory circle $S^1_{\mathcal{M}}$ and a non-K\"{a}hler six-fold which is itself a thermal-circle fibration over a non-Einsteinian deformation of the Sasaki--Einstein $T^{1,1}$, in the context of ${\cal M}$-theory dual of thermal QCD-like theories carrying the Contact 3-Structures of~\cite{ACMS} in the intermediate-$N$ MQGP limit (\ref{eq:MQGP}), and let $(\omega_\Phi,\Omega)$ be the transverse $SU(3)$-structure it induces via
\eqref{Transverse-conditions}. Throughout, $\rh$ denotes the
horizon radius, $g_s$ the  string coupling, $N$ the number of color $D3$-branes, $M$ ($N_f$) the number of fractional $D3$-(flavor $D7$-) branes of the parent type IIB dual of \cite{Mia:2009wj}, and $\kappa_{r_h}=\bigl[3(6\pi)^{1/3}(g_sN_f)^{2/3}(g_sM^2)^{1/3}\bigr]^{-1}$
the $\mathcal{O}(1)$ constant in $|\log\rh|=\kappa_{r_h}N^{1/3}$ \cite{IITR-McGill-bulk-viscosity}. Then:
\begin{enumerate}
\item[\textup{(i)}] \emph{(Symplectic nearly half-flat geometry.)} In the IR and in the MQGP
limit the transverse $SU(3)$-structure is symplectic, $d\omega_\Phi = 0$, and the holomorphic
$(3,0)$-form is closed up to exponentially large-$N$-suppressed terms,
\begin{equation}
d\Omega_+\Big|_{r\in\mathrm{IR}}
= \kappa_{r_h}^{3}\,g_s^{1/12} M N_f^{4/3}\,
\frac{e^{-\frac{\kappa_{r_h}}{2}N^{1/3}}\log^2 N}{|\log\rh|^{5/3}}\,
dt\wedge e^{245}
= \mathcal{O}\!\left(e^{-N^{1/3}/\mathcal{O}(1)}\right),
\label{eq:thm-dOmega}
\end{equation}
so that, dropping the $e^{-\kappa_{r_h}N^{1/3}}$-suppressed terms, the only non-vanishing
intrinsic torsion class is
\begin{equation}
W_\Phi^{SU(3)} = W_2^-;
\label{eq:thm-Ws}
\end{equation}
equivalently $(\omega_\Phi,\Omega)$ is symplectic and nearly half-flat, the deviation from
exact half-flatness being exponentially suppressed in $N^{1/3}$.

\item[\textup{(ii)}] \emph{(magnetic field from KK reduction and its norm.)} The eleven-dimensional Hodge dual of the differential of the $G_2$-structure positive three-form, under KK reduction on $\mathcal{M}_7$ through the $G_2$/symplectic
half-flat $SU(3)$ structure of \textup{(i)}, generates a weak four-dimensional ``magnetic field'' ${\cal B}_{\rm KK}$* whose norm squared factorizes, into an external and an internal factor,
\begin{equation}
|\mathcal{M}|^2_{11d}
= |F_2\wedge\omega_1|^2_{g_{\mu\nu},\grr}\cdot
|\ta_1\wedge\Om_+|^2_{g_{mn}},
\label{eq:thm-normfact}
\end{equation}
with
\begin{equation}
|F_2\wedge\omega_1|^2
= \frac{|\omega_1|^2_{\gthree}}{\grr}(x^1)^2\bigl[(\partial_r\Nf)^2\bigr]
+ (\Nf)^2\mathcal{N}_\omega,
\qquad
\mathcal{N}_\omega = \left(g^{x^1x^1}\right)^2\bigl(|\omega_1|^2_{\gthree}
-\omega_\mu\omega_\nu g^{\mu x^1}g^{\nu x^1}g_{x^1x^1}\bigr),
\label{eq:thm-extnorm}
\end{equation}
and ${\cal B}_{\rm KK}$ given by (\ref{KK-Hodge dPhi}) - (\ref{B-M-theory-2}),
\begin{equation}
\Ical \equiv |\ta_1\wedge\Om_+|^2_{g_{mn}}
= |\ta_1|^2_{g_{mn}}|\Om_+|^2_{g_{mn}}
- \tfrac{1}{2}(\ta_1)^m(\ta_1)^{m'}(\Om_+)_{mnp}(\Om_+)_{m'}{}^{np}.
\label{eq:thm-I7def}
\end{equation}

\item[\textup{(iii)}] \emph{(Bessel-type radial equation.)} The radial profile $\mathcal{B}_{\rm KK}(r)$
of the descended field obeys the Bessel-type ordinary differential equation with a regular
singular point at the horizon $r=\rh$,
\begin{equation}
\anu(r-\rh)^2\mathcal{B}_{\rm KK}'(r)
+ \ade(r-\rh)^3\mathcal{B}_{\rm KK}''(r)
+ \ath\,\mathcal{B}_{\rm KK}(r)(r-\rh)^2 = 0,
\label{eq:thm-NeffEOM}
\end{equation}
where the real constants $\anu,\ade,\ath$ encode the integrated metric and torsion data of the
background. Its general (normalizable at the horizon) solution is given by: 
\begin{equation}
\mathcal{B}_{\rm KK}(r) = \mathcal{B}_{\rm KK}(r=r_h)
\frac{J_2\!\left(2\sqrt{\frac{\ath}{\ade}(r-\rh)}\right)}{r - r_h},
\quad r > \rh,
\label{eq:thm-solution}
\end{equation}
where one is free to choose $\mathcal{B}_{\rm KK}(r=r_h)$.
\end{enumerate}
\end{theorem}

\noindent Part~\textup{(i)} is proved by Proposition~\ref{proposition}
(Lemmas~\ref{LemmaSymplecticSU3}--\ref{HalfFlatSymplecticSU3}), the exponential
$e^{-N^{1/3}/\mathcal{O}(1)}$-suppression of $d\Omega_+$ being established in
\eqref{eq:dOmegaIR}, and parts~\textup{(ii)}--\textup{(iv)} are assembled from
Propositions~\ref{prop:normfact} and  \ref{prop:ODE}.

\noindent * The reason for referring to ${\cal B}_{\rm KK}(r)$ as magnetic, as discussed in (\ref{IR-EOM-BD6}) - (\ref{BD6-UV}), is summarized as follows. Consider turning on an $r$-dependent field ${\cal B}_{D6}(r): F_{x^1x^2} = \partial_{[x^1}A_{x^2]}, A_{x^1} = x^2 {\cal B}_{D6}(r)/2, A_{x^2} = - x^1{\cal B}_{D6}(r)/2$, on the world-volume $\Sigma^{(1+6)}\left(\cong S^1_t\times_w\mathbb{R}^3\right)\times_w S^2_{\rm squashed}$ of the dual type IIA flavor $D6$-branes results in the DBI action $\int_{\Sigma^{(1+6)}}e^{-\Phi^{IIA}}\sqrt{i^*g^{IIA} + F}, i:\Sigma^{(1+6)}\hookrightarrow M_{11}$. As is shown in Proposition \ref{prop:ODE} and the aforementioned equations,  ${\cal B}_{D6}(r\in{\rm UV})\sim\frac{1}{r - r_h}$ and  ${\cal B}_{KK}(r\in{\rm UV})\sim\frac{1}{(r - r_h)^{5/4}}\sim{\cal B}_{D6}^{5/4}(r\in{\rm UV})$; we further identity  ${\cal B}_{\rm KK}(r_h) = {\cal B}_{D6}^{5/4}(r=r_h)$.

\section{${\cal M}$-theory dual of thermal QCD-like theories and its geometric structures}

\subsection{From type IIB brane dynamics to ${\cal M}$-theory at intermediate coupling}

The program of constructing a top-down holographic dual of large-$N$ thermal QCD-like
theories --- by which we mean the equivalence class of theories that are UV-conformal,
IR-confining, and contain quarks in the fundamental representation of the "color" and "flavor" groups ---
at intermediate `t Hooft coupling requires working beyond classical supergravity. The reason
is conceptually straightforward: finite-coupling corrections on the gauge theory side translate
holographically into a regime where the relevant gravity description is no longer weakly coupled
type IIB supergravity but rather its non-perturbative, eleven-dimensional completion, ${\cal M}$-theory.
Specifically, genuine intermediate coupling necessitates the finite-$N$ limit:
$g_s^{-1} \sim \mathcal{O}(1)$--$\mathcal{O}(10)$ while keeping $M, N_f \sim \mathcal{O}(1)$
and $N > 1$. The condition
\begin{equation}
\hskip -0.4in g_s=\frac{1}{{\cal O}(1)},\  M, N_f = {\cal O}(1),\ g_sM^2 < 1,\  g_s N_f < 1,\ N>1,\  \frac{(g_s M^2)^{m_1}(g_s N_f)^{m_2}}{N} <1,
\quad m_{1,2} \in \mathbb{Z}^+ \cup \{0\},
\label{eq:MQGP}
\end{equation}
defines what has been termed the MQGP (M-theoretic Quark-Gluon Plasma) limit
\cite{MQGP,NPB, ACMS}. This is the operational regime for studying holographic
QGP physics at finite gauge coupling from a controlled ${\cal M}$-theory perspective.

The parent type IIB configuration consists of $N$ color D3-branes, $M$ fractional
(anti)-$D3$-brane pairs wrapping the vanishing $S^2$ of a resolved conifold, and $N_f$ flavour
(anti-)$D7$-branes embedded via the Ouyang embedding \cite{Ouyang}
\begin{equation}
\left(r^6 + 9a^2 r^4\right)^{1/4}
e^{\frac{i}{2}(\psi - \phi_1 - \phi_2)}
\sin\frac{\theta_1}{2}\sin\frac{\theta_2}{2} = \mu,
\quad |\mu| \ll r^{3/2}.
\label{eq:Ouyang}
\end{equation}
The Seiberg-like duality cascade that governs the RG flow progressively depletes the number
of effective color $D3$-branes from $N$ in the UV to $M$ in the deep IR, so that the
strongly-coupled IR gauge theory is effectively $SU(M)$. With the physically motivated
choices $M = 3$, $N_f = 2$ or $3$, and $g_s \approx 0.1$ --- corresponding, respectively,
to the number of IR color degrees of freedom in QCD, the number of light quark flavours,
and the QCD fine-structure constant at the electroweak scale --- one obtains
$N \approx 100 \pm \mathcal{O}(1)$ as the intermediate value of the color brane number at
which novel geometric structures emerge, a fact that  proved central to the Contact
3-Structure analysis of \cite{ACMS}.

To construct the ${\cal M}$-theory uplift, one first implements the Strominger--Yau--Zaslow (SYZ)
mirror of the type IIB geometry by performing a triple T-duality along three toroidal
isometry directions $(\phi_1, \phi_2, \psi)$ \cite{NPB,MQGP}.
Sequentially: T-duality along $\psi$ converts the $N$ $D3$-branes into $D4$-branes wrapping
that direction, with $M$ $D4$-branes straddling a pair of orthogonal NS5-branes; subsequent
T-dualities along $\phi_1$ and $\phi_2$ generate a pair of Taub-NUT geometries and convert
the D7-branes into D6-branes. Lifting the resulting type IIA configuration to eleven
dimensions eliminates all explicit brane sources --- the Taub-NUT spaces become
Kaluza-Klein monopoles in ${\cal M}$-theory --- so that the final ${\cal M}$-theory background is a purely
geometric flux compactification. The eleven-dimensional metric takes the form
\begin{equation}
ds^2_{11} = e^{-\frac{2\phi_{\rm IIA}}{3}}\!\left[
  \frac{1}{\sqrt{h(r,\theta_{1,2})}}\!\left(-g(r)\,dt^2 + d\vec{x}^{\,2}\right)
  + \sqrt{h}\left(\frac{dr^2}{g(r)} + ds^2_{\rm IIA}\right)\right]
  + e^{\frac{4\phi_{\rm IIA}}{3}}\!\left(dx^{11} + \mathcal{A}\right)^2,
\label{eq:11dmetric}
\end{equation}
where $g(r) = 1 - r_h^4/r^4$ is the blackening factor, $h(r,\theta_{1,2})$ is the warp
factor incorporating flux back-reaction, and $\mathcal{A}$ is the ${\cal M}$-theory circle one-form
generated from the type IIB RR fluxes under SYZ. Crucially, the ${\cal M}$-theory uplift is a
$G_2$-structure manifold with fluxes --- the correct geometric arena for the holographic
study of non-supersymmetric thermal QCD-like theories at intermediate coupling
\cite{House:2004pm,Dall'Agata:2005ff,Grana:2005jc}.

The eleven-dimensional supergravity action supplemented by the leading quantum-gravitational
$\mathcal{O}(l_p^6)$ corrections reads
\cite{Duff:1995wd,Green:1997di,Green:1997as,Russo:1997mk,Antoniadis:1997eg,Tseytlin:2000sf,Liu:2013dna}
\begin{align}
S_{D=11} &= \frac{1}{2\kappa_{11}^2}\Bigg[
  \int_{M_{11}}\!\sqrt{G}\,R
  + \int_{\partial M_{11}}\!\sqrt{h}\,K
  - \frac{1}{2}\int_{M_{11}}\!\sqrt{G}\,G_4^2
  - \frac{1}{6}\int_{M_{11}}\!C_3 \wedge G_4 \wedge G_4 \notag\\
&\quad + \frac{(4\pi\kappa_{11}^2)^{2/3}}{(2\pi)^4 \cdot 3^2 \cdot 2^{13}}
  \int_M \!d^{11}x\,\sqrt{G}\!\left(J_0 - \tfrac{1}{2}E_8\right)
  + 32\cdot 2^{13}\!\int\! C_3 \wedge X_8
  + \int t_8^2 G^2 R^3
  + \cdots\Bigg] - S_{\rm ct},
\label{eq:11daction}
\end{align}
where
\begin{align}
J_0 &= 3\cdot 2^8\!\left(
  R^{HMNK}R^P{}_{MNQ}R_H{}^{RSP}R^Q{}_{RSK}
  + \tfrac{1}{2}R^{HKMN}R^{PQ}{}_{MN}R_H{}^{RSP}R^Q{}_{RSK}
\right), \notag\\
E_8 &= \frac{1}{3!}\,\epsilon^{ABCM_1N_1\cdots M_4N_4}
       \epsilon_{ABCM'_1N'_1\cdots M'_4N'_4}
       R^{M'_1N'_1}{}_{M_1N_1}\cdots R^{M'_4N'_4}{}_{M_4N_4},
\end{align}
and $\kappa_{11}^2 = (2\pi)^8 l_p^9/2$, with $S_{\rm ct}$ the holographic counterterms
required for renormalisability at $\mathcal{O}(R^4)$ \cite{Gopal+Vikas+Aalok}. Near the Ouyang
embedding locus (small $\theta_{1,2}$), in the MQGP limit (\ref{eq:MQGP}) one establishes that
$\lim_{N\to\infty}E_8/J_0 = 0$ and $\lim_{N\to\infty}t_8 t_8 G^2 R^3/E_8 = 0$, so $E_8$
and $t_8^2 G^2 R^3$ contributions are suppressed and only $J_0$ survives at leading order
in $1/N$ \cite{OR4}. This truncation defines the operative
quartic-in-curvature ($\mathcal{O}(R^4)$) ${\cal M}$-theory effective action used throughout.

A metric ansatz $g_{MN} = g^{(0)}_{MN} + \beta\, g^{(1)}_{MN}$, with $\beta \sim l_p^6$,
and self-consistent truncation of the three-form potential correction, yields a well-posed
perturbative system whose solutions determine the $\mathcal{O}(R^4)$-corrected ${\cal M}$-theory
background \cite{OR4}. This truncation requires vanishing of (a linear combination of the) constants of integration appearing in the solutions to the EOMs for the ${\cal O}(R^4)$ corrections to the ${\cal M}$-theory metric components  precisely precisely along the non-compact four-cycle $\Sigma^{(4)}\cong \mathbb{R}_{\geq0}\times S^3_{\rm squashed}$ in the parent type IIB dual wrapped by the flavor $D7$-branes - referred to as "flavor memory" \cite{Gopal+Vikas+Aalok}; this has been implicitly used extensively to simplify the computations in this work. The IR-enhancement factors appearing in the metric
corrections --- schematically $\sim \beta\,(\log R_h)^m/R_h^n N^{\beta_N}$ with $m,n > 0$
and $\beta_N > 0$ --- compete against large-$N$ and Planckian suppression to determine the
domain of validity of the quartic truncation \cite{IITR-McGill-bulk-viscosity}. Choosing
$\beta \sim e^{-\gamma_\beta N^{\gamma_N}}$ with
$\gamma_\beta N^{\gamma_N} > 7\kappa_{r_h}N^{1/3}$ ensures that IR enhancement does not
overwhelm Planckian suppression, so one can consistently truncate at $\mathcal{O}(\beta)$.

This framework has enabled a sequence of precision agreements with QCD phenomenology: a
Hawking-Page deconfinement temperature compatible with lattice results
\cite{Misra:2020rlh,Gopal+Vikas+Aalok} --- with the remarkable feature that $T_c$ is not
renormalised by $\mathcal{O}(R^4)$ corrections, traceable via both a semiclassical
computation \cite{Witten:1998zw} and an entanglement entropy argument
\cite{Klebanov:2007ws,Vikas+Gopal+Aalok}; a conformal anomaly whose temperature variation
matches lattice data across the confinement-deconfinement transition \cite{Misra:2020rlh};
lattice-compatible shear-viscosity-to-entropy-density and bulk-viscosity-to-shear-viscosity
ratios \cite{bulk+gauge_KS+AM,IITR-McGill-bulk-viscosity,zeta-over-eta-int-coupling}; Particle Data Group-compatible
meson and glueball spectra \cite{Vikas+Gopal+Aalok,Sil:2019vbo}; and the
resolution of the long-standing mesino-meson isospectrality problem of the Sakai-Sugimoto
model \cite{SS} through the existence of QCD-compatible supermassive
inert mesinos \cite{Aalok+Gopal-Mesino}. The geometric underpinning of all these results is the
$G_2$-structure torsion class structure of the relevant seven-fold $M_7$.

\subsection{\texorpdfstring{$G_2$}{G2}-Structure, Almost Contact 3-Structures,
and the Road to Contact 3-Structures}

The seven-fold $M_7$ that enters the ${\cal M}$-theory uplift is the warped product
$S^1_M \times_w (S^1_t  \times_w M_5)$, where $M_5$ is a non-Einsteinian
deformation of the Sasaki-Einstein coset
$T^{1,1} = \bigl(SU(2)\times SU(2)\bigr)/U(1)$, and the thermal circle
$S^1_t $ represents the compactified Euclidean time direction at temperature
$T = \beta^{-1}$. This manifold carries a positive, stable co-associative three-form $\Phi$
and an explicit set of co-frames $\{e^a\}_{a=1}^{7}$, and accordingly admits a $G_2$
structure \cite{OR4,NPB}. Since $G_2 \subset SO(7)$, the adjoint of
$SO(7)$ decomposes under $G_2$ as $\mathbf{21} \to \mathbf{7} \oplus \mathbf{14}$, and the
torsion of the $G_2$ structure belongs to
\begin{equation}
\tau \;\in\; \Lambda^1 \otimes g_2^\perp
  = W_1 \oplus W_7 \oplus W_{14} \oplus W_{27}
  \;\equiv\; \tau_0 \oplus \tau_1 \oplus \tau_2 \oplus \tau_3.
\end{equation}
Working near the Ouyang embedding in the coordinate patch $\psi = 2n\pi$, $n = 0,1,2$, and
in the MQGP limit (\ref{eq:MQGP}), one establishes that $\tau_0 = 0$ while $\tau_{1,2,3}$ are all
non-vanishing, so
\begin{equation}
\tau(M_7) = \tau_1 \oplus \tau_2 \oplus \tau_3.
\end{equation}

The relevance of this $G_2$-structure torsion data to the generation of contact-geometric
structures on $M_7$ stems from a fundamental theorem: any closed seven-fold admitting a
$G_2$ structure necessarily supports at least three nowhere-vanishing orthonormal vector
fields, and thereby carries an Almost Contact Metric 3-Structure (ACM3S)
\cite{Arikan:2011zm,Arikan:2012gt,Todd:2015apo}. The key definitions are as follows. An
\emph{Almost Contact Structure} (ACS) on a $(2n+1)$-dimensional Riemannian manifold
$(Y,g)$ consists of an endomorphism $J\colon TY \to TY$, a unit Reeb vector field $R$, and
a one-form $\sigma$ satisfying \cite{Sasaki}
\begin{equation}
J^2 = -\mathrm{id} + R \otimes \sigma, \qquad \sigma(R) = 1;
\end{equation}
this reduces the structure group of $TY$ from $SO(2n+1)$ to $U(n)\times\{1\}$. The ACS is
promoted to an \emph{Almost Contact Metric Structure} (ACMS) when the metric compatibility
condition
\begin{equation}
g(Ju,Jv) = g(u,v) - \sigma(u)\sigma(v), \qquad \forall\;u,v \in \Gamma(TY),
\end{equation}
holds. This determines a fundamental two-form $\omega(u,v) = g(Ju,v)$ and a foliation
$\mathcal{F}_R$ of $Y$ by the integral curves of $R$, inducing a transverse metric
$ds^2_\perp$ via the orthogonal decomposition $ds^2 = \sigma^2 + ds^2_\perp$.

An \emph{Almost Contact 3-Structure} (AC3S) on $Y$ is a triple of mutually compatible
ACS's $\{(J^\alpha, R^\alpha, \sigma^\alpha)\}_{\alpha=1,2,3}$ satisfying the quaternionic
compatibility relations \cite{Kuo-AC3S}
\begin{equation}
J^\gamma = J^\alpha J^\beta - R^\alpha \otimes \sigma^\beta,
\qquad
R^\gamma = J^\alpha(R^\beta),
\qquad
\sigma^\gamma = \sigma^\alpha \circ J^\beta,
\label{eq:AC3S}
\end{equation}
for all cyclic permutations $(\alpha,\beta,\gamma)$ of $(1,2,3)$, together with
$\sigma^\alpha(R^\beta) = 0$ for $\alpha \neq \beta$. Dimensionally, $Y$ must be
$(4n+3)$-dimensional for an AC3S to exist, and the structure group further reduces to
$Sp(n)\times\mathbf{1}_3$. For a manifold with $G_2$ structure $\Phi$, the cross-product
$\times_\Phi$ provides a natural mechanism for constructing such triples: given any two
orthonormal Reeb fields $R^1, R^2$, the third is $R^3 = R^1 \times_\Phi R^2$, and the
compatibility relations \eqref{eq:AC3S} follow from the algebraic properties of the
associative three-form. An AC3S consisting of three contact structures satisfying
\eqref{eq:AC3S} defines a 3-Sasakian geometry \cite{3-Sasakian-geometry}.

An AC3S elevates to a \emph{Contact 3-Structure} (C3S) when each of the three contact
one-forms $\sigma^\alpha$ individually satisfies the non-degeneracy condition
\begin{equation}
\sigma^\alpha \wedge (d\sigma^\alpha)^n \neq 0
\qquad {\rm everywhere\ on }\ Y.
\end{equation}
This is a genuinely stronger requirement: almost contact structures are topologically
unobstructed on odd-dimensional manifolds (which all have vanishing Euler characteristic),
but the contact condition imposes differential constraints that depend sensitively on the
geometry of $Y$ and, in the present setting, on the parameter space
$(g_s, M, N_f;\, N)$.

For $M_7$, the AC3S in the large-$N$ MQGP limit (\ref{eq:MQGP}) is realised by the explicit triple
\cite{ACMS}
\begin{equation}
(R^1, R^2, R^3)
= \left(
    \sqrt{G^M_{x^0 x^0}}\,\partial_{x^0},\;
    \sqrt{G^M_{x^{10}x^{10}}}\,\partial_{x^{10}},\;
    -G^{\mu\nu}_M\,e^2_\nu\,\partial_\mu
  \right),
\label{eq:ReebTriple}
\end{equation}
with associated contact one-forms
\begin{equation}
(\sigma^1,\,\sigma^2,\,\sigma^3) = (e^1,\;e^7,\;-e^2)
\label{eq:ACMS_oneforms}
\end{equation}
and fundamental two-forms $\omega^\alpha_\Phi = d\sigma^\alpha = i_{R^\alpha}\Phi$. The
fluid-mechanical analogy is instructive \cite{Etnyre:2000,Kholodenko:2013}: the condition
$i_{R^\alpha}\omega^\alpha_\Phi = 0$ encodes a generalised fluid incompressibility condition
along each Reeb direction \cite{Etnyre:2000}, and the triple of AC3S one-forms mimics three
mutually orthogonal complex lamellar fields of the kind that arise in magnetohydrodynamics
and the physics of twisted grain boundary phases of liquid crystals \cite{Kholodenko:2013,
Lubensky:1995}. However, the contact condition
$\sigma^\alpha \wedge (\omega^\alpha_\Phi)^3 \neq 0$ fails for the choice
\eqref{eq:ACMS_oneforms} in the $N \gg 1$ regime: specifically,
$\sigma^\alpha \wedge (\omega^\alpha)^3 = 0$ in the $\psi = 2n\pi$,
$r = \mathrm{const}$ coordinate patches \cite{ACMS}.

The emergence of a genuine Contact 3-Structure requires passing to the intermediate-$N$
MQGP limit (\ref{eq:MQGP}) \cite{ACMS},
\begin{equation}
g_s^{-1} \sim \mathcal{O}(1){--}\mathcal{O}(10),\quad
M \sim \mathcal{O}(1),\quad
N_f \sim \mathcal{O}(1),\quad
\frac{(g_s M^2)^{m_1}(g_s N_f)^{m_2}}{N} < 1,
\label{eq:intermN}
\end{equation}
rather than $\ll 1$. In this regime, for the QCD-motivated values $(g_s, M, N_f) =
(0.1, 3, 2{ or }3)$, the color brane number is numerically pinned to
$N \approx 100 \pm \mathcal{O}(1)$. The Contact 3-Structure is then realised by the
more general ansatz (also discussed as ``Lemma 4'' of \cite{ACMS}, in \ref{C3S-Review})
\begin{equation}
\sigma^1 = \alpha_1 e^1 + \alpha_3 e^3 + \alpha_7 e^7,
\qquad
\sigma^2 = \beta_1 e^1 + \beta_4 e^4 + \beta_7 e^7,
\qquad
\sigma^3:\; \sigma^3(R^1 \times_\Phi R^2) = 1,
\label{eq:C3S}
\end{equation}
with the constraints $\alpha_1 \sim \alpha_3 \sim \alpha_7$ and
$\beta_1 \sim \beta_4 \sim \beta_7$ (equality up to $\mathcal{O}(1)$ prefactors in the IR).
The contact condition
\begin{equation}
\sigma^1 \wedge (\omega^1)^3
\;\sim\;
\alpha_1\alpha_3^2\alpha_7\;\Omega^3_{23}\Omega^3_{45}\Omega^6_1\;
e^{1234567} \neq 0
\end{equation}
is verified to hold in this intermediate-$N$ limit. Geometrically, whereas the large-$N$
AC3S mimics an ensemble of complex lamellar fields, the intermediate-$N$ Contact 3-Structure
mimics a triple of magnetic-monopole-like configurations on the $G_2$-structure seven-fold,
with monopole strengths set by the generalised spin connections $\Omega^a_{bc}$ arising from
the non-Einsteinian deformation of $T^{1,1}$ \cite{ACMS}.

A key structural consequence is that the parameter space
$\mathcal{X}_{G_2}(g_s, M, N_f;\, N)$ of closed seven-folds supporting $G_2$ structures
and relevant to the ${\cal M}$-theory uplift of thermal QCD-like theories is \emph{not}
$N$-path connected with respect to Contact Structures in the IR \cite{ACMS}. The
large-$N$ ACM3S branch (with $\alpha_3 = \beta_4 = 0$) and the intermediate-$N$ C3S branch
(requiring $\alpha_3 \sim \alpha_1 \sim \alpha_7$ and $\beta_4 \sim \beta_1 \sim \beta_7$)
are topologically disconnected in this parameter space: no continuous deformation in $N$
interpolates between them. This disconnection has a direct physical interpretation via the
swampland \cite{Ooguri:2016pdq,Freivogel:2016qwc}: since the deviation of the MQGP warp
factor from its near-horizon AdS form,
\begin{equation}
h_{\rm MQGP} - h_{\rm AdS}
\;\sim\;
\frac{g_s M^2 \log r\;(g_s N_f)^{0,1}}{N},
\end{equation}
is suppressed by additional inverse powers of $N$ in the intermediate-$N$ limit
\eqref{eq:intermN} relative to the large-$N$ limit \eqref{eq:MQGP}, the intermediate-$N$
vacuum supporting Contact 3-Structures sits further from the non-supersymmetric AdS
swampland than its large-$N$ almost-contact counterpart \cite{ACMS}.

Finally, via the proposition of de la Ossa, Larfors, and Magill \cite{Ossa2013},
the ACMS (resp.\ CMS) on $M_7$ induces a reduction of the $G_2$ structure to a transverse
$SU(3)$ structure on the six-dimensional space orthogonal to each Reeb direction
\cite{Friedrich:1997,Arikan:2011zm}. The decomposition of the $G_2$ three-form according
to the ACS reads
\begin{equation}
\Phi = \sigma^\alpha \wedge \omega^\alpha_\Phi + \Omega^\alpha_+,
\qquad
{*}_7\Phi = -\sigma^\alpha \wedge \Omega^\alpha_-
          + \tfrac{1}{2}\,(\omega^\alpha_\Phi)^2,
\label{eq:SU3decomp}
\end{equation}
where the pair $(\omega^\alpha_\Phi,\,\Omega^\alpha = \Omega^\alpha_+ + i\Omega^\alpha_-)$
constitutes the transverse $SU(3)$ structure --- a fundamental two-form and a
nowhere-vanishing holomorphic three-form on the transverse geometry. For the large-$N$
ACM3S, one obtains
\begin{equation}
\Omega^\alpha_+ = \Phi - \sigma^\alpha \wedge \omega^\alpha_\Phi,
\qquad
\Omega^\alpha_- = \sigma^\alpha \;\lrcorner\;
  \Bigl(\tfrac{1}{2}(\omega^\alpha_\Phi)^2 - {*}_7\Phi\Bigr).
\end{equation}
For the intermediate-$N$ C3S, the transverse $SU(3)$ structure is richer: $\Omega_-$ is
determined by the simultaneous conditions
\begin{equation}
i_{R^{(1)}}\Omega_- = 0,
\qquad
{*}_7\Phi = -\sigma^{(1)} \wedge \Omega_- + \tfrac{1}{2}\,\omega^2_\Phi,
\label{eq:OmMinusCond}
\end{equation}
and admits a four-parameter family of solutions parametrised by $\Lambda^{(1)}_{124}, \Lambda^{(1)}_{156}, \Lambda^{(1)}_{356}, \Lambda^{(1)}_{756}$ (see ``Lemma 6'' of \cite{ACMS}, also discussed in \ref{C3S-Review}):
\begin{equation}
\Omega_- = \Lambda_{ABC}\,e^{ABC}
= \Lambda^{(1)}_{1b'c'}\,e^{1b'c'}
+ \Lambda^{(1)}_{3b'c'}\,e^{3b'c'}
+ \Lambda^{(1)}_{7b'c'}\,e^{7b'c'}
+ \Lambda_{a'b'c'}\,e^{a'b'c'},
\quad a',b',c' \in \{2,4,5,6\}.
\label{eq:OmMinus}
\end{equation}
It is precisely the torsion class membership of this Contact Structure-induced transverse
$SU(3)$ structure \eqref{eq:SU3decomp}--\eqref{eq:OmMinus} that constitutes the central
object of investigation in the present paper.

{\it Main results}: We prove the following pair of propositions in this paper:

\begin{proposition}
The $SU(3)$-structure $(\omega_\Phi, \Omega)$ induced from any of the Contact 3-Structures obtained from the $G_2$-structure ($\tau = \tau_1\oplus\tau_2\oplus\tau_3$) seven-fold $M_7$ (of \cite{MQGP}, \cite{OR4}, \cite{ACMS}, relevant to the ${\cal M}$-theory dual of thermal QCD-like theories) with the following fibration structure:
\begin{equation}
\begin{array}{ccc}
    S^1_{\cal M} & \rightarrow & M_7 \\
     & & \big\downarrow\pi \\
     S^1_t & \rightarrow & M_6^{\rm non-Kaehler} \\
           &             & \big\downarrow \tilde{\pi} \\
           &             & {\cal T}^{1,1}_{\rm NE},
\end{array}
\end{equation}
(where ${\cal T}^{1,1}_{\rm NE}$ is a non-Einsteinian deformation of the Sasaki-Einstein $T^{1,1}$), is symplectic nearly half-flat: \\
(i) $d\omega_\Phi=0$,\\
(ii) $W_\Phi^{SU(3)} = W_2^-$, \\
(iii) $d\Omega_+= \kappa_{r_h}^3g_s^{1/12}M N_f^{4/3}\frac{e^{-\frac{\kappa_{r_h}}{2}N^{1/3}}\log^2N}{\log r_h|^{5/3}}dt\wedge e^{245},\ d\Omega_-\neq0$,
\end{proposition}
where ``nearly'' implies an exponential-in-$N^{1/3}$-suppression in the deviation from being half-flat. 
The second is that one can generate a weak magnetic field using KK Reduction via $G_2$/symplectic nearly half-flat $SU(3)$ Structure, and harmonicity via Large-$N$ Suppression.

\section{C3S from $G_2$-structure $M_7$: a review}
\label{C3S-Review}

A recap of some of the results of \cite{ACMS}. It was shown that:
\begin{eqnarray}
\label{Omega^a_bc-i}
& & de^a = \Omega^a_{bc}e^b\wedge e^c,
\end{eqnarray}
where $a, b, c = 2,...,6$ and, the ``generalized spin connection'' $\Omega^a_{bc}$ are defined as:
{\footnotesize
\begin{eqnarray}
\label{Omega^a_bc-ii}
& & \hskip -0.8in \Omega^a_{bc}(r={\rm constant}\in{\rm IR}) = \partial_{[\theta_2}e^a_{\ \theta_1]}\Theta_{2b}\Theta_{1c}
+ \partial_{[x}e^a_{\ \theta_1]}{\cal X}_b\Theta_{1c} + 
\partial_{[y}e^a_{\ \theta_1]}{\cal Y}_b\Theta_{1c} + \partial_{[z}e^a_{\ \theta_1]}{\cal Z}_{b}\Theta_{1c}
+ \partial_{[x}e^a_{\ \theta_2]}{\cal X}_b\Theta_{2c}\nonumber\\
& &  \hskip -0.8in  + \partial_{[y}e^a_{\ \theta_2]}{\cal Y}_b\Theta_{2c} + \partial_{[z}e^a_{\ \theta_2]}{\cal Z}_{2b}\Theta_{2c}
+ \partial_{[y}e^a_{\ x]}{\cal Y}_b{\cal X}_{c} +\partial_{[z}e^a_{\ x]}{\cal Z}_b{\cal X}_{c} +\partial_{[z}e^a_{\ y]}{\cal Z}_b{\cal Y}_{c};
\end{eqnarray}
}
\begin{eqnarray}
\label{Theta_ia+X_a+Y_a+Z_a}
& & \hskip -0.3 in d\theta_{i=1/2} = \sum_{a=2}^6\Theta_{ia}e^a,\  dx = \sum_{a=2}^6{\cal X}_ae^a,\  dy = \sum_{a=2}^6{\cal Y}_ae^a,\
dz = \sum_{a=2}^6{\cal Z}_ae^a.
\end{eqnarray}

The following lemma was proved in \cite{ACMS}.\\
\noindent "{\it Lemma 1}" (of \cite{ACMS}): Near the Ouyang embedding locus of the flavor $D7$-branes and the $\psi=2n\pi, n=0, 1, 2$-coordinate patches, assuming $|\mu_{\rm Ouyang}|\ll1$, in the MQGP limit (\ref{eq:MQGP}) and in the IR ($r\in r_h$), 
\begin{equation}
\label{G2-torsion}
\tau\left(M_7\right) = \tau_1\oplus\tau_2\oplus\tau_3,
\end{equation}
wherein (using $\Phi = e^{127} + e^{347} + e^{567} + e^{135} - e^{146} - e^{236} - e^{245}$, \cite{Karigiannis}, \cite{J. G. J. Held's thesis [2012]}),
\begin{eqnarray}
\label{taus}
& & \tau_0 = d\Phi\lrcorner*_7\Phi =0;\nonumber\\
& & \tau_1 = \frac{1}{12}*_7\left(\Phi \wedge *_7 d\Phi\right) \sim e^1\Omega^3_{24}-e^7\Omega^3_{23};\nonumber\\
& & \tau_2 = \frac{1}{2}\left(d\psi\lrcorner\Phi - *_7d\psi\right) -2\tau_1\lrcorner\Phi \sim \left(e^{27} + e^{35} - e^{46}\right)\Omega^3_{24} - J \Omega^3_{23};\nonumber\\
& & \tau_3 = *_7d\Phi - \tau_0\Phi + 3 \tau_1\lrcorner\psi \sim \left(e^2\wedge J + e^{457} + e^{367}\right)\Omega^3_{24}
- \left(e^{246} - e^{235} - e^{456}\right)\Omega^3_{23}
\end{eqnarray}
where $J = e^{12} + e^{34} + e^{56}$, and the $\sim$ implies the most dominant-in-$N$ terms in the MQGP limit (\ref{eq:MQGP}). 

In \cite{ACMS}, the following lemma was proven:\\
\noindent "{\it Lemma 4}" (of \cite{ACMS}): Near the $\psi=2n\pi$-coordinate patch, in the IR, \\ $(\sigma^1,\sigma^2,\sigma^3) = \left(\alpha_1 e^1 + \alpha_3 e^3 + \alpha_7 e^7, \beta_1 e^1 + \beta_4 e^4 + \beta_7 e^7,\sigma^3\right), $ $\{\alpha_{i=1, 2, 3}\}, \{\beta_{j=1, 2, 3}\}\in\mathbb{R}$, with $\sigma^3(R^1\times_\Phi R^2)=1$ with\\
(a)
\begin{eqnarray}
\label{contact-x0-5}
& & R^x_1=\frac{\alpha_3 e^3_{\theta_1}}{G^{\cal M}_{x\theta_1}};\nonumber\\
& & R^y_1=\frac{\alpha_3 (e^3_{\theta_2} G^{\cal M}_{x\theta_1}-e^3_{\theta_1} G^{\cal M}_{x\theta_2})}{G^{\cal M}_{x\theta_1}
   G^{\cal M}_{y\theta_2}};\nonumber\\
& & R^z_1=\frac{\alpha_3 (e^3_y G^{\cal M}_{x\theta_1}-e^3_{\theta_1} G^{\cal M}_{xy})}{G^{\cal M}_{x\theta_1} G^{\cal M}_{yz}};\nonumber\\
& & R^{\theta_1}_1=\frac{\alpha_3 (-e^3_{\theta_1} G^{\cal M}_{x\theta_2} G^{\cal M}_{yz} G^{\cal M}_{xz}-e^3_{\theta_1}
   G^{\cal M}_{xz} G^{\cal M}_{xy} G^{\cal M
}_{z\theta_2}-e^3_x G^{\cal M}_{x\theta_1} G^{\cal M}_{yz} G^{\cal M}_{z\theta_2}+e^3_y
   G^{\cal M}_{x\theta_1} G^{\cal M}_{xz} G^{\cal M}_{z\theta_2}+e^3_z G^{\cal M}_{x\theta_1} G^{\cal M}_{x\theta_2}
   G^{\cal M}_{yz})}{G^{\cal M}_{x\theta_1} G^{\cal M}_{x\theta_2} G^{\cal M}_{yz} G^{\cal M}_{z\theta_1}};\nonumber\\
& & R^{\theta_2}_1=\frac{\alpha_3 (e^3_{\theta_1} G^{\cal M}_{xz} G^{\cal M}_{xy}+e^3_x G^{\cal M}_{x\theta_1}
   G^{\cal M}_{yz}-e^3_y G^{\cal M}_{x\theta_1} G^{\cal M}_{xz})}{G^{\cal M}_{x\theta_1} G^{\cal M}_{x\theta_2} G^{\cal M}_{yz}},
\end{eqnarray}

(b) $\alpha_1\sim\alpha_3\sim\alpha_7; \beta_1\sim\beta_4\sim\beta_7$ with $\sim$ implying equality up to ${\cal O}(1)$ terms, and $R_2=R_1(e^3\rightarrow e^4,\ \alpha_3\rightarrow\beta_4)$, \\
{\it provide Contact 3-Structures}.

Now, consider the following proposition\\
\noindent {\it Proposition} \cite{Ossa et al[2013]}:
The ACMS induces a reduction of the $G_2$ structure to an $SU(3)$ structure $(\omega_\Phi, \Omega)$ on the transverse geometry of the foliation with $\omega_{\Phi}$ being the fundamental 2-form on $M_7$ ($\omega_\Phi = i_R\Phi$) and $\Omega$  the transverse $(3,0)$-form w.r.t. $J (J(u) = R\times_\Phi u,\ \forall u\in\Gamma(T M_7))$.
$(\omega_\Phi, \Omega)$ are determined by the ACS decomposition of the $G_2$ structure $\Phi$ on $M_7$:
\begin{eqnarray}
\label{Transverse-conditions}
& & \Phi = \sigma\wedge\omega_\Phi + \Omega_+,\nonumber\\
& & \psi = *_7\Phi = - \sigma\wedge\Omega_- + \frac{1}{2}\omega_\Phi\wedge\omega_\Phi.
\end{eqnarray}
It was shown in \cite{ACMS} that the Contact 3-Structures obtained in Lemma 4, using the aforementioned proposition, induced transverse $SU(3)$ 3-structures given by the following lemma.\\
\noindent {\it ``Lemma 6''}(of \cite{ACMS}): $M_7$, near the Ouyang embedding in the limit of very small limit of the Ouyang embedding parameter and near the $\psi=2n\pi, n=0, 1, 2$-coordinate patch, inherits a transverse $SU(3)$ structure from the Contact 3-Structures constructed in 3, $\left(\Omega_+^{(\alpha)},\Omega_-^{(\alpha)}\right)$,
where $\Omega_+^{(\alpha)} = \Phi - \sigma^{(\alpha)}\wedge\omega_\Phi^{(\alpha)}$ and $\Omega_-^{(\alpha)}$ is given in \cite{ACMS}. The steps leading up to the determination of $\Omega_-$ in \cite{ACMS}, are summarized below.

From $i_{R^{(1)}}\Omega_- = 0\cap *\Phi = - \sigma^{(1)}\wedge\Omega_- + \frac{1}{2}\omega_\Phi^2$:
\begin{eqnarray}
\label{Lambda-constraints}
& & \Lambda^{(3)}_{17a_0} = \Lambda^{(3)}_{27a_0}= \Lambda^{(3)}_{37a_0} = 0,\nonumber\\
& & \alpha_1\Lambda^{(3)}_{376} + \alpha_7\Lambda^{(3)}_{136} =  -1,\nonumber\\
& & |\Lambda_{245}|,   |\Lambda_{246}|,  |\Lambda_{256}| \ll 1;\  \Lambda_{456} \approx -1\nonumber\\
&  & \alpha_3 \Lambda^{(1)}_{124} - \alpha_1 \Lambda^{(1)}_{324} = {\cal O}(10^3)\alpha_\theta,\nonumber\\
& & \Lambda^{(3)}_{27a_0} = 0\ {\rm for}\ a_0 = 2, 4, 5, 6 \nonumber\\
&  & \alpha_7 \Lambda^{(1)}_{325} - \alpha_3 \Lambda^{(1)}_{725} = \frac{{\cal O}(10^5)e^{-2.2\alpha_\theta\kappa_{r_h}}}{\alpha_{\theta_1}^4};\nonumber\\
& & \alpha_1 \Lambda^{(1)}_{325} = \alpha_3 \Lambda^{(1)}_{125} \nonumber\\
& &  \alpha_7 \Lambda^{(1)}_{145} - \alpha_1 \Lambda^{(1)}_{745} + \frac{100\alpha_\theta^2}{\kappa_{\log r^{\rm IR}\ ^2}}  = 0\nonumber\\
& & \alpha_7 \Lambda^{(1)}_{325} - \alpha_3 \Lambda^{(1)}_{725} - \frac{10^5 e^{-10^{7/20} \alpha_\theta}}{\alpha_{\theta_1}^4} = 0,\nonumber\\
& & \alpha_3 \Lambda^{(1)}_{124} - \alpha_1 \Lambda^{(1)}_{324} -  {\cal O}(10^3) \alpha_\theta^2=0,\nonumber\\
& & \alpha_3 \Lambda^{(1)}_{125} = \alpha_1 \Lambda^{(1)}_{325};
\nonumber\\
& & \alpha_3 \Lambda^{(1)}_{126} = 
 \alpha_1 \Lambda^{(1)}_{326} = \frac{\alpha_1 \alpha_3}{\alpha_7} \Lambda^{(1)}_{726};\nonumber\\
& & \alpha_3 \Lambda^{(1)}_{145} = 
 \alpha_1 \Lambda^{(1)}_{345} = \frac{\alpha_1 \alpha_3}{\alpha_7} \Lambda^{(1)}_{745};
 \nonumber\\
& & \alpha_3 \Lambda^{(1)}_{146} = 
 \alpha_1 \Lambda^{(1)}_{346} = \frac{\alpha_1 \alpha_3}{\alpha_7} \Lambda^{(1)}_{746};\nonumber\\
& & \alpha_3 \Lambda^{(1)}_{156} = 
 \alpha_1 \Lambda^{(1)}_{356} = \frac{\alpha_1 \alpha_3}{\alpha_7} \Lambda^{(1)}_{756},\nonumber\\
& & \Lambda^{(1)}_{756}\sim\frac{N}{g_s^{14}}>>1,\ \Lambda^{(1)}_{724}\sim\frac{g_s^{21/2}}{N^{3/4}}<<1,
\end{eqnarray}
wherein $\Lambda^{(1),(3)}_{pqr}, p, q, r={1,...,7}$ are constants. 
Therefore, in the MQGP limit (\ref{eq:MQGP}),
\begin{equation}
\label{Omega-}
\Omega_- = \Lambda^{(1)}_{124}e^{124} + \Lambda^{(1)}_{156}e^{156} + \Lambda^{(1)}_{356} e^{356} + \Lambda^{(1)}_{756}e^{756}.
\end{equation}

The non-conformal corrections in all quantities $Q$ that can/are determined by field fluctuations have contributions that are functions of $\log  (r_h) $, which we would denote by $\tilde{f}(\log  (r_h) ) = f(\sqrt{-\alpha^2 + g^2})$ (using standard KS(Klebanov-Strassler)-like RG-flow equations \cite{Klebanov:2000hb}, \cite{IITR-McGill-bulk-viscosity}, one sees that (similar to \cite{effec_kin_model_zetaovereta}):
{\footnotesize
$g^2 \sim \left.\left(e^{-\phi^{\rm IIA}}\int_{S^2}B^{\rm IIA}\right)^{-1}\right|_{r=r_h}
\sim\frac{1}{(g_s M) (g_s N_f)}\left[\left|\log r_h\right| \left(\log N - 3 \log r_h\right)\right]^{-1},$}), $\alpha\in\mathbb{R}$, with $g$ being the temperature-dependent gauge coupling. Upon complexification $g\rightarrow g\otimes\mathbb{C}\Rightarrow Q\rightarrow Q\otimes\mathbb{C}$, $Q\otimes\mathbb{C}$ develops branch-point singularities connected by a branch cut in $g\otimes\mathbb{C}$; the branch cut is the analog of the $AC3S_{N\gg1}-C3S_{1<N\slashed{\gg}1}$-obstruction - Fig. 1.
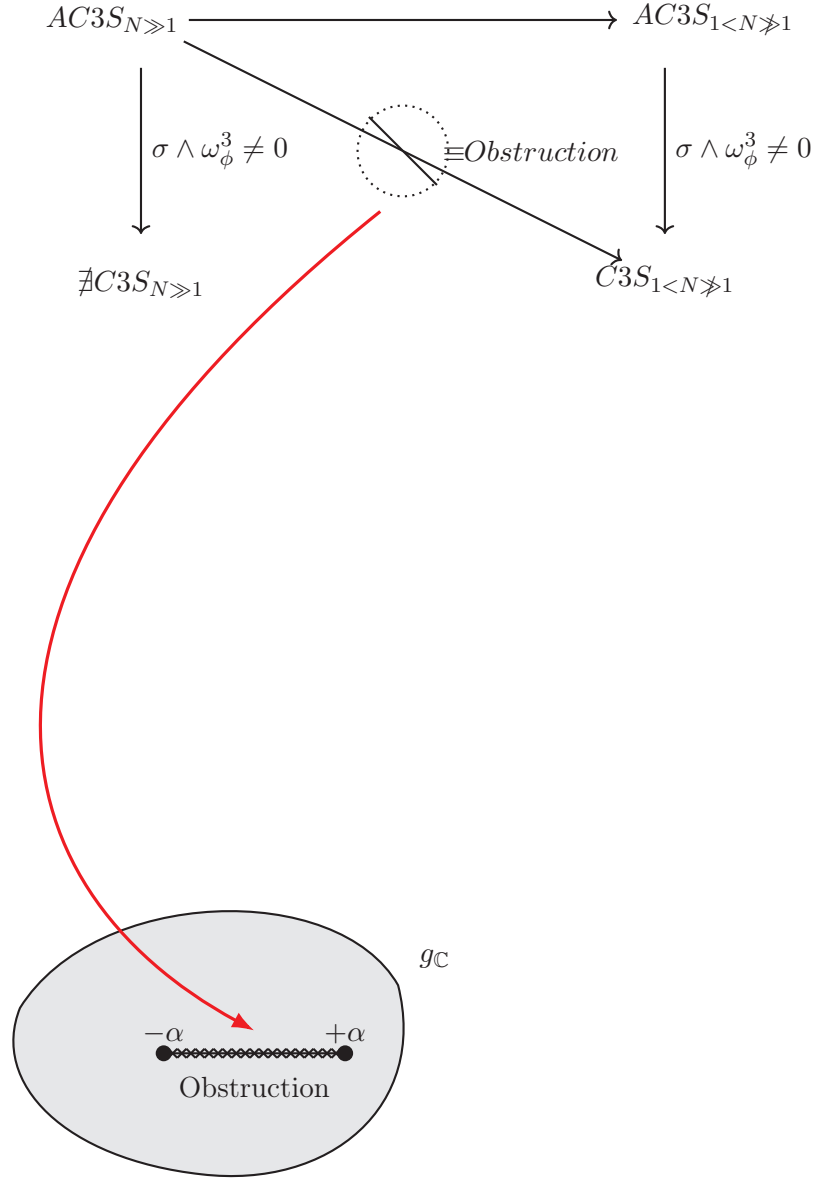
\begin{figure}
\centering
\begin{tikzpicture}[
  remember picture,
  arrow/.style={->, thick, shorten <=18pt, shorten >=18pt},
  redarrow/.style={red, very thick, ->, >=Latex}
]


\def\s{1.732}

\coordinate (a) at (0,\s*2);
\coordinate (b) at (\s*4,\s*2);
\coordinate (c) at (0,0);
\coordinate (d) at (\s*4,0);

\node at (a) {$\hskip -0.3in AC3S_{N\gg1}$};
\node at (b) {$\hskip 0.5in AC3S_{1<N\slashed{\gg}1}$};
\node at (c) {$\slashed{\exists}C3S_{N\gg1}$};
\node at (d) {$C3S_{1<N\slashed{\gg}1}$};

\draw[arrow] (a) -- (b);

\draw[arrow]
(a) -- (c)
node[pos=0.5,right]
{$\sigma\wedge\omega_\phi^3\neq0$};

\draw[arrow]
(b) -- (d)
node[pos=0.5,right]
{$\sigma\wedge\omega_\phi^3\neq0$};

\draw[arrow] (a) -- (d);

\coordinate (m) at ($(a)!0.5!(d)$);

\draw[thick]
($(m)+(-0.45,0.45)$) --
($(m)+(0.45,-0.45)$);

\draw[thick,dotted]
(m) circle [radius=0.6];

\node[right]
at ($(m)+(0.9,0)$)
{$\hskip -0.2in\equiv\hskip -0.05in
\scriptsize{{Obstruction}}$};

\coordinate (startred) at ($(m)+(-0.3,-0.8)$);


\begin{scope}[yshift=-10cm]

\fill[gray!20]
  (-1.6,0.4)
  .. controls (-0.6,2.0) and (2.6,2.1) ..
  (3.4,0.7)
  .. controls (3.8,-0.9) and (2.5,-2.0) ..
  (0.8,-1.8)
  .. controls (-1.0,-1.6) and (-2.0,-0.6) ..
  (-1.6,0.4)
  -- cycle;

\draw[thick]
  (-1.6,0.4)
  .. controls (-0.6,2.0) and (2.6,2.1) ..
  (3.4,0.7)
  .. controls (3.8,-0.9) and (2.5,-2.0) ..
  (0.8,-1.8)
  .. controls (-1.0,-1.6) and (-2.0,-0.6) ..
  (-1.6,0.4)
  -- cycle;

\node[anchor=west]
at (3.55,1.05)
{$g_{\mathbb{C}}$};

\coordinate (p1) at (0.3,-0.2);
\coordinate (p2) at (2.7,-0.2);

\fill (p1) circle (3pt);
\fill (p2) circle (3pt);

\node[above] at (p1) {$-\alpha$};
\node[above] at (p2) {$+\alpha$};

\draw[
  thick,
  postaction={
    decorate,
    decoration={
      markings,
      mark=between positions 0 and 1 step 0.05
      with {
        \draw (-0.06,-0.06) -- (0.06,0.06);
        \draw (-0.06,0.06) -- (0.06,-0.06);
      }
    }
  }
]
(p1) -- (p2);

\node at (1.5,-0.65)
{{Obstruction}};

\coordinate (endred) at (1.5,0.1);

\end{scope}


\draw[redarrow]
(startred)
.. controls (-3,-4) and (-2,-8) ..
(endred);
\end{tikzpicture}
\caption{Under an $N$-flow, one can not obtain C3S for $N\gg1$, but obtains C3S for $1<N\slashed{\gg}1$. The ``Obstruction'' in the $N$-flow diagram corresponds to the
branch-cut obstruction in the complexified gauge-coupling plane preventing from defining a single-valued holomorphic $Q_{\mathbb{C}}$ in the complexified gauge-coupling plane $g_{\mathbb{C}}$.
}
\end{figure}

\section{Symplectic nearly half-flat $SU(3)$ structure induced from C3S}

Let us now discuss how the Contact Structure-induced $SU(3)$ structure $\left(\omega_\Phi,\Omega\right)$ of (\ref{Transverse-conditions}), corresponds to a symplectic nearly half-flat $SU(3)$-structure $(\omega, \Omega)$; $d\omega=0, d\Omega_+={\cal O}\left(e^{-N^{1/3}/{\cal O}(1)}\right)$ (cf.\ \eqref{eq:dOmegaIR}), in the MQGP limit (\ref{eq:MQGP}).

\noindent \begin{proposition}\label{proposition}The torsion classes $W_\Phi^{SU(3)}$ of the $SU(3)$-structure of (\ref{Transverse-conditions}) induced from any of the C3S of "Lemma 4" of \cite{ACMS}, in the MQGP limit (\ref{eq:MQGP}), are given by: $W_\Phi^{SU(3)} = W_2$; the $SU(3)$-structure is symplectic \end{proposition}
This proposition is proved by proving two Lemmas \ref{LemmaSymplecticSU3} and \ref{HalfFlatSymplecticSU3}.
\noindent 
\begin{lemma}
\label{LemmaSymplecticSU3}
 The torsion classes of the $SU(3)$-structure of (\ref{Transverse-conditions}) induced from any of the C3S of "Lemma 4" of \cite{ACMS}, are given by:
$W_\Phi^{SU(3)} \in W_2\oplus W_5$; the $SU(3)$-structure is symplectic.
\end{lemma}

\begin{proof}
 The $SU(3)$-structure torsion classes, as is well known, are generically given as under:
\begin{eqnarray}
\label{SU3-torsion-defs}
& & d\omega_\Phi = \frac{3}{2}\Im m\left(\bar{W}_1\Omega\right) + W_4\wedge\omega_\Phi + W_3,\nonumber\\
& & d\Omega = W_1\omega_\Phi^2 + W_2\wedge\omega_\Phi + \bar{W}_5\wedge\Omega,
\end{eqnarray}
from where one reads off:
\begin{eqnarray}
\label{W345}
& & W_3 = \left(d\omega_\Phi\right)^{(2,1)} - \left(J\wedge W_4\right)^{(2,1)},\nonumber\\
& & W_4 = \frac{1}{2}\omega_\Phi\lrcorner d\omega_\Phi,\nonumber\\
& & W_5 = \frac{1}{2}\Omega_+\lrcorner d\Omega_+.
\end{eqnarray}
Also, 
\begin{eqnarray}
\label{W12}
& & d\Omega_+\wedge\omega = \Omega_+\wedge d\omega_\Phi = W_1^+J^3,\nonumber\\
& & d\Omega_-\wedge\omega_\Phi = \Omega_-\wedge d\omega_\Phi = W_1^-J^3;\nonumber\\
& & \left(d\Omega_+\right)^{(2,2)} = W_1^+ J^2 + W_2^+\wedge J,\nonumber\\
& & \left(d\Omega_-\right)^{(2,2)} = W_1^- J^2 + W_2^-\wedge J.
\end{eqnarray}
As $d\omega_\Phi=0$, one immediately sees:
\begin{equation}
\label{SU3-torsion-CS}
W_\Phi^{SU(3)} = W_2\oplus W_5,
\end{equation}
which could be thought of as a symplectic $SU(3)$-structure. 
\end{proof}

\begin{lemma} 
\label{HalfFlatSymplecticSU3}
The symplectic $SU(3)$-structure of {\it Lemma I}, is nearly half flat: $W_\Phi^{SU(3)} = W_2^-$. 
\end{lemma}

\begin{proof}
Using (\ref{Transverse-conditions}), one sees:
\begin{eqnarray}
\label{dOmega+}
d\Omega_+ = d\Phi - \omega_\Phi^2,
\end{eqnarray}
wherein,
\begin{equation}
\label{dPhi-G2torsion}
d\Phi = 4\tau_0 *_7\Phi - 3\tau_1\wedge\Phi - *_7\tau_3;    
\end{equation}
the $\tau_{0,1,3}$ are as given in (\ref{taus}). Alternatively, from \cite{ACMS}, one knows that:
\begin{eqnarray}
\label{Phi-dPhi}
& & \hskip -0.3in \Phi= e^{127}+e^{347}+e^{567}+e^{135}-e^{146}-e^{236}-e^{245},\nonumber\\
& & \hskip -0.3in d\Phi= \left(e^{2135}+e^{2146}+e^{3127}+e^{3146}-e^{5127}-e^{5146}\right)\Omega^1_2 \nonumber\\
& & \hskip -0.3in+\left(e^{2734}+e^{2756}+e^{3712}+e^{3756}-e^{4712}-e^{4756}-e^{5712}-e^{5734}\right)\Omega^7_2 \nonumber\\
& & \hskip -0.3in + \left(e^{2317}-e^{2345}-e^{2517}-e^{2536}-e^{2617}-e^{2645}+e^{3517}-e^{3617}-e^{3645}+e^{5617}\right)\Omega^2_{23}\nonumber\\
& & \hskip -0.3in + \left(e^{2347}+e^{2315}+e^{2547}+e^{2647}+e^{2615}+e^{3547}-e^{3526}+e^{3647}+e^{3615}+e^{5647}\right)\Omega^3_{23}+e^{2415}\Omega^3_{24}\nonumber\\
& & \hskip -0.3in + \Biggl(e^{2367}+e^{2567}+e^{3567}-e^{3524}-e^{3624}-e^{5624}+e^{2357}-e^{2314}+e^{2514}-e^{2657}+e^{2614}-e^{3657}+e^{3614} \nonumber\\
& & \hskip -0.3in -e^{5614}-e^{5623}\Biggr)\Omega^5_{23}\nonumber\\
& & \hskip -0.3in \approx \left(e^{2347}+e^{2315}+e^{2547}+e^{2647}+e^{2615}+e^{3547}-e^{3526}+e^{3647}+e^{3615}+e^{5647}\right)\Omega^3_{23}+e^{2415}\Omega^3_{24}.
\end{eqnarray}

It was estimated in \cite{IITR-McGill-bulk-viscosity}: $|\log  (r_h) |=\kappa_{r_h}N^{1/3}, \kappa_{r_h}\equiv \frac{1}{{\cal O}(1)},$ i.e., $r_h=e^{-\kappa_{r_h}N^{1/3}}$ constant. Hence, in the IR,
\begin{eqnarray}
\label{dPhi-components}
& & \Omega^1_2 = \frac{200 \sqrt[4]{N} \sin ^2\left(\theta_2\right) \csc
   ^2\left(\theta_1\right) \cot (\theta_2)}{3 {g_s}^{3/4} M \log
   ({r_h}) {\cal D}},\nonumber\\
& &  \Omega^7_2=\frac{400 \sqrt[4]{N} \sin ^2\left(\theta_2\right) \csc
   ^2\left(\theta_1\right) \cot (\theta_2)}{3{g_s}^{3/4} M \log  (r_h) {\cal D}},\nonumber\\
 & & \Omega^2_{23}=\frac{16,927 \sqrt[4]{N} \sin \left(\theta_2\right) \csc
   ^4\left(\theta_1\right) \left({r_h}^2 \log
   ^2({r_h})-0.03 {r_h} \log
   ({r_h})-0.01\right)}{{g_s}^{3/4} M {\cal D} },\nonumber\\
& & \Omega^3_{23}=7.3 {g_s}^{7/4} M \sqrt[4]{\frac{1}{N}}  N_f 
   {r_h} \csc ^2\left(\theta_1\right) \csc \left(\theta_2\right) \log
   (N) \log ^2({r_h}),\nonumber\\
& & \Omega^3_{24}=8 {g_s}^{7/4} M \sqrt[4]{\frac{1}{N}}  N_f 
   \log ^2(N) \log ({r_h}),\nonumber\\
& & \Omega^5_{23}=-\frac{3096.2 \sqrt[4]{N} \sin \left(\theta_2\right) \csc
   ^2\left(\theta_1\right)}{{g_s}^{3/4} M  {\cal D}}.           
\end{eqnarray}
In the IR, $r\in[r_h, {\cal R}_{D5/\overline{D5}}] = r_h[1,1 + \frac{g_sM^2(c_1 + c_2 \log \tilde{r_h})}{N}]$. One thus sees that $\Omega^3_{24}$ is the most dominant generalized spin connection component in the IR and in the MQGP limit (\ref{eq:MQGP}).  As,
{\footnotesize
\begin{eqnarray}
\label{GMtt}
& & \hskip -0.8in g^{\cal M}_{tt} = \sqrt{\frac{1}{N}} r^2
   \left(1-\frac{{r_h}^4}{r^4}\right)\nonumber\\
   & & \hskip -0.8in \times\frac{\left(\frac{8 \pi 
   \left(-\frac{12 a^2 {g_s} M^2  N_f  ({c_1}+{c_2} \log
   ({r_h}))}{9 a^2+r^2}-\frac{9 M^2 \log (r) \left(2 {g_s}
    N_f  \log \left(\sin \left(\theta_1\right) \sin \left(\theta
   _2\right)\right)+12 {g_s}  N_f  \log (r)+6 {g_s}  N_f -2
   {g_s}  N_f  \log (4)+8 \pi \right) g_s\mathbb{Y}}{256 \pi ^3}\right)}{N}+3
   g_s\mathbb{Y}\right)}{8
   \sqrt[3]{3} \pi ^{7/6} \sqrt{{g_s}} \sqrt[3]{ N_f  \left(-\log
   \left(9 a^2 r^4+r^6\right)\right)+\frac{8 \pi }{{g_s}}-4  N_f 
   \log \left(\sin \left(\theta_1\right) \sin \left(\theta
   _2\right)\right)+4  N_f  \log (4)}}\nonumber\\
& & \hskip -0.8in + \frac{\left(9 b^2+1\right)^3 \left(4374 b^6+1035 b^4+9
   b^2-4\right) \beta  b^8 M \left(\frac{1}{N}\right)^{11/4} r^2 \Sigma _1
   \left(6 a^2+{r_h}^2\right) (r-{r_h})^2
   \left(1-\frac{{r_h}^4}{r^4}\right) \log ({r_h})\mathbb{Y}^{2/3}}{72 \sqrt[3]{3} \pi
   ^{13/6} \left(18 b^4-3 b^2-1\right)^5 \sqrt{{g_s}}  (\log N) ^2
    N_f  {r_h}^2 \alpha _{\theta_2}^3 \left(9
   a^2+{r_h}^2\right)},\nonumber\\
   & & 
\end{eqnarray}
}
where $\mathbb{Y}\equiv  \left( N_f 
   \left(-\log \left(9 a^2 r^4+r^6\right)\right)+\frac{8 \pi }{{g_s}}-4
    N_f  \log \left(\sin \left(\theta_1\right) \sin \left(\theta
   _2\right)\right)+4  N_f  \log (4)\right)$,\\ $\Sigma_1 \equiv N^{6/5}\left(19683\sqrt{6}\sin^6\theta_1 + 6642\sin\theta_2\sin^3\theta_1 - 40\sqrt{6}\sin^4\theta_2\right)$, we hence see that:
\begin{eqnarray}
\label{dPhi-LO-IR}
& & \hskip -0.5in \left.d\Phi\right|_{r\in{\rm IR}} \approx \left.\Omega^3_{24}e^{2415}\right|_{r\in{\rm IR}}=\frac{3^{1/3}}{2^{1/6}g_s^{7/12}\pi^{7/12}}\frac{{g_s}^{7/6} M  N_f  \sqrt{{r_h}} \log ^2(N)
   \sqrt{r-{r_h}} \log ({r_h}) \sqrt[3]{\cal D}}{\sqrt{N}}dt\wedge e^{245}\nonumber\\
& & \hskip -0.5in \sim \kappa_{r_h}^{4/3}\frac{g_s^{1/12}M N_f^{4/3}e^{-\frac{\kappa_{r_h}}{2}N^{1/3}}\log^2N}{N^{5/9}}dt\wedge e^{245},
\end{eqnarray}
where $|\log r_h|=\kappa_{r_h}N^{1/3}, \kappa_{r_h}\equiv\frac{1}{3(6\pi)^{1/3}(g_sN_f)^{2/3}(g_sM^2)^{1/3}}$ \cite{IITR-McGill-bulk-viscosity}. One hence sees that in the IR, $d\Phi$ is IR- and large-$N$ suppressed. 

Further, it was shown in \cite{ACMS} that $\omega_\Phi^2$ consists of the following four kinds of contributions:
\begin{eqnarray}
\label{omegaPhisq-kinds-contributions}
& (i) & \Omega^1_3\Omega^3_{bc}e^{31bc}, \Omega^1_a\Omega^3_{3b}e^{a13b}, 
\Omega^1_a\Omega^3_{bc} e^{a1bc};\nonumber\\
& (ii) & \Omega^1_{[3}\Omega^7_{a]} e^{31a7};\nonumber\\
& (iii) & \Omega^3_{[bc}\Omega^7_{a]} e^{bca7}, \Omega^3_{[3c}\Omega^7_{a]}e^{3ca7},
\Omega^3_{[bc}\Omega^7_{3]}e^{bc37};\nonumber\\
& (iv) & \Omega^3_{[bc}\Omega^3_{df]}e^{bcdf}. 
\end{eqnarray}
It is further understood that we are working in the IR in the following.

\begin{itemize}
\item (i) From \cite{ACMS}, ones notes ($\kappa_{\Omega^1_a\Omega^3_{bc}} \ll 1$):\\
(I)
{\scriptsize
\begin{eqnarray}
\label{Omega13Omega3ab}
& & \hskip -0.8in \Omega^1_3\Omega^3_{23} = \frac{{5.1\kappa_{\Omega^1_a\Omega^3_{bc}}\times10^{-2}} {g_s}  (\log N) 
    N_f  {r_h}^2 \sin \left(\theta_2\right) \csc ^4\left(\theta
   _1\right) \log ({r_h}) (1. {r_h} \log ({r_h})-0.9)
   \cot (\theta_2)}{(r-{r_h}) {\cal D}};\nonumber\\
& & \hskip -0.8in \Omega^1_3\Omega^3_{24} = \frac{{2.7\kappa_{\Omega^1_a\Omega^3_{bc}}\times10^{2}} {g_s}  (\log N) ^2
    N_f  {r_h} \sin ^2\left(\theta_2\right) \csc ^2\left(\theta
   _1\right) \left({r_h}^2 \log ^2({r_h})-0.1 {r_h} \log
   ({r_h})-0.002\right) \cot (\theta_2)}{(r-{r_h})
   {\cal D}},\nonumber\\
& & \hskip -0.8in \Omega^1_3\Omega^3_{25} = \frac{{0.95\kappa_{\Omega^1_a\Omega^3_{bc}}} {g_s}  (\log N) 
    N_f  {r_h}^2 \sin \left(\theta_2\right) \csc ^4\left(\theta
   _1\right) \log ({r_h}) ({r_h} \log ({r_h})+0.029) \cot
   (\theta_2)}{(r-{r_h}) {\cal D}},\nonumber\\
& & \hskip -0.8in \Omega^1_3\Omega^3_{26} = -\frac{{0.55\kappa_{\Omega^1_a\Omega^3_{bc}}} {g_s}  (\log N) 
    N_f  {r_h}^2 \sin \left(\theta_2\right) \csc ^4\left(\theta
   _1\right) \log ({r_h}) ({r_h} \log ({r_h})-0.05) \cot
   (\theta_2)}{(r-{r_h}) {\cal D}},\nonumber\\
& & \hskip -0.8in \Omega^1_3\Omega^3_{34} = -\frac{{6.1\kappa_{\Omega^1_a\Omega^3_{bc}}\times10^{-4}} {g_s}^{9/2}
    (\log N) ^3 M^2 \sqrt{\frac{1}{N}}  N_f ^3 {r_h} \csc
   ^2\left(\theta_1\right) \log ^2({r_h}) ({r_h} \log
   ({r_h})-0.7) ({r_h} \log ({r_h})+0.004) \cot
   (\theta_2) (1.  (\log N) +263. {r_h} \log
   ({r_h}))}{(r-{r_h}) {\cal D}},\nonumber\\
& & \hskip -0.8in \Omega^1_3\Omega^3_{35}=\frac{{3.9\kappa_{\Omega^1_a\Omega^3_{bc}}\times10^{-2}} {g_s}  (\log N) 
    N_f  {r_h} \sin \left(\theta_2\right) \csc ^4\left(\theta
   _1\right) \left({r_h}^2 \log ^2({r_h})+0.06 {r_h} \log
   ({r_h})-{1.710^{-16}}\right) \cot
   (\theta_2) (\log (N)+263. {r_h} \log
   ({r_h}))}{(r-{r_h}) {\cal D}},\nonumber\\
& & \hskip -0.8in \Omega^1_3\Omega^3_{36}=\frac{{0.17\kappa_{\Omega^1_a\Omega^3_{bc}}} {g_s}  (\log N) 
    N_f  {r_h} \sin \left(\theta_2\right) \csc ^4\left(\theta
   _1\right) ({r_h} \log ({r_h})-0.05) ({r_h} \log
   ({r_h})+0.09) \cot (\theta_2) ({2.22\times10^{-16}}  (\log N) -33.3333 {r_h} \log
   ({r_h}))}{(r-{r_h}) {\cal D}},\nonumber\\
& & \hskip -0.8in \Omega^1_3\Omega^3_{45} = 10^{-2}\kappa_{\Omega^1_a\Omega^3_{bc}}\frac{{g_s}^{9/2}  (\log N) ^3 M^2 \sqrt{\frac{1}{N}}  N_f ^3
   {r_h} \csc ^2\left(\theta_1\right) \log ^2({r_h}) ({r_h}
   \log ({r_h})+0.004) ({r_h} \log ({r_h})+0.04) \cot
   (\theta_2) (\log (N)+263. {r_h} \log
   ({r_h}))}{ (r-{r_h}) {\cal D}},\nonumber\\
& & \hskip -0.8in \Omega^1_3\Omega^3_{46} = -\kappa_{\Omega^1_a\Omega^3_{bc}}\frac{{6.3\times10^{-3}} {g_s}^{9/2}
    (\log N) ^3 M^2 \sqrt{\frac{1}{N}}  N_f ^3 {r_h} \csc
   ^2\left(\theta_1\right) \log ^2({r_h}) ({r_h} \log
   ({r_h})-0.05) ({r_h} \log ({r_h})-0.008) \cot
   (\theta_2) (\log (N)+263. {r_h} \log
   ({r_h}))}{(r-{r_h}) {\cal D}},\nonumber\\
& & \hskip -0.8in \Omega^1_3\Omega^3_{56}=-\frac{{0.26\kappa_{\Omega^1_a\Omega^3_{bc}}} {g_s}  N_f 
   {r_h} \sin \left(\theta_2\right) \csc ^4\left(\theta_1\right)
   \log (N) ({r_h} \log ({r_h})-0.05) ({r_h} \log
   ({r_h})+0.09) \cot (\theta_2) ({2.22\times10^{-16}}  (\log N) -33.33 {r_h} \log
   ({r_h}))}{(r-{r_h}) {\cal D}}.\nonumber\\
   & &
\end{eqnarray}
}
(II)
{\scriptsize
\begin{eqnarray*}
& & \hskip -0.8in \Omega^1_2\Omega^3_{34}+\Omega^1_4\Omega^3_{23} = -500 {g_s}^{\frac{7}{2}}  (\log N) ^3 M^2
   \sqrt{\frac{1}{N}}  N_f ^2 \csc ^2\left(\theta_1\right)r_h^2 \log^3({r_h}) \csc (\theta_2),\nonumber\\
& & \hskip -0.8in  \Omega^1_2\Omega^3_{35}+\Omega^1_5\Omega^3_{23} = 8 {g_s}\kappa_{\Omega^1_a\Omega^3_{bc}}  N_f \nonumber\\
& & \hskip -0.8in = \times \frac{ {r_h} \sin \left(\theta_2\right) \csc
   ^4\left(\theta_1\right) \log (N) \cot (\theta_2)
   \left({r_h}^2 ({0.21}-{0.14} \log (N))
   \log ^2({r_h})+{r_h} ({0.02}
   \log (N)-{0.04}) \log
   ({r_h})+{8.3\times10^{-3}} \log
   (N)-{2.3} {r_h}^3 \log
   ^3({r_h})\right)}{(r-{r_h}) {\cal D}},\nonumber\\
& & \hskip -0.8in  \Omega^1_2\Omega^3_{36}+\Omega^1_6\Omega^3_{23} = -\frac{264.9  {g_s}  (\log N)   N_f  {r_h} \sin \left(\theta
   _2\right) \csc ^4\left(\theta_1\right) \log ({r_h}) \cot
   (\theta_2)}{{\cal D}},\nonumber\\
& & \hskip -0.8in  \Omega^1_3\Omega^3_{34} = -\frac{{6.1\times10^{-4}\kappa_{\Omega^1_a\Omega^3_{bc}} } {g_s}^{9/2}
    (\log N) ^3 M^2 \sqrt{\frac{1}{N}}  N_f ^3 {r_h} \csc
   ^2\left(\theta_1\right) \log ^2({r_h}) ({r_h} \log
   ({r_h})+0.004) (1. {r_h} \log ({r_h})-0.7) \cot
   (\theta_2) (1.  (\log N) +263. {r_h} \log
   ({r_h}))}{(r-{r_h}) {\cal D}},\nonumber\\
& & \hskip -0.8in  \Omega^1_3\Omega^3_{35} =  \frac{{0.04\kappa_{\Omega^1_a\Omega^3_{bc}} } {g_s}  (\log N) 
    N_f  {r_h} \sin \left(\theta_2\right) \csc ^4\left(\theta
   _1\right) \left({r_h}^2 \log ^2({r_h})+0.06 {r_h} \log
   ({r_h})\right) \cot
   (\theta_2) (\log (N)+263. {r_h} \log
   ({r_h}))}{(r-{r_h}) {\cal D}},\nonumber\\
& & \hskip -0.8in  \Omega^1_3\Omega^3_{36} = \frac{{0.17\kappa_{\Omega^1_a\Omega^3_{bc}} } {g_s}  N_f 
   {r_h} \sin \left(\theta_2\right) \csc ^4\left(\theta_1\right)
   \log (N) ({r_h} \log ({r_h})-0.05) ({r_h} \log
   ({r_h})+0.09) \cot (\theta_2) ({2.2\times10^{-16}} \log (N)-33.3333 {r_h} \log
   ({r_h}))}{(r-{r_h}) {\cal D}},\nonumber\\
\end{eqnarray*}
{\scriptsize
\begin{eqnarray}  
\label{Omega1aOmega33b} 
& & \hskip -0.8in  \Omega^1_4\Omega^3_{35}-\Omega^1_5\Omega^3_{34} = \frac{\kappa_{\Omega^1_a\Omega^3_{bc}} }{(r-{r_h}) {\cal D}^2}\Biggl\{4.5 {g_s}^{9/2}  (\log N) ^3 M^2 \sqrt{\frac{1}{N}}
    N_f ^3 {r_h} \csc ^2\theta_1 \log
   ^2({r_h}) \Biggl(\cot (\theta_2) 
\biggl({g_s}  N_f 
   \log (\alpha_{\theta_1} ) \mathbb{X}_1+{g_s}  N_f  \log (\alpha_{\theta_2} )
  \mathbb{X}_1\nonumber\\
& & \hskip -0.8in-{0.14} {g_s}
    (\log N) ^2  N_f  {r_h}^2 \log
   ^2({r_h})+{3.6\times10^{-16}} {g_s}
    (\log N) ^2  N_f +{0.74}
   {g_s}  (\log N)   N_f  {r_h}^2 \log
   ^2({r_h})\nonumber\\
& & \hskip -0.8in+{0.4} {g_s}
    (\log N)   N_f  {r_h}^2 \log ^2({r_h}) \log
   r_h -{10^{-15}}
   {g_s}  (\log N)   N_f  \log
   r_h  -
   {g_s}  (\log N)   N_f  {r_h} \log ({r_h}) \log
   r_h +0.02{g_s}  (\log N)   N_f 
   {r_h} \log ({r_h})\nonumber\\
& & \hskip -0.8in-{6.6} {g_s}  (\log N) 
    N_f +{6.2\times10^{-2}} {g_s}  N_f 
   {r_h} \log ^2(N) \log ({r_h})+{0.36} {g_s}  N_f  {r_h}^2 \log
   ^2({r_h})-{4} {g_s}
    N_f  {r_h}^2 \log ^2({r_h}) \log
  ({r_h})-{0.2}
   {g_s}  N_f  {r_h} \log ({r_h}) \log
   r_h \nonumber\\
& & \hskip -0.8in+{0.23}
   {g_s}  N_f  {r_h} \log ({r_h})-{2.3}  (\log N)  {r_h}^2 \log
   ^2({r_h})-{0.85}  (\log N) 
   {r_h} \log ({r_h})+{7.7\times10^{-2}}
    (\log N) +{14} {r_h}^2 \log
   ^2({r_h})+{1.2} {r_h} \log
   ({r_h})\biggr)\nonumber\\
& & \hskip -0.8in+\csc \theta_2\biggl({r_h}^2 \log
   ^2({r_h}) \biggl(-{8.9\times10^{-2}}
   {g_s}  (\log N)   N_f  \log \biggl(\sin(\theta_1) \sin
   (\theta_2)\biggr)+{g_s} ({0.57}-{0.1}
    (\log N) )  N_f  \log
   r_h +{5.4\times10^{-14}}
   {g_s}  (\log N)   N_f \nonumber\\
& & \hskip -0.8in+{0.6} {g_s}  N_f  \log \biggl(\sin(\theta_1) \sin
   (\theta_2)\biggr)-{0.4}
   {g_s}  N_f +{0.7}
    (\log N) -{6.8}\biggr)+{r_h}
   \log ({r_h}) \biggl(-{3.6\times10^{-2}}
   {g_s}  (\log N)   N_f -{5.8\times10^{-2}} {g_s}  N_f  \log
   r_h +{0.16}
   {g_s}  N_f +{0.5}
    (\log N) \nonumber\\
& & \hskip -0.8in+6\times 10^{-3}\biggr)+ (\log N) 
   \biggl(-{1.4\times10^{-3}} {g_s}  (\log N) 
    N_f +{4.3\times10^{-3}} {g_s}  N_f 
   \log  (r_h)  +{1.3\times10^{-2}}
   {g_s}  N_f \nonumber\\
& & \hskip -0.8in+{2.2\times10^{-2}}\biggr)+{r_h}^3 \log ^3({r_h})
   \biggl(-{1.7} {g_s}  N_f 
   \log ({r_h})-{2.7}
   {g_s}  N_f +{9.1}\biggr)\biggr)\Biggr)\Biggr\},\nonumber\\
& & \hskip -0.8in  \Omega^1_4\Omega^3_{36} - \Omega^1_6\Omega^3_{34} = -\frac{3 {g_s}^{9/2}  (\log N) ^3 M^2 \sqrt{\frac{1}{N}}
    N_f ^3 {r_h} \csc ^2\left(\theta_1\right) \csc \left(\theta
   _2\right) \log ^3({r_h})}{\cal D},\nonumber\\
& & \hskip -0.8in  \Omega^1_5\Omega^3_{36} - \Omega^1_6\Omega^3_{35} \nonumber\\
& & \hskip -0.8in= -\frac{{2\times10^{-14}} {g_s}
    (\log N)   N_f  {r_h} \sin \left(\theta_2\right) \csc
   ^4\left(\theta_1\right) ({r_h} \log ({r_h})-0.05) ({r_h} \log ({r_h})+0.09) \cot (\theta_2)
   ({2.2\times10^{-16}}
    (\log N) -33.3 {r_h} \log ({r_h}))}{(r-{r_h}) {\cal D}},
\end{eqnarray}
}}
where
{\scriptsize
\begin{eqnarray}
\label{Xs}
& & \hskip -0.8in \mathbb{X}_1 \equiv  \left(({0.3}  (\log N) -{2.8}) {r_h}^2 \log
   ^2({r_h})+(-{0.1}
    (\log N) -{0.2}) {r_h} \log
   ({r_h})-{7.2\times10^{-3}}
    (\log N) \right),\nonumber\\
& &
\end{eqnarray}
}
(III)
{\footnotesize
\begin{eqnarray*}
& & \hskip -0.8in  \Omega^1_2\Omega^3_{34} + \Omega^1_3\Omega^3_{42} + 
 \Omega^1_4\Omega^3_{23}\nonumber\\
 & &  \hskip -0.8in {\scriptsize = -\frac{{270\kappa_{\Omega^1_a\Omega^3_{bc}} } {g_s}
    (\log N) ^2  N_f  {r_h} \sin ^2\left(\theta_2\right) \csc
   ^2\left(\theta_1\right) \left(1. {r_h}^2 \log
   ^2({r_h})-0.12 {r_h} \log ({r_h})-0.002\right)
   \cot (\theta_2)}{(r-{r_h}) {\cal D}},}\nonumber\\
& &  \hskip -0.8in \Omega^1_2\Omega^3_{35} + \Omega^1_3\Omega^3_{52} + 
 \Omega^1_5\Omega^3_{23}  =    \frac{3.1 \kappa_{\Omega^1_a\Omega^3_{bc}} {g_s}  (\log N)   N_f  {r_h} \sin
   \left(\theta_2\right) \csc ^4\left(\theta_1\right) \cot
   (\theta_2)}{(r-{r_h}) {\cal D}}\nonumber\\
& & \hskip -0.8in \times  \left(({2.1}-{0.36}  (\log N) ) {r_h}^2 \log
   ^2({r_h})+({0.09}-{0.02}  (\log N) ) {r_h} \log
   ({r_h})-{2.4\times10^{-3}}
    (\log N) +{0.28}
   {r_h}^3 \log ^3({r_h})\right),\nonumber\\
& &  \hskip -0.8in \Omega^1_2\Omega^3_{36} + \Omega^1_3 \Omega^3_{62} + 
 \Omega^1_6\Omega^3_{23} = -\frac{1.4 \kappa_{\Omega^1_a\Omega^3_{bc}}  {g_s}  (\log N)   N_f  {r_h} \sin \left(\theta
   _2\right) \csc ^4\left(\theta_1\right) \cot (\theta_2)
   }{(r-{r_h}) {\cal D}}\nonumber\\
& & \hskip -0.8in \times \left(({3.6\times10^{-2}}
    (\log N) -{0.64}) {r_h}^2 \log
   ^2({r_h})+({1.3\times10^{-2}}
    (\log N) -{0.21}) {r_h} \log
   ({r_h})+{0.012}
    (\log N) \right),\nonumber\\
 & &  \hskip -0.8in \Omega^1_2\Omega^3_{45} + \Omega^1_4\Omega^3_{52} + 
 \Omega^1_5\Omega^3_{24}\nonumber\\
& &  \hskip -0.8in = -\frac{{320\kappa_{\Omega^1_a\Omega^3_{bc}} } {g_s}  (\log N) ^2
    N_f  {r_h} \sin ^2\left(\theta_2\right) \csc ^2\left(\theta
   _1\right) \left(1. {r_h}^2 \log ^2({r_h})-0.12 {r_h} \log
   ({r_h})-0.002\right) \cot (\theta_2)}{(r-{r_h}) {\cal D}},\nonumber\\
\end{eqnarray*}
\begin{eqnarray} 
\label{Omega1aOmega3bc}  
   & &  \hskip -0.8in \Omega^1_2\Omega^3_{46} + \Omega^1_4\Omega^3_{62} + 
 \Omega^1_6\Omega^3_{24} = \frac{0.02 {g_s}^{9/2} M^2 \sqrt{\frac{1}{N}}  N_f ^3 \csc
   ^2\left(\theta_1\right) \log ^4(N) \log ^2({r_h}) \cot
   (\theta_2)}{{\cal D}},\nonumber\\
 & &  \hskip -0.8in \Omega^1_3\Omega^3_{45} + \Omega^1_4\Omega^3_{53} + 
 \Omega^1_5\Omega^3_{34} =   -\frac{9\kappa_{\Omega^1_a\Omega^3_{bc}} }{5(r-{r_h})
   {\cal D}^2}\Biggl\{1.5 {g_s}^{9/2} M^2 \sqrt{\frac{1}{N}} {N_f}^3
   {r_h}^2 \csc ^2\left(\theta_1\right) \log ^3(N) \log
   ^3({r_h}) \Biggl(-{0.28} {g_s} {N_f} \csc \left(\theta_2\right)
   \log (N)\nonumber\\
   & &  \hskip -0.8in +{0.11}
   {g_s} {N_f} \csc \left(\theta_2\right) \log (N) \log
   ({r_h})-{0.53}
   {g_s} {N_f} \log (N) \log ({r_h}) \cot
   (\theta_2)  +{0.071}
   {g_s} {N_f} \csc \left(\theta_2\right) \log (N) \log (\sin
   (\theta_1)\sin (\theta_2))\nonumber\\
   & &  \hskip -0.8in -{0.11} {g_s} {N_f} \log ^2(N) \cot
   (\theta_2)+{0.29}
   {g_s} {N_f} \log (N) \cot
   (\theta_2)   -{0.36}
   {g_s} {N_f} \log (N) \cot (\theta_2) \log (\sin(\theta_1)\sin
   (\theta_2))\nonumber\\
& &  \hskip -0.8in -{0.29} {g_s} {N_f} \csc \left(\theta
   _2\right)  -{0.64}
   {g_s} {N_f} \csc \left(\theta_2\right) \log
   ({r_h})+{0.38}
   {g_s} {N_f} \log ({r_h}) \cot
   (\theta_2)\nonumber\\
& &  \hskip -0.8in   -{0.43}
   {g_s} {N_f} \csc \left(\theta_2\right) \log (\sin
   (\theta_1) \sin(\theta_2))+{0.26} {g_s} {N_f} \cot (\theta_2) \log
   (\sin (\theta_1) \sin(\theta_2))  +{0.78} {g_s} {N_f} \cot
   (\theta_2)+{0.58}
   \csc \left(\theta_2\right)-{1.1} \csc \left(\theta_2\right) \log
   (N)\nonumber\\
& &   \hskip -0.8in  +{0.43} \log (N)
   \cot (\theta_2)-{0.49} \cot (\theta_2)\Biggr) + {\cal O}(r_h^3)\Biggr\},\nonumber\\
& &  \hskip -0.8in \Omega^1_3\Omega^3_{46} + \Omega^1_4\Omega^3_{63} + 
 \Omega^1_6\Omega^3_{34} = {\scriptsize\frac{9\kappa_{\Omega^1_a\Omega^3_{bc}} }{5(r-{r_h})
   {\cal D}^2}\Biggl\{0.9 {g_s}^{9/2} M^2 \sqrt{\frac{1}{N}}
   {N_f}^3 {r_h}^2 \csc ^2\left(\theta _1\right) \log
   ^3(N) \log ^3({r_h})}\nonumber\\
& &  \hskip -0.8in   \Biggl(-{0.53} {g_s} {N_f} \csc \left(\theta
   _2\right) \log (N)-{0.01} {g_s} {N_f} \csc \left(\theta
   _2\right) \log (N) \log
   ({r_h})\nonumber\\
& &  \hskip -0.8in -{0.21} {g_s} {N_f} \log (N) \log
   ({r_h}) \cot
   (\theta_2)-{0.05} {g_s} {N_f} \csc \left(\theta
   _2\right) \log (N) \log (\sin
   (\theta_1) \sin(\theta_2))\nonumber\\
& &  \hskip -0.8in -{5.3} {g_s} {N_f} \log ^2(N) \cot
   (\theta_2)+{0.3} {g_s} {N_f} \log (N) \cot
   (\theta_2)+{0.46} {g_s} {N_f} \csc \left(\theta
   _2\right)\nonumber\\
& &  \hskip -0.8in -{8.8} {g_s} {N_f} \csc \left(\theta
   _2\right) \log
   ({r_h})+{0.28} {g_s} {N_f} \log ({r_h})
   \cot (\theta_2)\nonumber\\
& &  \hskip -0.8in -{0.04} {g_s} {N_f} \csc \left(\theta
   _2\right) \log (\sin
   (\theta_1)\sin
   (\theta_2))+{0.16} {g_s} {N_f} \cot
   (\theta_2) \log (\sin
   (\theta_1)\sin(\theta_2))\nonumber\\
& &  \hskip -0.8in +{0.57} {g_s} {N_f} \cot
   (\theta_2)-{0.69} \csc \left(\theta
   _2\right)-{0.42} \csc \left(\theta _2\right) \log
    N -{0.14}
   \log (N) \cot
   (\theta_2)-{1.5} \cot
   (\theta_2)\Biggr)+{\cal O}(r_h^3)\Biggr\},\nonumber\\ 
& & \hskip -0.8in \Omega^1_4 \Omega^3_{56} + \Omega^1_5 \Omega^3_{64} + 
 \Omega^1_6\Omega^3_{45} = \frac{9\kappa_{\Omega^1_a\Omega^3_{bc}} }{5(r-{r_h})
   {\cal D}^2}\nonumber\\
& & \hskip -0.8in \times\Biggl\{1.4 {g_s}^{9/2} \log N ^3 {(\log r_h})^3 M^2
   \sqrt{\frac{1}{N}} {N_f}^3 {r_h}^2 \csc
   ^2\left(\theta _1\right) \Biggl(\cot (\theta_2) \Biggl[{0.03\times} {g_s}
   \log N ^2 {N_f}+{0.14} {g_s} \log N  \log r_h 
   {N_f}+{g_s} ({0.09} \log N\nonumber\\
& & \hskip -0.8in   -{0.12}) {N_f} \log (\sin (\theta_1)
   \sin (\theta_2))-{0.31} {g_s} \log N 
   {N_f}-{0.28} {g_s}
   \log r_h  {N_f}+{0.07} {g_s} {N_f}-{0.43} \log N +{1.3}\Biggr]\nonumber\\
& & \hskip -0.8in+\csc \left(\theta _2\right)
   (-{0.25} {g_s}
   \log N  \log r_h  {N_f}+{g_s}
   ({0.18}-{0.14} \log N ) {N_f} \log (\sin
   (\theta_1) \sin (\theta_2))+{0.16} {g_s} \log N 
   {N_f}\nonumber\\
& & \hskip -0.8in+{0.27} {g_s}
   \log r_h  {N_f}+{0.25} {g_s} {N_f}+{0.43} \log N -{0.98})\Biggr)\Biggr\}\nonumber\\
& &    
\end{eqnarray}
}
where ${\cal D}\equiv (-3
   {g_s}  N_f  \log ({r_h})-2 {g_s}  N_f  \log (\sin
   (\theta_1) \sin (\theta_2))+2.08 {g_s}
    N_f +12.57)$

\item From \cite{ACMS}, one can show:
\begin{equation}
\label{Omega13Omega7aantisymm}
\Omega^1_{[3}\Omega^7_{a]} = 0, a= 2, 4, 5, 6.
\end{equation}
\item
From \cite{ACMS}, one can further show:
\begin{eqnarray}
\label{Omega3bcOmega7aantisymm}
& & \Omega^3_{[bc}\Omega^7_{a]} = \Omega^3_{[3c}\Omega^7_{a]} = 0,     
\end{eqnarray}
$a, b, c = 2, 4, 5, 6$.
and upon using $|\log r_h|=\kappa_{r_h}N^{1/3}, \kappa_{r_h}=\frac{1}{3(6\pi)^{1/3}\left(g_s N_f\right)^{2/3}\left(g_s M^2\right)^{1/3}}$ \cite{IITR-McGill-bulk-viscosity}, the most dominant term in the MQGP limit (\ref{eq:MQGP}) is given by:
\begin{eqnarray}
\label{Omega3bcOmega73antisymm}
& & \Omega^3_{bc}\Omega^7_{3} e^{bc37}\approx \Omega^3_{24}\Omega^7_3 e^{2437} =  \kappa_{\Omega^3_{24}\Omega^7_{3}}\left(\log N\right)^2 r_h^2|\log r_h|\csc^2\theta_1 \sin(2\theta_2)e^{2437}\nonumber\\
& & = \kappa_{r_h}^2 \kappa_{\Omega^3_{24}\Omega^7_{3}}\left(\log N\right)^2 e^{-2\kappa_{r_h}N^{1/3}}|\log r_h|\csc^2\theta_1 \sin(2\theta_2)e^{2437},
\end{eqnarray}
where $\kappa_{\Omega^3_{24}\Omega^7_{3}}$ is a numerical constant, and it is assumed that $M, N_f = {\cal O}(1)$ (e.g., near the QCD-inspired values of $(g_s, M, N_f) = (0.1, 3, 3)$), corresponding to which $N\equiv {\cal O}(10^2) - {\cal O}(10^3)$. Hence, $\Omega^3_{bc}\Omega^7_{3} e^{bc37}$ is exponentially suppressed in the MQGP limit (\ref{eq:MQGP}) (with an additional $N^{-2/9}$ factor from $|\log r_h|^{-2/3}$ upon writing $e^{7}=\sqrt{g^{\cal M}_{x^{10}x^{10}}}\,dx^{10}$).

\item 
It was shown in \cite{ACMS} that:
\begin{equation}
\label{Omega3bcOmega3dfantisymm}
\Omega^3_{[bc}\Omega^3_{df]}e^{bcdf} = 0.
\end{equation}
\end{itemize}
From (\ref{Omega13Omega3ab}) - (\ref{Omega3bcOmega3dfantisymm}), we see that in the large-$N$ MQGP limit (\ref{eq:MQGP}), the most dominant term in $d\omega_\Phi^2$ term is given by:
\begin{eqnarray}
\label{omegaPhi squared}
& & \left.\omega_\Phi^2\right|_{r\in{\rm IR}} = \left(\Omega^1_{[2}\Omega^3_{46]}\right)e^{2146} + \left(\Omega^1_{[2}\Omega^3_{34]}\right)e^{2134}\nonumber\\
& & \sim \frac{\kappa_{r_h}^{11/6}}{100}g_s^{15/4}N_f^{7/3}M \left(\log N\right)^2\frac{e^{-\frac{\kappa_{r_h}}{2}}N^{1/3}}{N^{23/26}}\csc^2\theta_1\cot\theta_2dt\wedge e^{24}\left(e^3, e^6\right).
\end{eqnarray}
Therefore, from (\ref{dPhi-LO-IR}) and (\ref{omegaPhi squared}),
\begin{eqnarray}
& & \left.d\Omega_+\right|_{r\in{\rm IR}}\approx d\Phi = \kappa_{r_h}^{4/3}\frac{g_s^{1/12}M N_f^{4/3}e^{-\frac{\kappa_{r_h}}{2}N^{1/3}}\log^2N}{N^{5/9}}dt\wedge e^{245}\nonumber\\
& & = \kappa_{r_h}^3g_s^{1/12}M N_f^{4/3}\frac{e^{-\frac{\kappa_{r_h}}{2}N^{1/3}}\log^2N}{\log r_h|^{5/3}}dt\wedge e^{245}
\label{eq:dOmegaIR}
\end{eqnarray}
from where one notes a large-$N$ (exponential) suppression of $d\Omega_+$, which we, therefore, refer to as corresponding to almost half-flat symplectic $SU(3)$ structure. We hence end up with a symplectic nearly half-flat transverse $SU(3)$ structure induced from Contact Structure that in turn is induced from the $G_2$-structure, in the IR. Dropping the $e^{-\kappa_{r_h}N^{1/3}}$-suppressed terms, yields:
\begin{equation}
\label{Ws}
W_\Phi^{SU(3)} = W_2^-.
\end{equation}
\end{proof}

\section{Weak magnetic field from KK reduction via $G_2$/symplectic half-flat $SU(3)$ structure}

Consider the positive three-form component: $\Phi_{\theta_i t\theta_j}$.
In the IR sector: $r\in[r_h,\sqrt{3}a]$ and using the results of \cite{OR4}, one obtains,
\begin{equation}
\label{B-M-theory-1}
\Phi_{\theta_i t\theta_j}\xrightarrow{r\in[r_h,\sqrt{3}a]}
N_f^{\rm eff}(r){\cal F}(\theta_1, \theta_2, r; N_f^{\rm eff}(r)),
\end{equation}
with $N_f^{\rm eff}(r) = N_f\Theta(r - {\cal R}_{D/\overline{D5}}), {\cal R}_{D5/\overline{D5}} = r_h + {\cal O}\left(\frac{1}{N}\right)$.
The corresponding "field strength" would contain $\left(d\Phi\right)_{r\theta_i t\theta_j}\ni\partial_{[r}\Phi_{\theta_i t\theta_j]}\ni \partial_r\left( N_f^{\rm eff}(r){\cal F}(\theta_1,\theta_2; N_f^{\rm eff}(r))\right)$. 
Its eleven-dimensional Hodge dual: 
generates a seven-form field strength. Now, this can be KK reduced to $\mathbb{R}^2(x^{1,2})$:
\begin{equation}
\label{KK-Hodge dPhi}
*_{11}d\Phi\ni\mathbb{F}\wedge\omega_1\wedge\tau_1\wedge\Omega_+,
\end{equation}
with $\mathbb{F}_{x^1x^2} =\partial_{[x^1}A_{x^2]} $, where:
\begin{equation}
\label{B-M-theory-2}
{\cal A}_{x^2}(x^1, x^2, r, \theta_1, \theta_2) \sim x^1 {\cal B}_{\rm KK}(r) \mathbb{X}(\theta_1,\theta_2, r),
\end{equation}
wherein ${\cal B}_{\rm KK}(r)$ is essentially localized around $r=r_h$ in the large-$N$ limit such as the MQGP limit (\ref{eq:MQGP}), and the constant harmonic one-form $\omega^1 = \omega_i dx^i = x^{1, 2, 3}$ where $\omega_i$'s are constants, even though non-normalizable on $M_5=\left(S^1_t \times \mathbb{R}^3\right)\times\mathbb{R}_{\geq0}$, but is normalizable on $M_{11} =M_5\times_w M_6^{\rm non-Kaehler}$ because of the vanishing of $M_{\rm eff}(r)$ and $N_f^{\rm eff}(r)$ in the UV ($r>\sqrt{3}a$)
 (Normalisability of $\omega_1$ follows from $\|\omega_1\|^2<\infty$.) Strictly speaking, as even in the UV it is assumed that $r<\left(4\pi g_s N\right)^{1/4}$, the non-normalizability in the UV is a non-issue even on $M_5$. We had assumed that the $r$-dependent magnetic field would damp out in the UV. This is now happily captured by $\mathcal{B}_{\rm KK}(r)\sim N_f^{\rm eff}\ '(r)$. 

In this section, we will obtain the Euler-Lagrange (EL) equation of motion (EOM) satisfied by ${\cal B}_{\rm KK}(r)$ and show that it is a non-linear power of the solution of the EL EOM satisfied by an IR-localized radial ${\cal B}_{D6}(r)$ turned on, on the world-volume of the corresponding type IIA flavor $D6$-branes such that ${\bf B}(r=r_h) e_{x^3} = B e_{x^3}, B$ being a weak external magnetic field. Hence, the type IIA $D6$-brane-valued IR-localized radial ${\cal B}_{D6}(r)$ and therefore ${\cal B}(r=r_h)$, are mapped non-linearly to the bulk KK field $\mathcal{B}(r)$ and its IR boundary value ${\cal B}_{\rm KK}(r)$.

Further, even though $d\tau_1\neq0$, but not only does it vanish in the UV for the same reason as the $M_{11}$-normalizability of $\omega_1$, but is large-$N$ suppressed in the IR. The reason is that one can show:
\begin{eqnarray}
\label{dtau1-1}
& & d\tau_1(r\in{\rm IR})\approx \Omega^3_{24}\Omega^1_a e^{a1} - \Omega^3_{25}\Omega^7_a e^{7a}\approx \Omega^3_{24}\Omega^1_4e^{41} - \Omega^3_{25}\Omega^7_6e^{76}\nonumber\\
& & \sim\frac{\left(g_s^{7/4}M_{\rm eff}N_f^{\rm eff}\log N\right)^2}{N^{3/4}}\sqrt{r_h}\sqrt{r - r_h}\left(\log  (r_h) \right)^{1/3} e^4\wedge dx^0 \nonumber\\ 
& & - 10^2\frac{r_h^3\left(\log  (r_h) \right)^{4}\left(g_s^{7/4}M_{\rm eff}N_f^{\rm eff}\log N\right)^2\cos\theta_1}{\left(\sin\theta_{1}\sin\theta_{2}\right)^2\sqrt{N}}dx^{10}\wedge e^6,\nonumber\\
& & 
\end{eqnarray}
which is large-$N$ suppressed.

The 11-dimensional spacetime has coordinates
\[
  x^A = \bigl(
    \underbrace{x^1,x^2,x^3}_{\mathbb{R}^3},\;
    r,\;
    \underbrace{t, \theta_1, \theta_2, \phi_1, \phi_2, \psi, x^{10}}_{M_7}
  \bigr).
\]
The metric is block-diagonal with $g_{rM}=0$ for all $M\neq r$:
\begin{equation}
  ds^2_{11}
  = g_{\mu\nu}(r,y)\,dx^\mu dx^\nu
  + \grr(r,y)\,dr^2
  + g_{mn}(r,y)\,dy^m dy^n,
\label{eq:metric}
\end{equation}
with all blocks depending on $(r,y^m)$.  The 11d volume element factorises:
\[
  \sqrt{\gone}
  = \sqrt{\gthree(r,y)}\cdot\sqrt{\grr(r,y)}\cdot\sqrt{\gseven(r,y)},
\]
and we define the orthonormal radial 1-form and the $\mathbb{R}^3$ volume form:
\[
  \er \equiv \sqrt{\grr}\,dr,
  \qquad
  \vthree \equiv \sqrt{\gthree}\,dx^1\wedge dx^2\wedge dx^3.
\]

\[
  A_1 = x^1\,{\cal B}_{\rm KK}(r),
  \qquad
  F_2 = dA_1,
  \qquad
  \omega_1 = \omega_\mu\,dx^\mu
  \quad(\omega_\mu = \mathrm{const},\;\mu=1,2,3).
\]
Therefore,
\[
  A_{x^2} = x^1\,{\cal B}_{\rm KK}(r)
\]
\[
  F_2 = dA_1
  = \Nf\,dx^1\wedge dx^2
  + x^1\partial_r\Nf\,dr\wedge dx^2.
\]

\noindent
The non-zero independent components of $F_2$ are:
\begin{align}
  (F_2)_{x^1\,x^2}       &= x^1\partial_r\Nf,
    \label{eq:F_xr}\\
  (F_2)_{r\,x^\mu}   &= \Nf\,\delta^{x^2}_\mu.
    \label{eq:F_xx}
\end{align}

Now,
\begin{align*}
  d^{(7)}\vp
  &= d\bigl(\sigma\wedge J\bigr) + d\Omega_+ \\
  &= d\sigma\wedge J
     - \sigma\wedge\underbrace{dJ}_{=\,0}
     + \underbrace{d\Omega_+}_{=\,0} \\
  &= d\sigma\wedge J
\end{align*}

\noindent
The torsion-class expansion $d^{(7)}\vp = \ta_0\ps+3\ta_1\wedge\vp+\str\ta_3$
therefore becomes
\[
  d\sigma\wedge J = \ta_0\ps + 3\ta_1\wedge\vp + \str\ta_3.
\]
Note $\tau_2$ enters only via
$d\ps = 4\ta_1\wedge\ps + \ta_2\wedge\vp$,
which is sourced by $d\Omega_-\neq 0$.

Since $g_{mn}=g_{mn}(r,y)$ the 3-form $\vp$ also depends on $r$, so
\[
  \partial_r\vp = \partial_r\sigma\wedge J + \sigma\wedge\partial_r J
               + \partial_r\Omega_+
  \;\in\;\Omega^3(M_7)
\]
carries non-trivial radial flow.
This piece will contribute to Term~A of $\sEl(d\vp)$ but will be killed by $\Pint$ and hence does not enter the torsion classes at fixed $r$.

We use the standard $\frac{1}{p!}$ convention for the squared norm of a
$p$-form $\omega$ on $M_7$:
\[
  |\omega|^2_7 \;\equiv\; \frac{1}{p!}\,\omega_{a_1\cdots a_p}\,\omega^{a_1\cdots a_p}
  \qquad\Longrightarrow\qquad
  \omega_{a_1\cdots a_p}\,\omega^{a_1\cdots a_p} = p!\,|\omega|^2_7.
\]
For the $G_2$ forms (Bryant--Joyce normalisation):
\begin{align}
  \vp_{abc}\vp^{abc} &= 3!\,|\vp|^2_7 = 6\times 7 = 42,
  \label{eq:norm_phi}\\
  \ps_{abcd}\ps^{abcd} &= 4!\,|\ps|^2_7 = 24\times 7 = 168.
  \label{eq:norm_psi}
\end{align}
The equality $|\ps|^2_7 = |\vp|^2_7 = 7$ follows because
$\ps=\str\vp$ and the 7d Hodge star preserves norms:
$|\str\omega|^2_7 = |\omega|^2_7$.

Apply $\str$ to both sides of $d^{(7)}\vp = \ta_0\ps+3\ta_1\wedge\vp+\str\ta_3$,
using the $G_2$ Hodge identities
$\str\ps = \vp$, $\str(\ta_1\wedge\vp) = -\iot{\ta_1^\sharp}\ps$,
$\str\str\ta_3 = +\ta_3$ (Riemannian $M_7$, 3-form):
\[
  \str(d^{(7)}\vp)
  = \ta_0\,\underbrace{\str\ps}_{=\,\vp}
    \;-\; 3\,\underbrace{\str(\ta_1\wedge\vp)}_{=\,\iot{\ta_1^\sharp}\ps}
    \;+\; \underbrace{\str\str\ta_3}_{=\,\ta_3}
  \;=\; \ta_0\vp \;-\; 3\,\iot{\ta_1^\sharp}\ps \;+\; \ta_3.
\]
\noindent
The three summands lie in $\mathbf{1}$, $\mathbf{7}$, $\mathbf{27}$ respectively,
so $\str d^{(7)}\vp\in\Omega^3(M_7)$ is the $G_2$-irrep decomposition.

\begin{lemma}[$\ps_{ma}{}^{bcd}$ contraction identity]
\label{lem:psicontract}
For any Riemannian seven-manifold with $G_2$ holonomy, \cite{Karigiannis}
\begin{equation}
\ps_{mbcd}\,\ps_a{}^{bcd} = 24\,g_{ma}.
\label{eq:psipsi}
\end{equation}
\end{lemma}
\begin{proof}
By $G_2$-invariance (the group $G_2$ preserves $\ps$ and $g$), the
symmetric tensor $X_{ma}\equiv\ps_{mbcd}\ps_a{}^{bcd}$ must be
proportional to $g_{ma}$:
\[
  \ps_{mbcd}\,\ps_a{}^{bcd} = \mu\,g_{ma}.
\]
Contract with $g^{ma}$ to find $\mu$:
\begin{align*}
  g^{ma}\,\ps_{mbcd}\,\ps_a{}^{bcd}
  &= \ps^a{}_{bcd}\,\ps_a{}^{bcd}
   = \ps_{abcd}\,\ps^{abcd}
   = 4!\,|\ps|^2_7
   = 24\times 7 = 168,\\
  &= \mu\,g^{ma}g_{ma} = \mu\times 7.
\end{align*}
Hence $\mu = 168/7 = 24$, so
\[
  \ps_{mbcd}\,\ps_a{}^{bcd} = 24\,g_{ma}.
\]
\end{proof}

First, let us show that: 
\begin{equation}
\label{phidotpsi0}
\vp_{bcd}\,\ps_a{}^{bcd} = 0.
\end{equation}
 As $\psi_{abcd} = \frac{\sqrt{g_7}}{3!}\epsilon_{abcdefg}G^{ee'}G^{ff'}G^{gg'}\varphi_{e'f'g'}$, implying $\varphi_{bcd}\psi_a^{\ \ bcd} = \frac{1}{6}\sqrt{g_7}\epsilon_{abcdefg}\varphi^{bcd}\varphi^{efg}$. Now, $\epsilon_{abcdefg}$ is anti-symmetric under $(bcd)\leftrightarrow(efg)$ whereas
$\varphi^{bcd}\varphi^{efg}$ is symmetric. Hence the claim. Alternatively, group theoretically, as $\varphi\in\Lambda^3_1; \psi\in\Lambda^4_1, \varphi_{bcd}\psi_a^{\ \ bcd}\in\Lambda^1_7$, transforms like $\mathbf{1}\otimes\mathbf{1}\otimes\mathbf{7} = \mathbf{7}$. Since 

Second, let us prove that 
\begin{lemma}
\begin{equation}
\label{itau1psipsi0}
(\iot{\ta_1^\sharp}\ps)_{bcd}\,\ps_a{}^{bcd}
= 24\,(\ta_1)_a.
\end{equation} 
\end{lemma}
\begin{proof}
Writing $(\iot{\ta_1^\sharp}\ps)_{bcd} = (\ta_1)^m\ps_{mbcd}$:
\[
  (\iot{\ta_1^\sharp}\ps)_{bcd}\,\ps_a{}^{bcd}
  = (\ta_1)^m\,\ps_{mbcd}\,\ps_a{}^{bcd}
  = (\ta_1)^m\cdot 24\,g_{ma}
  = 24\,(\ta_1)_a.
\]
\end{proof}

Third, let us now prove: 
\begin{lemma}
\begin{equation}
\label{tau3dotpsi0}
(\ta_3)_{bcd}\,\ps_a{}^{bcd} = 0.
\end{equation}
\end{lemma}
\begin{proof}
\begin{proof}
Consider the $G_2$-equivariant map
$W:\Omega^3(M_7)\to\Omega^6(M_7),\;\alpha\mapsto\alpha\wedge\vp$.
The irreducible decomposition
$\Omega^3(M_7)=\Lambda^3_{\mathbf{1}}\oplus\Lambda^3_{\mathbf{7}}
\oplus\Lambda^3_{\mathbf{27}}$ has $\ta_3\in\Lambda^3_{\mathbf{27}}$.
Since $\Omega^6(M_7)\cong\Omega^1(M_7)\cong\mathbf{7}$ via $\str$,
equivariance requires $W\big|_{\mathbf{27}}:\mathbf{27}\to\mathbf{7}$.
By Schur's lemma,
$\HomG(\mathbf{27},\mathbf{7})=\{f:\mathbf{27}\to\mathbf{7}
\mid f(g\cdot v)=g\cdot f(v)\ \forall\,g\in G_2\}=0$ (see Appendix \ref{HomG27270}),
since $\mathbf{27}$ and $\mathbf{7}$ are non-isomorphic irreducible
$G_2$-representations.  Hence $(\ta_3)_{bcd}\ps_a{}^{bcd}$, being
a $G_2$-equivariant map from $\mathbf{27}$ to $\mathbf{7}$,
vanishes.
\end{proof}

By contracting both sides of $\str(d^{(7)}\vp)_{bcd} = \ta_0\vp_{bcd} - 3(\iot{\ta_1^\sharp}\ps)_{bcd} + (\ta_3)_{bcd}$
with $\ps_a{}^{bcd}$, and usingod (\ref{phidotpsi0}) - (\ref{tau3dotpsi0}),  
one obtains:
\[
  (\str d^{(7)}\vp)_{bcd}\,\ps_a{}^{bcd}
  = \ta_0\underbrace{\vp_{bcd}\ps_a{}^{bcd}}_{=\;0}
    \;-\; 3\underbrace{(\iot{\ta_1^\sharp}\ps)_{bcd}\ps_a{}^{bcd}}_{=\;24(\ta_1)_a}
    \;+\; \underbrace{(\ta_3)_{bcd}\ps_a{}^{bcd}}_{=\;0}
  = -72\,(\ta_1)_a,
\]
\end{proof}
or
\begin{proposition}
\begin{equation}
\label{tau1*7dphi}
  (\ta_1)_a(r,y)
  = -\frac{1}{72}\,\bigl(\str d^{(7)}\vp\bigr)_{bcd}(r,y)\;\ps_a{}^{bcd}(r,y),
\end{equation}
\end{proposition}
where $\ps_a{}^{bcd} = g^{bb'}g^{cc'}g^{dd'}\ps_{ab'c'd'}$.

Inserting $d^{(7)}\vp = d\sigma\wedge J$:

\[
  (\ta_1)_a(r,y)
  = -\frac{1}{72}\,\bigl(\str(d\sigma\wedge J)\bigr)_{bcd}(r,y)\;
    \ps_a{}^{bcd}(r,y).
\]

Using $d^{(7)}\vp = \ta_0\ps+3\ta_1\wedge\vp+\str\ta_3$ and applying $\str$:
\[
  \str(d^{(7)}\vp)
  = \ta_0\vp - 3\,\iot{\ta_1^\sharp}\ps + \ta_3.
\]

For the block-diagonal metric with $g_{rM}=0, g_{\mu m}=0, M=\mu, r, m$, the 11d Hodge dual satisfies:
\begin{align}
  \sEl(dr\wedge\beta)
  &= \frac{1}{\sqrt{\grr}}\,\vthree\wedge\str\beta,
  \qquad \beta\in\Omega^p(M_7),
  \label{eq:hodge_r}\\
  \sEl(\gamma)
  &= \sqrt{\grr}\,\vthree\wedge dr\wedge\str\gamma,
  \qquad \gamma\in\Omega^p(M_7).
  \label{eq:hodge_int}
\end{align}
\noindent

As
\[
  d\vp
  = \underbrace{\frac{1}{3!}(\partial_r\vp_{mnp})\,dr\wedge dy^m\wedge dy^n\wedge dy^p}_{
      \text{radial part }\in\,\Omega^1(r)\otimes\Omega^3(M_7)}
  + \underbrace{\frac{1}{4!}(d^{(7)}\vp)_{qmnp}\,dy^q\wedge dy^m\wedge dy^n\wedge dy^p}_{
      \text{internal part }= d^{(7)}\vp = d\sigma\wedge J\,\in\,\Omega^4(M_7)},
\]
therefore,
\begin{align*}
  \sEl(d\vp)
  &= \underbrace{
      \frac{1}{\sqrt{\grr}}\cdot\frac{1}{3!}
      (\partial_r\vp_{mnp})\;\vthree\wedge\bigl(\str dy^m\wedge dy^n\wedge dy^p\bigr)
    }_{\displaystyle\text{Term A: no }dr\text{ leg --- encodes radial flow }\partial_r\vp}\\
  &\quad
    +\underbrace{
      \sqrt{\grr}\cdot\frac{1}{4!}(d\sigma\wedge J)_{qmnp}\;
      \vthree\wedge dr\wedge\bigl(\str dy^q\wedge dy^m\wedge dy^n\wedge dy^p\bigr)
    }_{\displaystyle\text{Term B: has }dr\text{ leg --- encodes }G_2\text{ torsion via }d\sigma\wedge J}
\end{align*}

\noindent
Term~A uses rule~\eqref{eq:hodge_r} with $\beta=\tfrac{1}{3!}\partial_r\vp_{mnp}\,dy^m\wedge dy^n\wedge dy^p$;
Term~B uses rule~\eqref{eq:hodge_int} with $\gamma=d\sigma\wedge J$.

Defining $\Pint = \iot{\partial_r/\sqrt{\grr}}\,\iot{e_3}\iot{e_2}\iot{e_1}$
\[
  \Pint\sEl(d\vp)
  = \str(d\sigma\wedge J)
  = \ta_0\vp - 3\,\iot{\ta_1^\sharp}\ps + \ta_3
  \;\in\;\Omega^3(M_7).
\]

With $d^{(7)}\vp = d\sigma\wedge J$ and using
$\vp = \sigma\wedge J+\Omega_+$, $\ps = \tfrac{1}{2}J\wedge J-\sigma\wedge\Omega_-$:

\begin{align*}
  \str(d\sigma\wedge J)
  &= \ta_0\vp - 3\,\iot{\ta_1^\sharp}\ps + \ta_3 \\
  &= \ta_0(\sigma\wedge J) + \ta_0\Omega_+
     - 3\,\iot{\ta_1^\sharp}\bigl(\tfrac{1}{2}J\wedge J - \sigma\wedge\Omega_-\bigr)
     + \ta_3.
\end{align*}
Therefore Term~B is
\[
  \text{B}
  = \sqrt{\grr}\;\vthree\wedge dr\wedge
    \bigl[
      \ta_0\Omega_+
      + \ta_0(\sigma\wedge J)
      - 3\,\iot{\ta_1^\sharp}(\tfrac{1}{2}J\wedge J)
      + 3\,\iot{\ta_1^\sharp}(\sigma\wedge\Omega_-)
      + \ta_3
    \bigr].
\]
The $\Omega_+$-containing sub-term is $\sqrt{\grr}\;\vthree\wedge dr\wedge\ta_0\Omega_+$;
no other term in Term~B contains $\Omega_+$.

The external factor $\vthree\wedge dr$ carries legs $\{x^1,x^2,x^3,r\}$.
\[
  \sqrt{\grr}\vthree\wedge dr\propto
  \omega_1\wedge F\wedge e^r.
\]

Now, $R^{(11)}\supset   R^{(7)}\supset \,|\tau_1|^2\subset|*_{11}d^{(11)}\varphi|^2=| \sqrt{\grr}\vthree\wedge dr\wedge*_7d^{(7)}\varphi|^2\propto  |F_2\wedge\omega_1|^2|d^{(7)}\phi|^2\supset|F_2\wedge\omega_1|^2|\tau_1\wedge\varphi|^2\supset |F_2\wedge\omega_1|^2 |\tau_1\wedge\Omega_+|^2\equiv|{\cal M}|^2$,
where $\mathcal{M}=\mathcal{M}_{\mathrm{ext}}\wedge\mathcal{M}_{\mathrm{int}}$
with $\mathcal{M}_{\mathrm{ext}}=F_2\wedge\omega_1\in\Omega^3$ and
$\mathcal{M}_{\mathrm{int}}=\ta_1\wedge\Omega_+\in\Omega^4(M_7)$. From (\ref{tau1*7dphi}), and \cite{Karigiannis}
\begin{equation}
\label{Kiragiannis-psi.psi}
\begin{aligned}
\psi_{ijkl}\psi_{abcd}g^{ld} = \;&
  -\varphi_{ajk}\varphi_{ibc}
  - \varphi_{iak}\varphi_{jbc}
  - \varphi_{ija}\varphi_{kbc} \\
&+ g_{ia}g_{jb}g_{kc}
  + g_{ib}g_{jc}g_{ka}
  + g_{ic}g_{ja}g_{kb} \\
&- g_{ia}g_{jc}g_{kb}
  - g_{ib}g_{ja}g_{kc}
  - g_{ic}g_{jb}g_{ka} \\
&- g_{ia}\psi_{jkbc}
  - g_{ja}\psi_{kibc}
  - g_{ka}\psi_{ijbc} \\
&+ g_{ab}\psi_{ijkc}
  - g_{ac}\psi_{ijkb},
\end{aligned}
\end{equation}
one notes, as implied by the discussion immediately above (\ref{Kiragiannis-psi.psi}),
\begin{equation*}
|\tau_1|^2\supset |*_7d\varphi|^2.
\end{equation*}

Using the block-diagonal metric:
\begin{proposition}
\label{prop:normfact}
With the block-diagonal metric \eqref{eq:metric},
\begin{equation}
|\mathcal{M}|^2_{11d}
= |F_2\wedge\omega_1|^2_{g_{\mu\nu},\grr}
\cdot
|\ta_1\wedge\Om_+|^2_{g_{mn}}.
\label{eq:normfact}
\end{equation}
The external factor is
\begin{equation}
|F_2\wedge\omega_1|^2
=
\frac{|\omega_1|^2_{\gthree}}{\grr}(x^1)^2
\!\Bigl[(\partial_r\Nf)^2\Bigr]
+ (\Nf)^2\mathcal{N}_\omega,
\label{eq:extnorm}
\end{equation}
where $\mathcal{N}_\omega = \bigl(|\omega_1|^2_{\gthree}
-\omega_\mu\omega_\nu g^{\mu x^1}g^{\nu x^1}g_{x^1x^1}\bigr)$.
The internal factor is
\begin{equation}
\Ical \equiv |\ta_1\wedge\Om_+|^2_{g_{mn}}
= |\ta_1|^2_{g_{mn}}|\Om_+|^2_{g_{mn}}
- \tfrac{1}{2}(\ta_1)^m(\ta_1)^{m'}(\Om_+)_{mnp}(\Om_+)_{m'}{}^{np}.
\label{eq:I7def}
\end{equation}
\end{proposition}

\begin{proof}
\subsection*{External norm $|F_2\wedge\omega_1|^2$}

$F_2\wedge\omega_1$ is a 3-form with legs in
$\{x^\mu,r,\theta_1,\theta_2\}\times\{x^\nu\}$.
Its independent non-zero components (using $F_2$  and
$\omega_1=\omega_\mu dx^\mu$) are:
\begin{align}
  (F_2\wedge\omega_1)_{x^1,\,r,\,x^\nu}
    &= (F_2)_{x^1 r}\,\omega_\nu
     = x^1\partial_r\Nf\;\omega_\nu,
  \label{ext r}\\[2pt]
  (F_2\wedge\omega_1)_{x^1,\,x^\mu,\,x^\nu}
    &= \Nf\,\bigl(\delta^1_\mu\,\omega_\nu - \delta^1_\nu\,\omega_\mu\bigr).
  \label{eq:ext_xx}
\end{align}
Computing the squared norm
$|F_2\wedge\omega_1|^2
= \frac{1}{3!}(F_2\wedge\omega_1)_{ABC}(F_2\wedge\omega_1)^{ABC}$
by contracting each component with the appropriate inverse metric factors:
\begin{itemize}
  \item From (\ref{ext r}):
    $\displaystyle
      \frac{|\omega_1|^2_{g_{3d}}}{\grr}\,\left(g^{x^1x^1}\right)(x^1)^2(\partial_r\Nf)^2$.
  \item From \eqref{eq:ext_xx}:
    $\displaystyle
      (\Nf)^2\,\bigl(|\omega_1|^2_{\gthree}
-\omega_\mu\omega_\nu g^{\mu x^1}g^{\nu x^1}g_{x^1x^1}\bigr)
      \equiv (\Nf)^2\,\mathcal{N}_\omega$.
\end{itemize}
Collecting:
\[
  |F_2\wedge\omega_1|^2
  =
  \frac{|\omega_1|^2_{g_{3d}}}{\grr}\,(x^1)^2
  \Bigl[
    (\partial_r\Nf)^2  \Bigr]
  + (\Nf)^2\,\mathcal{N}_\omega,
\]
where
\[
  \mathcal{N}_\omega
  \;\equiv\;
  \left(g^{x^1x^1}\right)^2\Bigl(
    |\omega_1|^2_{g_{3d}}
    - \omega_\mu\omega_\nu\,g^{\mu x^1}g^{\nu x^1}g_{x^1x^1}
  \Bigr).
\]

\subsection*{Internal norm $|\ta_1\wedge\Omega_+|^2_{g_{mn}}$}

The 4-form components are
\[
  (\ta_1\wedge\Omega_+)_{mnpq}
  = 4\,(\ta_1)_{[m}(\Omega_+)_{npq]}
  = (\ta_1)_m(\Omega_+)_{npq}
    - (\ta_1)_n(\Omega_+)_{mpq}
    + (\ta_1)_p(\Omega_+)_{mnq}
    - (\ta_1)_q(\Omega_+)_{mnp}.
\]

\noindent
Using the general wedge-product norm identity for a 1-form $\alpha$
and a $q$-form $\beta$:
\[
  |\alpha\wedge\beta|^2 = |\alpha|^2|\beta|^2 - |\iot{\alpha^\sharp}\beta|^2,
\]
with $\alpha=\ta_1\in\Omega^1(M_7)$ and $\beta=\Omega_+\in\Omega^3(M_7)$:

\[
  \Ical
  \;\equiv\;
  |\ta_1\wedge\Omega_+|^2_{g_{mn}}
  = |\ta_1|^2_{g_{mn}}\,|\Omega_+|^2_{g_{mn}}
  - \bigl|\iot{\ta_1^\sharp}\Omega_+\bigr|^2_{g_{mn}},
\]
where
\[
  |\ta_1|^2_{g_{mn}} = g^{mn}(\ta_1)_m(\ta_1)_n,
  \qquad
  |\Omega_+|^2_{g_{mn}} = \frac{1}{3!}(\Omega_+)_{mnp}(\Omega_+)^{mnp},
\]
\[
  \bigl|\iot{\ta_1^\sharp}\Omega_+\bigr|^2_{g_{mn}}
  = \frac{1}{2!}\,(\ta_1)^m(\ta_1)^{m'}\,
    (\Omega_+)_{mnp}\,(\Omega_+)_{m'}{}^{np},
  \qquad
  (\ta_1)^m = g^{mn}(\ta_1)_n.
\]

\noindent
$\Ical$ vanishes if and only if $\ta_1$ lies in the kernel of
$\iot{(\cdot)}\Omega_+$, which is non-generic.

\subsection*{Full 11d norm}

Combining the two factors:

\begin{align}
  |\mathcal{M}|^2_{11d}
  &=
  \Bigl[
    \frac{|\omega_1|^2_{g_{3d}}}{\grr}\,(x^1)^2
    \Bigl(
      (\partial_r\Nf)^2    \Bigr)
    + (\Nf)^2\,\mathcal{N}_\omega
  \Bigr]
  \nonumber\\[6pt]
  &\quad\times
  \Bigl[
    |\ta_1|^2_{g_{mn}}\,|\Omega_+|^2_{g_{mn}}
    - \tfrac{1}{2}(\ta_1)^m(\ta_1)^{m'}\,(\Omega_+)_{mnp}(\Omega_+)_{m'}{}^{np}
  \Bigr].
  \label{eq:fullnorm}
\end{align}


The ${\cal B}_{\rm KK}(r)$-dependence enters because the KK mode $\mathcal{M}$ identifies
$|\ta_1\wedge\Omega_+|^2_{g_{mn}} = \Ical$ as the internal coupling;
the full $A_1$-dependent action is then
\[
  S_{EH}^{(A_1)}
  \;\propto\;
  \int\sqrt{\gone}\,\Ical\,|\mathcal{M}_{\mathrm{ext}}|^2\,d^{11}x,
\]
with $|\mathcal{M}_{\mathrm{ext}}|^2 = |F_2\wedge\omega_1|^2$.

\subsection*{Full expression}

\begin{align}
  S_{EH}^{(A_1)}
  &= \frac{-12}{2\kap^2}
     \int_{r_h}^{r^{\rm UV}}dr
     \int_{\mathbb{R}^3}d^3x
     \int_{M_7}d^7y\;
        \sqrt{\gthree}\,\sqrt{\grr}\,\sqrt{\gseven}\,
     \Ical(r,y)
  \nonumber\\[6pt]
  &\quad\times
  \Biggl\{
    \underbrace{
      \frac{|\omega_1|^2_{g_{3d}}}{\grr}\,(x^1)^2
      \Bigl[
        (\partial_r\Nf)^2\Bigr]}_{\displaystyle\text{(II) kinetic terms in }(r)}
    \;+\;
    \underbrace{
      (\Nf)^2\,\mathcal{N}_\omega
    }_{\displaystyle\text{(III) mass term}}
  \Biggr\},
  \label{eq:SEH_correct}
\end{align}
where
\begin{equation}
  \Ical(r,y)= |\ta_1|^2_{g_{mn}}\,|\Omega_+|^2_{g_{mn}}
     - \tfrac{1}{2}(\ta_1)^m(\ta_1)^{m'}\,
       (\Omega_+)_{mnp}(\Omega_+)_{m'}{}^{np}.
  \label{eq:I7}
\end{equation}
All quantities in \eqref{eq:SEH_correct}--\eqref{eq:I7}
depend on $(r,y^m)$.
\end{proof}

\begin{proposition}[Reduction to a Bessel-type ODE]
\label{prop:ODE}
The Euler--Lagrange equation for $\mathcal{B}_{\rm KK}(r)$ derived from
\eqref{eq:SEH_correct} via the effective Lagrangian
\begin{align}
\Leff
&\propto -\frac{1}{2\kap^2}\sqrt{\gthree}\sqrt{\grr}|\omega_1|^2_{\gthree}V_{x^1}
\Bigl[\mathcal{K}_{rr}(r)(\partial_r\Nf)^2
+ \mathcal{K}_m(r)\frac{(\Nf)^2}{|\omega_1|^2_{\gthree}}\Bigr]
\label{eq:Leff}
\end{align}
is
\begin{equation}
\partial_r\!\Bigl(\sqrt{\gthree}\sqrt{\grr}|\omega_1|^2\mathcal{K}_{rr}(r)
\partial_r\Nf\Bigr)
- \sqrt{\gthree}\sqrt{\grr}\,\mathcal{K}_m(r)\,\Nf = 0.
\label{eq:EOM}
\end{equation}
This reduces to the Bessel-type ODE with a regular singular point at $r=\rh$:
\begin{equation}
\anu(r-\rh)^2\mathcal{B}_{\rm KK}'(r)
+ \ade(r-\rh)^3\mathcal{B}_{\rm KK}''(r)
+ \ath\,\mathcal{B}_{\rm KK}(r)(r-\rh)^2 = 0,
\label{eq:NeffEOM}
\end{equation}
where the real constants $\anu,\ade,\ath$ encode the integrated
metric and torsion data of the background. Horizon normalizability yields:
\begin{equation*}
{\cal B}_{\rm KK}(r) = B(r=r_h) \frac{J_2\!\left(2\sqrt{\frac{\ath}{\ade}(r-\rh)}\right)}{r - r_h},
\quad r > \rh,
\end{equation*}
where one is free to choose $B(r=r_h)$.
\end{proposition}
\begin{proof}
\section*{Euler--Lagrange Equation for ${\cal B}_{\rm KK}(r)$}

\subsection*{$M_7$-integrated coefficient functions}

Define the $r$-dependent coefficient functions obtained by integrating $\Ical$
over $M_7$ with appropriate metric weights:
\begin{align}
  \mathcal{K}_{rr}(r)
  &\equiv
    \int_{M_7}d^7y\;\sqrt{\gseven}\;
    \frac{\Ical(r,y)}{\grr(r,y)},
  \label{eq:Krr}\\[4pt]
  \mathcal{K}_m(r)
  &\equiv
    \int_{M_7}d^7y\;\sqrt{\gseven}\;\Ical(r,y)\,\mathcal{N}_\omega(r,y).
  \label{eq:Km}
\end{align}

\subsection*{Effective Lagrangian after $M_7$ integration}

Integrating \eqref{eq:SEH_correct} over $M_7$ and denoting
$V_{x^1}\equiv\int_{\mathbb{R}^3}d^3x\,(x^1)^2$ (a regulated IR volume):
\begin{align}
  \Leff(r,\theta_1,\theta_2)
  &\propto \frac{-1}{2\kap^2}\,
     \sqrt{\gthree}\,\sqrt{\grr}\,
     |\omega_1|^2_{g_{3d}}\,V_{x^1}
  \nonumber\\
  &\quad\times
  \Bigl[
    \mathcal{K}_{rr}(r)\,(\partial_r\Nf)^2
        + \mathcal{K}_m(r)\,\frac{(\Nf)^2}{|\omega_1|^2_{g_{3d}}}
  \Bigr].
  \label{eq:Leff-proof}
\end{align}
The Euler-Lagrange EOM for ${\cal B}_{\rm KK}(r)$ is:
\begin{align}
  &\partial_r\!\Bigl(
    \sqrt{\gthree}\,\sqrt{\grr}\,
    |\omega_1|^2_{g_{3d}}\,V_{x^1}\,
    \mathcal{K}_{rr}(r)\,\partial_r\Nf
  \Bigr)
  \nonumber\\[5pt]
   -\;&\sqrt{\gthree}\,\sqrt{\grr}\,V_{x^1}\,
    \mathcal{K}_m(r)\,\Nf
  \;=\; 0.
  \label{eq:EOM-proof}
\end{align}

\begin{eqnarray}
\label{LagB-i}
& &   \mathcal{I}_7(r, \left\{y^m\right\})
  = |\tau_1|^2\,|\Omega_+|^2
          - \tfrac{1}{2}(\mathcal{I}\tau)^{m}(\mathcal{I}\tau)^{m'}\Omega_{+, mnp}\Omega_{+, m'}^{\ \ \ \ np},\nonumber\\
& & = \left(\Omega^3_{24}\right)^2\Biggl[|\Omega_+|^2 - \frac{G_{\cal M}^{tt}}{2}\left(\Omega_+\right)_{tnp}\left(\Omega_+\right)_t^{\ \ np}\Biggr],
\end{eqnarray}
\noindent where $\tau_1 \sim e^1\,\Omega^3_{24}$ and
$\Omega_+ = \Phi - \sigma\wedge \omega_\Phi$.


\[
  \sigma = \alpha_1\,e + \alpha_3\,e^3 + \alpha_7\,e^7
\]
\[
  \Rightarrow\quad \omega = d\sigma
  = \alpha_1\,\Omega_a^1\,e^{a1}
        + \alpha_3\,\Omega^3_{ab}\, e^{ab}
    + \alpha_7\,\Omega_a^7\,e^{a7},
\]
\noindent where $a, b = 2, 3, 4, 5, 6$.

\[
  \Rightarrow\quad \sigma\wedge \,\omega
  =  \alpha_3^2\,\Omega^3_{ab}\,e^{3 a b}
      +\; \alpha_1\alpha_3\!\left(-\Omega^3_{ab}\,e^{1 a b}
    + \Omega^1_a\,e^{3 1 1}\right)
\]
\[
  +\; \alpha_1\alpha_7\!\left(\Omega^7_a e^{1a7}
    + \Omega^1_a\,e^{7a1}\right)
  +\; \alpha_3\alpha_7\!\left(\Omega^7_{a}\,e^{3a7}
    + \Omega^3_{ab}\,e^{7 a b}\right).
\]

\[
  \Phi = e^{127} + e^{347} + e^{567} + e^{135} - e^{146} - e^{236} - e^{245}.
\]

$\Omega_6^1,\;\Omega_6^3,\;\Omega_{24}^3$ most dominant; $|\Omega^3_{24}|>|\Omega_6^1|, |\omega_6^3|$
\[
  \Rightarrow\quad
  \Phi_{tmn} = \bigl(e^{127} + e^{135} - e^{146}\bigr)_{tmn}
\]
and,
\[
    \bigl(\Omega_{+}\bigr)_{tmn}
    \approx -\bigl(e^{124}\bigr)_{tmn}\,\alpha_1\alpha_3\,\Omega_{24}^3.
\]

\bigskip

\noindent As $e^3 = \sqrt{G_{x^{10}x^{10}}^M}\,dx^{10}$,
\[
  \sqrt{G_{x^{10}x^{10}}^M}
  \;\sim\; \frac{1}{{|\log r_h|}^{2/3}}
  \;\sim\; \frac{1}{N^{2/9}}
\]

\noindent i.e.\ $-\Omega_{24}^3\,e^{724}$ is sub-dominant as compared to
$-\Omega_{24}^3\,e^{324}$

\[
  \Rightarrow\quad
  \left(\Omega_{+}\right)_{pmn}
  \approx +\,\Omega_{24}^3\bigl(e^{234}\bigr)_{pmn},
  \qquad
  \left|\Omega_{24}^3\right|
  \sim N^{\frac{1}{12}}\, (\log r_h)^2
\]

\begin{eqnarray}
\label{Lag_B-ii}
& & G_{\cal M}^{tt}\left(\Omega^3_{24}\right)^2\left(\Omega_+\right)_{tmp}\left(\Omega_+\right)_t^{\ \ mp} = \left(\alpha_1\alpha_3\right)^2\left(\Omega^3_{24}\right)^4\left(\left(e^2\right)^TG_{\cal M}^{-1}e^2\right)\left(e^4 G_{\cal M}^{-1}\left(e^4\right)^T\right);\nonumber\\
& & \left(\Omega_+\right)_{pmn}\left(\Omega_+\right)^{pmn}=\left(\Omega^3_{24}\right)^2 \left(\left(e^2\right)^TG_{\cal M}^{-1}e^2\right) \left(\left(e^3\right)^TG_{\cal M}^{-1}e^3 \right)\left(\left(e^4\right)^TG_{\cal M}^{-1}e^4\right).
\end{eqnarray}

After using results of \cite{OR4}, one obtains:
\begin{eqnarray}
\label{einvGeTs}
\left(e^2\right)^T G_{\mathcal{M}}^{-1} e^2
&=& -\frac{0.0007\,\lambda_5^2\,(\sqrt{3})^2\,
            \sin^{61}\!\left(\theta_2\right)
            \left(0.3\,\rh^2 - 0.24\,r^2\right)^2}
          {\epsilon^2\,g_s^6\,\logN\,M^3\,N_f^3\,\rh^4\,\logr^5}
\nonumber\\[4pt]
&& \times
   \frac{g_s^6\,\logN^3\,M^3\,N_f^3\,
         \csc^{63}\!\left(\theta_2\right)\logr^3
         + 2.28571\times10^{34}\,r^6\sin^{64}\!\left(\theta_1\right)}
        {\left(\dfrac{-\gsNf\log\!\left(r^6+3r^4\rh^2\right)
                      +4\gsNf\log 4+8\pi}
                     {g_s}\right)^{2/3}};
\nonumber\\[8pt]
\left(e^3\right)^T G_{\mathcal{M}}^{-1} e^3
&=& -\frac{0.0009\,\logN^2\,(\sqrt{3})^2\,\csc^2\!\left(\theta_2\right)
           \left(\dfrac{\rh^2}{3}-\dfrac{r^2}{3}\right)^2}
          {\epsilon^2\,\rh^4\,\logr^2}
\nonumber\\[4pt]
&& \times
   \frac{1}{\left(\dfrac{-\gsNf\log\!\left(r^6+3r^4\rh^2\right)
                         +4\gsNf\log 4+8\pi}
                        {g_s}\right)^{2/3}};
\nonumber\\[8pt]
\left(e^4\right)^T G_{\mathcal{M}}^{-1} e^4
&=& \frac{3.4\times10^{33}\,N\,r^{14}
          \sin^{60}\!\left(\theta_1\right)\sin^{63}\!\left(\theta_2\right)}
         {g_s^{13}\,\logN^5\,M^7\,N_f^7
          \left(r^4 - 1.7\,r^2\rh^2 + 1.1\,\rh^4\right)^2
          \logr^9\,\mathcal{D}^{2/3}}
\nonumber\\[4pt]
&& -\frac{0.04\,g_s^6\,\logN^3\,M^3\,N\,N_f^3\,r^8
          \csc^4\!\left(\theta_1\right)\logr^3}
         {g_s^{13}\,\logN^5\,M^7\,N_f^7
          \left(r^4 - 1.7\,r^2\rh^2 + 1.1\,\rh^4\right)^2
          \logr^9\,\mathcal{D}^{2/3}},
\end{eqnarray}
where we define
\begin{equation}
\mathcal{D} \;\equiv\;
\frac{8\pi}{g_s} + N_f\!\left(-\log\!\left(r^6+3r^4\rh^2\right)\right)
+ 4N_f\log 4.
\end{equation}

Further,
\begin{equation}
\label{Omega324LON}
\Omega^3_{24} = 80\,\frac{g_s^{7/4}\left(\logN\right)^2 M\,N_f\logrh}{N}.
\end{equation}

Using (\ref{einvGeTs}) and (\ref{Omega324LON}), one obtains:
\begin{eqnarray}
\label{I7}
\mathcal{I}_7
&=& \frac{0.510696\,\lambda_5^2\,(\sqrt{3})^4\,
          \sin^{59}\!\left(\theta_2\right)\logN^8\,
          (r-\rh)^2}
         {\epsilon^4\,g_s^{12}\,\logN^4\,M^6\,N_f^6\,
          \rh^{10}\,\logrh^{12}\,
          \left(\dfrac{8\pi}{g_s}-6N_f\logrh\right)^2}
\nonumber\\[4pt]
&& \times
   \Bigl(g_s^6\,\logN^3\,M^3\,N_f^3\,
         \csc^{63}\!\left(\theta_2\right)\logrh^3
         + 2.28571\times10^{34}\,\rh^6\sin^{64}\!\left(\theta_1\right)\Bigr)
\nonumber\\[4pt]
&& \times
   \Bigl(3.4314\times10^{33}\,\rh^{14}
         \sin^{60}\!\left(\theta_1\right)\sin^{63}\!\left(\theta_2\right)
\nonumber\\[2pt]
&& \quad
         -\,0.0369536\,g_s^6\,\logN^3\,M^3\,N_f^3\,\rh^8\,
         \csc^4\!\left(\theta_1\right)\logrh^3\Bigr),
\end{eqnarray}
using which one obtains:
\begin{eqnarray}
\label{I7overGrr}
\frac{\mathcal{I}_7}{G_{rr}^{\mathcal{M}}}
&=& \frac{0.00325329\,\lambda_5^2\,\logN^4\,
          \sqrt{\dfrac{1}{N}}\,(\sqrt{3})^4\,
          \sin^{124}\!\left(\theta_1\right)\sin^{122}\!\left(\theta_2\right)
          (r - \rh)^3}
         {\epsilon^4\,g_s^{19/2}\,M^6\,N_f^6\,\rh\,\logrh^{12}}
\nonumber\\[4pt]
&& \times
   \frac{\left(\dfrac{8\pi - \gsNf\log\!\left(4\rh^6\right)}{g_s}\right)^{1/3}}
        {\left(\gsNf\logrh - 4.18879\right)^2
         \left(\gsNf\log\!\left(4\rh^6\right) - 25.1327\right)}
\nonumber\\[4pt]
&& \times
   \Bigl(g_s^6\,\logN^3\,M^3\,N_f^3\,
         \csc^{64}\!\left(\theta_1\right)\csc^{63}\!\left(\theta_2\right)\logrh^3
         - 9.28571\times10^{34}\,\rh^6\Bigr)
\nonumber\\[4pt]
&& \times
   \Bigl(g_s^6\,\logN^3\,M^3\,N_f^3\,
         \csc^{63}\!\left(\theta_2\right)\csc^{64}\!\left(\theta_1\right)\logrh^3
         + 2.28571\times10^{34}\,\rh^6\Bigr)
\nonumber\\[4pt]
&& +\; \mathcal{O}\!\left(\left(\frac{1}{N}\right)^{3/2}\right).
\end{eqnarray}
Further,
\begin{eqnarray}
\label{gM7 plus 3}
g_7 &=& \frac{g_s^{7/2}\,M^2\,N^{19/10}\,N_f^5\,\rh^2\,
              (r-\rh)\,\logrh^7}
             {3456\sqrt[3]{2}\cdot 3^{2/3}\pi^{31/6}\,r\,
              \alpha_{\theta_1}^6\,\alpha_{\theta_2}^4\,
              \sqrt[3]{-N_f\logrh}};
\nonumber\\[6pt]
g_3 &=& \frac{9\left(\dfrac{1}{N}\right)^{3/2} r^6}
             {512\,\pi^{7/2}\,g_s^{3/2}}
        \left(-N_f\log\!\left(9a^2 r^4+r^6\right)
              +\frac{8\pi}{g_s}
              +2\logN\,N_f
              -4N_f\log\!\left(\alpha_{\theta_1}\right)\right.
\nonumber\\[4pt]
&&\qquad\qquad\qquad\left.
              -\,4N_f\log\!\left(\alpha_{\theta_2}\right)
              +4N_f\log 4
        \right)^2
        +\mathcal{O}\!\left(\left(\frac{1}{N}\right)^{5/2}\right).
\end{eqnarray}

Using (\ref{I7overGrr}) and (\ref{gM7 plus 3}), one obtains:
\begin{eqnarray}
\label{Krr}
\mathcal{K}_{rr}
&=& -\int_{M_7}\!d^7y\;
    \frac{1.8\times10^{-6}\,\lambda_5^2\,\logN^4\,(\sqrt{3})^4\,
          (g_s N)^{3/4}}
         {\epsilon^4\,g_s^{35/4}\,M^5\,N^{3/4}\,N_f^{9/2}\,
          \sqrt{r}\,\logrh^{13}\,
          \left(\gsNf\logrh - 4.18879\right)^2}
\nonumber\\[4pt]
&& \times
   \sin^{122}\!\left(\theta_1\right)\sin^{121}\!\left(\theta_2\right)
   (r-\rh)^{7/2}(-\logrh)^{7/2}
   \sqrt[6]{-N_f\logrh}
\nonumber\\[4pt]
&& \times
   \Bigl(g_s^6\,\logN^3\,M^3\,N_f^3\,
         \csc^{64}\!\left(\theta_1\right)\csc^{63}\!\left(\theta_2\right)\logrh^3
         - 9.28571\times10^{34}\,\rh^6\Bigr)
\nonumber\\[4pt]
&& \times
   \Bigl(g_s^6\,\logN^3\,M^3\,N_f^3\,
         \csc^{63}\!\left(\theta_2\right)\csc^{64}\!\left(\theta_1\right)\logrh^3
         + 2.28571\times10^{34}\,\rh^6\Bigr).
\end{eqnarray}

Assuming $\omega_{x^1} = \omega_{x^2} = \omega_{x^3}$, one obtains:
\begin{eqnarray}
\label{Nomega}
\mathcal{N}_\omega
&=& \frac{128 g_s N\omega^2\pi^{7/3}}{3^{4/3}\,r^4
    \left(2 N_f\log N -6 N_f\log\! r \right)^{4/3}},
\end{eqnarray}
using which one obtains:
\begin{eqnarray}
\label{Km}
\hskip -0.5 in  \mathcal{K}_m
&=& -\int_{M_7}\!d^7y\;
    \frac{8.7\times10^{-6}\,\lambda_5^2\,\omega^2\,r^{7/2}\,
          (\sqrt{3})^4\,(g_s N)^{3/4}}
         {\epsilon^4\,g_s^{45/4}\,\logN^4\,M^5\,N^{5/4}\,N_f^{11/3}\,
          \rh^9\,|\logrh|^{73/6}}
\nonumber\\[4pt]
\hskip -0.5 in && \times
   \frac{\sin^{58}\!\left(\theta_2\right)\csc^2\!\left(\theta_1\right)\,
         \logN^8\,(r-\rh)^{5/2}|\logrh|^{7/2}\,
         }
        {\left(2\log N-N_f\log\!\left(r^6+3r^4\rh^2\right)\right)^{10/3}}\left( \frac{4096g_s^2N^2\pi^{14/3}}{6^{8/3}N_f^{8/3}r^8}\right)
\nonumber\\[4pt]
\hskip -0.5 in&& \times
   \Bigl(g_s^6\,\logN^3\,M^3\,N_f^3\,
         \csc^{63}\!\left(\theta_2\right)\logrh^3
         + 2.28571\times10^{34}\,\rh^6\sin^{64}\!\left(\theta_1\right)\Bigr)
\nonumber\\[4pt]
\hskip -0.5 in&& \times
   \Bigl(3.4314\times10^{33}\,\rh^{14}
         \sin^{60}\!\left(\theta_1\right)\sin^{63}\!\left(\theta_2\right)
\nonumber\\[2pt]
\hskip -0.5 in&&\quad
         -\,0.0369536\,g_s^6\,\logN^3\,M^3\,N_f^3\,\rh^8\,
         \csc^4\!\left(\theta_1\right)\logrh^3\Bigr).
\end{eqnarray}

One can show,
\begin{eqnarray}
\label{Kmtheta12integral-finite}
&&\lim_{\delta_{1,2}\rightarrow0^+}
  \int_{\delta_2}^{\pi-\delta_2}\!\!d\theta_2
  \int_{\delta_1}^{\pi-\delta_1}\!\!d\theta_1\;
  \sin^{58}\!\left(\theta_2\right)\csc^2\!\left(\theta_1\right)
\nonumber\\[4pt]
&&\quad\times
  \Bigl(g_s^6\,\logN^3\,M^3\,N_f^3\,
        \csc^{63}\!\left(\theta_2\right)\logrh^3
        + 2.28571\times10^{34}\,\rh^6\sin^{64}\!\left(\theta_1\right)\Bigr)
\nonumber\\[4pt]
&&\quad\times
  \Bigl(3.4314\times10^{33}\,\rh^{14}
        \sin^{60}\!\left(\theta_1\right)\sin^{63}\!\left(\theta_2\right)
\nonumber\\[2pt]
&&\qquad
        -\,0.0369536\,g_s^6\,\logN^3\,M^3\,N_f^3\,\rh^8\,
        \csc^4\!\left(\theta_1\right)\logrh^3\Bigr)
  \bigg|^{\rm finite}
\nonumber\\[6pt]
&&= \frac{2\,\kappa_{\phi_1\phi_2\psi x^{10}}\,\rh^8}{T}
   \Bigl(-2\times10^{156}\,g_s^{12}\,\logN^6\,M^6\,N_f^6\,\logrh^6
\nonumber\\[4pt]
&&\qquad
         +\,1.40112\times10^{32}\,g_s^6\,\logN^3\,M^3\,N_f^3\,\rh^6\,\logrh^3
         + 1.01399\times10^{66}\,\rh^{12}\Bigr),
\end{eqnarray}
where $\kappa_{\phi_1\phi_2\psi x^{10}} = 4\pi\left(2\pi\right)^3$.
Similarly, in the context of $\mathcal{K}_{rr}$,
\begin{eqnarray}
\label{Krrtheta1theta2-finite}
&&\lim_{\delta_{3,4}\rightarrow0^+}
  \int_{\delta_3}^{\pi-\delta_3}\!\!d\theta_2
  \int_{\delta_4}^{\pi-\delta_4}\!\!d\theta_1\;
  \sin^{122}\!\left(\theta_1\right)\sin^{121}\!\left(\theta_2\right)
\nonumber\\[4pt]
&&\quad\times
  \Bigl(g_s^6\,\logN^3\,M^3\,N_f^3\,
        \csc^{64}\!\left(\theta_1\right)\csc^{63}\!\left(\theta_2\right)\logrh^3
        - 9.28571\times10^{34}\,\rh^6\Bigr)
\nonumber\\[4pt]
&&\quad\times
  \Bigl(g_s^6\,\logN^3\,M^3\,N_f^3\,
        \csc^{63}\!\left(\theta_2\right)\csc^{64}\!\left(\theta_1\right)\logrh^3
        + 2.28571\times10^{34}\,\rh^6\Bigr)
\nonumber\\[6pt]
&&= \frac{\kappa_{\phi_1\phi_2\psi x^{10}}}{T}
   \Bigl(7.51807\times10^{33}\,g_s^6\,\logN^3\,M^3\,N_f^3\,\rh^6\,\logrh^3
         + 5.46547\times10^{67}\,\rh^{12}\Bigr).
\end{eqnarray}

Therefore, one obtains:
\begin{eqnarray}
\label{Nfeff+ders coeff}
\frac{\left.\sqrt{g_3\,G_{rr}^{\mathcal{M}}}\,\mathcal{K}_m(r)\right|_{r=r_h}}{ (r-\rh)^2}
&=& \frac{{\cal K}_{\sqrt{g_3\,G_{rr}^{\mathcal{M}}}\,\mathcal{K}_m}\,\lambda_5^2\,
           \kappa_{\phi_1\phi_2\psi x^{10}}\,N\,\omega^2\,
           (\sqrt{3})^4\,(\logN)^{10}g_s^{3}M\,N_f^{1/3}}
          {\epsilon^4\,\rh^{7/2}\,
           T\,|\logrh|^{32/3}\,} \equiv\alpha_3  \nonumber\\[4pt]
-\frac{\left.\sqrt{g_3\,G_{rr}^{\mathcal{M}}}\,|\omega_1|^2\,\mathcal{K}_{rr}\right|_{r=r_h}}{(r-\rh)^3}
&=& \frac{{\cal K}_{\sqrt{g_3\,G_{rr}^{\mathcal{M}}}\,|\omega_1|^2\,\mathcal{K}_{rr}}\,\lambda_5^2\,
          \kappa_{\phi_1\phi_2\psi x^{10}}\,\logN^4\,\omega^2\,\rh^{4}\,
          (\sqrt{3})^4\,\sqrt{N}}
         {\epsilon^4\,g_s^{7/2}\,M^2\,N_f^{11/3}\,T\,
          |\logrh|^{25/3}}\equiv \alpha_2,
\end{eqnarray}
where ${\cal K}_{\sqrt{g_3\,G_{rr}^{\mathcal{M}}}\,\mathcal{K}_m}\sim{\cal K}_{\sqrt{g_3\,G_{rr}^{\mathcal{M}}}\,|\omega_1|^2\,\mathcal{K}_{rr}}^6\gg1$.

One thus obtains the following differential equation:
\begin{equation}
\label{B(r)-EOM}
{\alpha_1} (r-{r_h})^2 {\cal B}_{\rm KK}'(r)+{\alpha_2}
   (r-{r_h})^3 {\cal B}_{\rm KK}''(r)+{\alpha_3} {\cal B}_{\rm KK}(r)
   (r-{r_h})^2=0,
\end{equation}
where $\alpha_1 = 3 \alpha_2$.

The general solution of \eqref{eq:NeffEOM} for $r>\rh$ is
\begin{equation}
\mathcal{B}_{\rm KK}(r) = \left(\frac{\ath(\rh-r)}{\ade}\right)^{\!(\ade-\anu)/(2\ade)}
\left[
  C_1  I_{\nu}\!\left(2\sqrt{\frac{\ath(\rh-r)}{\ade}}\right)
  + C_2
  K_{\nu}\!\left(2\sqrt{\frac{\ath(\rh-r)}{\ade}}\right)
\right],
\label{eq:solution}
\end{equation}
where $\nu = \anu/\ade - 1 = 2$. 
The two conditions imposed throughout are
\begin{equation}
\textbf{(A)}\;\;{\cal B}_{\rm KK}(r)\to 0\quad(r\to+\infty),\qquad
\textbf{(B)}\;\;0<|{\cal B}_{\rm KK}(\rh)|<+\infty.
\label{eq:condAB}
\end{equation}
For $r > \rh$ the quantity $\rh - r < 0$, so $z^2 =
4\ath(\rh-r)/\ade < 0$ as $\ath/\ade > 0$. As $\ath/\ade > 0$, so one defines $z = iw$ with $w = 2\sqrt{\ath(r-\rh)/\ade} \in \mathbb{R}^+$ for $r > \rh$. Hence, (\ref{eq:solution}) is effectively expressed in terms of $J_\nu(w)$ and $Y_\nu(w)$:
\begin{equation}
\mathcal{B}_{\rm KK}(r) = \left(\frac{\ath(\rh-r)}{\ade}\right)^{\!(\ade-\anu)/(2\ade)}
\left[
  \tilde{C}_1  J_{\nu}\!\left(2\sqrt{\frac{\ath(r - \rh)}{\ade}}\right)
  + \tilde{C}_2
  Y_{\nu}\!\left(2\sqrt{\frac{\ath(r -\rh)}{\ade}}\right)
\right],
\label{eq:solution-JY}
\end{equation}

For every $\nu > 0$, the net exponent of $(r-\rh)$ contributed by the
outer prefactor \eqref{eq:solution} and the leading power of $I_\nu$
or $J_\nu$ in \eqref{eq:solution} is identically zero:
\begin{equation}
\frac{\ade-\anu}{2\ade} + \frac{\nu}{2} = 0.
\label{eq:identity}
\end{equation}
Consequently, $\mathrm{Prefactor}(r)\cdot I_\nu(z)$ and
$\mathrm{Prefactor}(r)\cdot J_\nu(w)$ both approach a finite nonzero
constant as $r \to \rh^+$, for any $\nu > 0$.

For $K_\nu$ and $Y_\nu$, the corresponding net exponent is
\begin{equation}
\frac{\ade-\anu}{2\ade} - \frac{\nu}{2} = \frac{\ade-\anu}{\ade} - \nu
= \frac{-2(\anu - \ade)}{\ade},
\label{eq:Kexponent}
\end{equation}
which is negative for $\anu > \ade$ (divergence), zero for $\anu =
\ade$ (logarithmic divergence in $K_0, Y_0$), and positive for
$\anu < \ade$ (the function vanishes). Hence, for normalizable solutions at the horizon, $C_2 = \tilde{C}_2=0$. Horizon normalizability picks out the following solution from (\ref{eq:solution-JY}):
\begin{equation}
{\cal B}_{\rm KK}(r) = {\cal B}_{\rm KK}(r=r_h) \frac{J_2\!\left(2\sqrt{\frac{\ath}{\ade}(r-\rh)}\right)}{r - r_h},
\quad r > \rh,
\label{eq:solution-final}
\end{equation}
where one is free to choose $B(r=r_h)$.
\end{proof}

Consider turning on an $r$-dependent field ${\cal B}_{D6}(r): F_{x^1x^2} = \partial_{[x^1}A_{x^2]}, A_{x^1} = x^2 {\cal B}_{D6}(r)/2, A_{x^2} = - x^1{\cal B}_{D6}(r)/2$, on the world-volume $\Sigma^{(1+6)}\left(\cong S^1_t\times_w\mathbb{R}^3\right)\times_w S^2_{\rm squashed}$ of the dual type IIA flavor $D6$-branes results in the DBI action $\int_{\Sigma^{(1+6)}}e^{-\Phi^{IIA}}\sqrt{i^*g^{IIA} + F}, i:\Sigma^{(1+6)}\hookrightarrow M_{11}$.

In the IR, the EOM for ${\cal B}_{D6}(Z)$ can be written as
\begin{equation}
\label{IR-EOM-BD6}
 r_h^{2}\big[48\,Z B''(Z) + (196 Z+48) B'(Z)\big]
- 16\pi g_s N (30Z+9) {\cal B}_{D6}(Z) = 0 \, ,
\end{equation}
whose general solution is
{\footnotesize
\begin{eqnarray}
{\cal B}_{D6}(Z) &=& \exp\!\left(-\frac{Z\big(\sqrt{5760\pi g_sN + 2401 r_h^2} + 49 r_h\big)}{24 r_h}\right) \nonumber\\
&&\times\Bigg[c_1\, U\!\left(\frac{36 g_s\pi N}{r_h\sqrt{2401 r_h^2+5760 g_sN\pi}}+\frac{49 r_h}{2\sqrt{2401r_h^2+5760g_sN\pi}}+\frac12,\;1,\;\frac{Z\sqrt{2401 r_h^2+5760 g_sN\pi}}{12 r_h}\right) \nonumber\\
&&\hphantom{\times\Bigg[}+\,c_2\, L_{-\frac{36 g_s\pi N}{r_h\sqrt{2401 r_h^2+5760g_sN\pi}}-\frac{49 r_h}{2\sqrt{2401r_h^2+5760g_sN\pi}}-\frac12}\!\left(\frac{Z\sqrt{2401 r_h^2+5760g_sN\pi}}{12 r_h}\right)\Bigg] \, ,
\end{eqnarray}
}
where $U(a,b,x)$ is Tricomi's confluent hypergeometric function and $L_\nu(x)$ the generalized Laguerre function of order $\nu$.

Expanding this solution about $Z=0$ shows that the coefficient of $c_1$ carries a logarithmic divergence, $\log Z$. Regularity at the horizon therefore requires, $c_1 = 0$,leaving the regular IR solution
\begin{equation}
{\cal B}_{D6}(Z) = C_2\, e^{-AZ}\, L_\nu(BZ) \, ,
\end{equation}
with
\begin{align}
A &= \frac{49 r_h + \sqrt{5760 g_sN\pi + 2401 r_h^2}}{24 r_h} \, , \\
B &= \frac{\sqrt{5760 g_sN\pi + 2401 r_h^2}}{12 r_h} \, , \\
\nu &= -\frac12 - \frac{36 g_sN\pi}{ r_h\sqrt{5760g_sN\pi+2401 r_h^2}} - \frac{49 r_h}{2\sqrt{5760g_sN\pi+2401 r_h^2}} \, .
\end{align}

Expanding for small $Z$,
\begin{equation}
e^{-AZ} = 1 - AZ + \mathcal{O}(Z^2) \, , \qquad
L_\nu(BZ) = L_\nu(0) + BZ\,L_\nu'(0) + \mathcal{O}(Z^2) \, ,
\end{equation}
and using the identities $L_\nu(0)=1$ and $L_\nu'(0)=-\nu$, the solution reduces to
\begin{equation}
{\cal B}_{D6}(Z) = C_2\big[1 + (-A - B\nu)\,Z + \mathcal{O}(Z^2)\big] \, .
\end{equation}
Substituting the explicit expressions for $A$, $B$, and $\nu$, the coefficient of the linear term simplifies to
\begin{equation}
-A - B\nu = \frac{3 g_sN\pi}{ r_h^2} \, ,
\end{equation}
yielding the IR solution
\begin{equation}
\,{\cal B}_{D6}(Z) = B(Z=0)\left(1 + \frac{3g_sN\pi}{ r_h^2}\,Z\right) + \mathcal{O}(Z^2)\, .
\end{equation}

In the UV, the equation of motion for ${\cal B}_{D6}(Z)$ reduces to
\begin{equation}
 r_h^{2}\Big[e^{8Z}B''(Z) + 3e^{8Z}B'(Z)\Big] - 16\pi g_s N\, e^{6Z}\,{\cal B}_{D6}(Z) = 0 \, ,
\end{equation}
whose general solution takes the form
\begin{equation}
\begin{aligned}
{\cal B}_{D6}(Z)=\;&
\frac{2\sqrt{g_sN\pi}}{r_h}\,e^{-Z}C_1
\left[
-\frac{r_h\cosh(4\xi)}{2\sqrt{g_sN\pi}}
+2\sinh(4\xi)
\right]\\
&-\frac{3i\sqrt{g_sN\pi}}{4r_h}\,e^{-Z}C_2
\left[
2\cosh(4\xi)
-\frac{r_h\sinh(4\xi)}{2\sqrt{g_sN\pi}}
\right],
\end{aligned}
\end{equation}
where
\begin{equation}
\xi=\frac{\sqrt{g_sN\pi}}{r_h}\,e^{-Z}.
\end{equation}

Using $\cosh(4\xi)=\tfrac12(e^{4\xi}+e^{-4\xi})$ and $\sinh(4\xi)=\tfrac12(e^{4\xi}-e^{-4\xi})$, the solution decomposes cleanly into exponentially growing and decaying pieces,
\begin{equation}
{\cal B}_{D6}(Z)
=
-\frac{1}{16}(4\xi-1)(-8C_1+3iC_2)\,e^{4\xi}
+\frac{1}{16}(4\xi+1)(-8C_1-3iC_2)\,e^{-4\xi}.
\end{equation}

Regularity requires the exponentially growing mode to vanish. We consider the two natural conditions that achieve this.

\paragraph{Case 1: $8C_1-3iC_2=0$.} 
One gets,
\begin{equation}
{\cal B}_{D6}(Z)
=
-\frac{1}{16}(4\xi+1)(8C_1+3iC_2)\,e^{-4\xi}.
\end{equation}
Imposing the constraint, this simplifies to
\begin{equation}
{\cal B}_{D6}(Z)
=
-C_1(1+4\xi)\,e^{-4\xi},
\qquad
\xi=\frac{\sqrt{g_sN\pi}}{r_h}\,e^{-Z},
\end{equation}
or equivalently,
\begin{equation}
{\cal B}_{D6}(Z)
=
-C_1
\left(
1+\frac{4\sqrt{g_sN\pi}}{r_h}e^{-Z}
\right)
\exp\!\left[
-\frac{4\sqrt{g_sN\pi}}{r_h}e^{-Z}
\right].
\end{equation}

\paragraph{Case 2: $8C_1+3iC_2=0$.}
This results in
\begin{equation}
{\cal B}_{D6}(Z)\propto \xi\,C_1\,e^{4\xi},
\qquad
\xi=\frac{\sqrt{g_sN\pi}}{r_h}\,e^{-Z}.
\end{equation}
Redefining $\frac{\sqrt{g_sN\pi}}{r_h}C_1\to {\cal B}_{D6}(r_h)$, the asymptotic solution becomes
\begin{equation}
{\cal B}_{D6}(Z)\sim
\, {\cal B}_{D6}(r_h)\,e^{-Z}
\exp\!\left(
\frac{4\sqrt{g_sN\pi}}{r_h}e^{-Z}
\right).
\label{BD6-UV}
\end{equation}
Assuming one performs a large-$Z$ expansion first, ${\cal B}_{D6}(r\in{\rm UV})\sim\frac{1}{r - r_h}$ and from (\ref{eq:solution-final}), ${\cal B}_{KK}(r\in{\rm UV})\sim\frac{1}{(r - r_h)^{5/4}}\sim{\cal B}_{D6}^{5/4}(r\in{\rm UV})$; we further identity  ${\cal B}_{\rm KK}(r_h) = {\cal B}_{D6}^{5/4}(r=r_h)$.

\section{Conclusion and outlook}
\label{Conclusion}

The results of this paper, consolidated in Theorem~\ref{thm:mainresult}, identify a contact structure(derivable from a $G_2$-structure)-induced symplectic ``nearly'' half-flat transverse
$SU(3)$-structure on a closed seven-fold in the Infra Red (i.e., small values of the conifold radial coordinate $r$) in the finite-$N$ MQGP ${\cal M}$-theory dual of thermal QCD-like theories and exhibit one of its concrete dynamical consequences - geometrical origin of a weak ``magnetic'' field. We summarize the
mathematical content and indicate the questions it leaves open.

On the geometric side, we have shown
(Proposition~\ref{proposition}, via
Lemmas~\ref{LemmaSymplecticSU3}--\ref{HalfFlatSymplecticSU3}) that the transverse $SU(3)$
structure $(\omega_\Phi,\Omega)$ induced from the Contact 3-Structure of the $G_2$-structure
seven-fold $\mathcal{M}_7$ (a warped product of the ${\cal M}$-theory circle $S^1_{\mathcal{M}}$ with a non-K\"{a}hler six-fold, which is itself a thermal-circle fibration over a non-Einsteinian deformation of the Sasaki--Einstein coset $T^{1,1}$) has, in the MQGP limit (\ref{eq:MQGP}) and in the infrared, a single non-vanishing
intrinsic torsion class, $W_\Phi^{SU(3)} = W_2^-$. Further, $d\omega_\Phi = 0$ and
$d\Omega_+ = \mathcal{O}\!\left(e^{-N^{1/3}/\mathcal{O}(1)}\right)$ (cf.\ \eqref{eq:dOmegaIR}), so that the
structure is symplectic and \emph{nearly} half-flat, the departure from exact half-flatness
being exponentially suppressed in $N^{1/3}$ in the MQGP limit (\ref{eq:MQGP}). This places the transverse geometry unambiguously within the
Chiossi--Salamon classification of $SU(3)$ structures and completes the torsion-class
identification for this dual: the half-flat locus is recovered exactly in the strict large-$N$
limit, and the finite-$N$-suppressed corrections are exponential-in-$N^{1/3}$-suppressed and hence negligible. We emphasize that this is a statement about the intrinsic torsion alone in the IR.

On the analytic side, the same half-flat structure was shown to control a genuine field
equation. Kaluza--Klein reduction of the eleven-dimensional Hodge dual of the differential of the positive $G_2$-structure three-form on $\mathcal{M}_7$,
organized through the $G_2$/symplectic half-flat $SU(3)$ data, yields a weak four-dimensional
magnetic field whose squared norm factorizes cleanly into external and internal contributions
(Proposition~\ref{prop:normfact}) and whose radial profile $\mathcal{B}_{\rm KK}(r)$ satisfies a
Bessel-type ordinary differential equation with a regular singular point at the horizon
(Proposition~\ref{prop:ODE}). Normalisability and infrared regularity select
$J_2$  as the solution that is normalisable at the horizon
and decays for $r\gg r_h$. The harmonicity of the aforementioned internal factor is not assumed but \emph{derived}, and is enforced by the same exponential-in-$N$ suppression that underlies the geometric result.

Taken together, the two halves of Theorem~\ref{thm:mainresult} show that a purely
torsion-theoretic statement about an $SU(3)$ structure leads up to a well-posed,
analytically solvable boundary value problem, with the same  parameter $g_s M^2/N<1$ and $r_h\sim e^{-\frac{N^{1/3}}{{\cal O}(1)}}\ll1, |\log r_h|\sim N^{1/3}>1$, governing both the geometry and the dynamics. The interest of the result for the present
context is precisely that the bridge between the symplectic nearly half-flat $SU(3)$-structure classification and the field equation is explicit and controlled at every step.

Several directions merit further study. First, the torsion-class identification has been
established to leading non-trivial order in $g_s M^2/N$; a systematic treatment of the
higher-order corrections, and in particular a proof that no class other than $W_2^-$ is
activated to all orders, would sharpen Proposition~\ref{proposition} into an exact statement. Second, the Bessel-type reduction relies on the integrated coefficient functions $\anu,\ade,\ath$ respectively of $(r - r_h)^2{\cal B}_{\rm KK}(r), (r - r_h)^3{\cal B}''_{\rm KK}(r), (r - r_h)^2{\cal B}_{\rm KK}(r)$ inheriting definite signs from the background; a coordinate-invariant characterisation of the ratio $\ath/\ade$ in terms of $G_2$-structure data would
clarify the geometric origin of the $J_2$ selection. Finally, the symplectic nearly half-flat
structure isolated here is a natural starting point for studying flux deformations and the
associated moduli, and for comparison with half-flat structures arising in other
compactifications. 

\section*{Acknowledgements}

The author thanks the Department of Physics, McGill University, for its warm hospitality during a visit in which this work was initiated, and Keshav Dasgupta in particular for  discussions and for encouragement to write up these results as a separate paper, as well as IIT Roorkee for travel support. The author acknowledges the use of Claude Opus 4.8~(Anthropic) as an AI research assistant for help with \LaTeX\ typesetting, organisation of results (that are the author's) into the lemma--proposition--theorem framework. The \textsf{TikZ} code for Fig.~1 was generated with the assistance of ChatGPT~(OpenAI). 

\appendix
\section{$\Omega^1_a, \Omega^7_a, \Omega^a_{bc}$}
\label{Omegas}
\setcounter{equation}{0} \seceqaa

In this appendix, we list out the expresions for $\Omega^1_a, \Omega^7_a, \Omega^a_{bc}$ inclusive of ${\cal O}(\beta)$-corrections as derived in \cite{ACMS} (arXiv version 1), relevant to (\ref{Omega13Omega3ab}) - (\ref{Omega3bcOmega3dfantisymm}).

\noindent$\bullet{\bf \Omega^1_a}$ Components

{\scriptsize
\begin{eqnarray}
\label{Omega1a}
& & \hskip -0.5in \Omega^1_2=\frac{80 N^{7/20} \alpha _{\theta_2} \left(2 a^2+r^2\right)}{M r^2 \alpha _{\theta_1}^2 N_f g_s^{7/4} \log (r) \left(2\tilde{\cal A}\right)} -\beta^{1/4}\left(\frac{\tilde{\omega}_{12} {\cal C}_q \sqrt{{\cal C}_{zz}^{(1)}} N^{17/20} \sqrt{\alpha _{\theta_1}}}{\alpha _{\theta_2}^4 \sqrt[4]{g_s} \left(\tilde{\cal A}\right)}\right)\nonumber\\
   & & \hskip -0.5in \Omega^1_3=\frac{56.7618 N^{7/20} \alpha _{\theta_2}}{M \alpha _{\theta_1}^2 N_f g_s^{7/4} \log (r) \left(\tilde{\cal A}\right)}  -\beta^{1/4}\left(\frac{{\tilde{\omega}_{13}} {\cal C}_q \sqrt{{\cal C}_{zz}^{(1)}} N^{17/20} \sqrt{\alpha _{\theta_1}}}{\alpha _{\theta_2}^4 \sqrt[4]{g_s} \left(\tilde{A}\right)}\right) \nonumber\\
   & & \hskip -0.5in \Omega^1_4=\frac{15.6835 M \left(r^2-2.85714 a^2\right) N_f g_s^{7/4} \log ^2(N) \log (r)}{\sqrt[4]{N} r^2 \left(\tilde{A}\right)}  -\beta^{1/4}\left(\frac{\tilde{\omega}_{14} {\cal C}_q \sqrt{{\cal C}_{zz}^{(1)}} N^{11/20} \sqrt{\alpha _{\theta_1}} N_f}{\alpha _{\theta_2}^3 \sqrt[4]{g_s} \log (N) \sqrt[3]{N_f \left(2 \log (N)-\log
   \left(9 a^2 r^4+r^6\right)\right)} \left(N_f \left(2\tilde{A}\right)\right){}^{2/3}}\right)\nonumber\\   
   & & \hskip -0.5in \Omega^1_5=-\frac{88.059 N^{7/20} \alpha _{\theta_2}}{M \alpha _{\theta_1}^2 N_f g_s^{7/4} \log (r) \left(\tilde{A}\right)} +\beta^{1/4}\left(-\frac{\tilde{\omega}_{15} {\cal C}_q \sqrt{{\cal C}_{zz}^{(1)}} N^{17/20} \sqrt{\alpha _{\theta_1}}}{\alpha _{\theta_2}^4 \sqrt[4]{g_s} \left(\tilde{A}\right)}\right)\nonumber\\
   & & \hskip -0.5in \Omega^1_6=-\frac{1.29 M \sqrt[20]{N} \left(a^2-0.5 r^2\right) N_f^2 g_s^{7/4} \log (N) \log (r)}{r^2 \alpha _{\theta_2} \sqrt[3]{N_f \left(2 \log (N)-\log \left(9 a^2 r^4+r^6\right)\right)} \left(N_f
   \left(-\log \left(9 a^2 r^4+r^6\right)-4 \left(\log \left(\alpha _{\theta_1}\right)+\log \left(\alpha _{\theta_2}\right)\right)+2 \log (N)\right)\right){}^{2/3}} \nonumber\\
   & & \hskip -0.5in +\beta^{1/4}\left(\frac{\tilde{\omega}_{16} \sqrt[4]{\beta } {\cal C}_q \sqrt{{\cal C}_{zz}^{(1)}} N^{17/20} \sqrt{\alpha _{\theta_1}} N_f}{\alpha _{\theta_2}^4 \sqrt[4]{g_s} \log (N) \sqrt[3]{N_f
   \left(2 \log (N)-\log \left(9 a^2 r^4+r^6\right)\right)} \left(N_f \left(2\tilde{A}\right)\right){}^{2/3}}\right),\nonumber\\
\end{eqnarray}
}
where $\tilde{\cal A} \equiv  \log\left(\frac{N}{\sqrt{\left(9 a^2 r^4+r^6\right)}\alpha^2_{\theta_1}\alpha^2_{\theta_2}}\right)$.

\noindent$\bullet{\bf\Omega^7_a}$ Components
{\scriptsize
\begin{eqnarray}
\label{Omega7a}
& & \hskip -0.5in\Omega^7_2=-\frac{160 N^{7/20} \alpha _{\theta_2} \left(2 a^2+r^2\right)}{M r^2 \alpha _{\theta_1}^2 N_f g_s^{7/4} \log (r) \left(2\tilde{\cal A}\right)} -\beta^{1/4}\tilde{\omega}_{72}\left(\frac{{\cal C}_q \sqrt{{\cal C}_{zz}^{(1)}} N^{17/20} \sqrt{\alpha _{\theta_1}} \left(\left(-{6.1} a^2-{0.7} r^2\right) \log (N)+\left({3.06} a^2+{0.3} r^2\right) \log \left(9 a^2 r^4+r^6\right)\right)}{\alpha _{\theta_2}^4 \sqrt[4]{g_s} \log (N) \left(\tilde{\cal A}\right) \left(\tilde{D}\right)}\right) \nonumber\\
   & &\hskip -0.5 in\Omega^7_3=-\frac{3784.1 N^{7/20} \alpha _{\theta_2}}{M r^2 \alpha _{\theta_1}^2 N_f^2 g_s^{11/4} \log (r) \left(\tilde{\cal A}\right){}^2}  \times  \Biggl[(r N_f g_s \left(a^2 \log (r) \left(\tilde{\cal A}\right) -0.01 r \log \left(9 a^2 r^4+r^6\right) +0.03 r \log
   (N)-0.07 r \log \left(\alpha _{\theta_1}\alpha _{\theta_2}\right)\right)\Biggr] \nonumber\\
   & & \hskip -0.5in +\beta^{1/4}\tilde{\omega}_{73}\left(\frac{{\cal C}_q \sqrt{{\cal C}_{zz}^{(1)}} N^{17/20} \sqrt{\alpha _{\theta_1}}}{\alpha _{\theta_2}^4 \sqrt[4]{g_s} \log (N) \left(\tilde{D}\right)}\right) \nonumber\\
 & & \hskip -0.5in  \times\left(\frac{N_f g_s \left(({2.5} a^2 + 0.3 r^2) \log \left(\alpha _{\theta_1}\alpha_{\theta_2}\right)+\left(-{1.2} a^2-{0.1} r^2\right) \log (N)+\left({0.6} a^2+{0.1} r^2\right) \log \left(9 a^2
   r^4+r^6\right)-{0.4}
   r^2\right)-{1.7} r^2}{N_f g_s \left(\tilde{\cal A}\right)}\right)\nonumber\\
   & & \hskip -0.5 in \Omega^7_4=\frac{5870.7 N^{7/20} \alpha _{\theta_2}}{M r^2 \alpha _{\theta_1}^2 N_f^2 g_s^{11/4} \log (r) \left(\tilde{\cal A}\right){}^2} \Biggl[r N_f g_s \left(a^2 \log (r) \left(\tilde{\cal A}\right)-0.01 r \log \left(9 a^2 r^4+r^6\right) +0.03 r \log
   (N)-0.07 r \log \left(\alpha _{\theta_1}\alpha_{\theta_2}\right)\right)\Biggr] \nonumber\\
   & & \hskip -0.5in -\beta^{1/4}\tilde{\omega}_{74}\Biggl[\frac{{\cal C}_q \sqrt{{\cal C}_{zz}^{(1)}} N^{17/20} \sqrt{\alpha _{\theta_1}}}{\alpha _{\theta_2}^4 \sqrt[4]{g_s} \log (N) \left(\tilde{D}\right)}\Biggr] \left(\frac{-({19} a^2 + 2.2 r^2) \log \left(\alpha _{\theta_1}\alpha_{\theta_2}\right)+\left( a^2+{1.1} r^2\right) \log (N)+\left(-{4.9} a^2-{0.6} r^2\right) \log \left(9 a^2 r^4+r^6\right)}{\tilde{\cal A}}\right) \nonumber
\end{eqnarray}
}

{\scriptsize
\begin{eqnarray}
   & &\hskip -0.5in \Omega^7_5=\frac{176.121 N^{7/20} \alpha _{\theta_2}}{M r \alpha _{\theta_1}^2 N_f g_s^{7/4} \log (r) \left(\tilde{\cal A}\right){}^2} \Biggl[\log (N) \left(33.3 a^2 \log (r)+ r\right)+\left(-16.7 a^2 \log (r)-0.3 r\right) \log \left(9 a^2 r^4+r^6\right)-66.7 a^2 \log (r) \log \left(\alpha _{\theta_1}\alpha
   _{\theta_2}\right) \nonumber\\
  & &\hskip -0.5in -2.3 r \log \left(\alpha _{\theta_1}\alpha _{\theta_2}\right)\Biggr]  -\beta^{1/4}\tilde{\omega}_{75}\Biggl[\frac{10. {\cal C}_q \sqrt{{\cal C}_{zz}^{(1)}} N^{17/20} \sqrt{\alpha _{\theta_1}}}{\alpha _{\theta_2}^4 \sqrt[4]{g_s} \log (N) \left(\tilde{\cal A}\right)} \Biggr]
    \nonumber\\
   & &\hskip -0.5in \times \Biggl[\frac{N_f g_s \left(-({19} a^2 + 2.2 r^2)\log \left(\alpha_{\theta_1}\alpha _{\theta_2}\right)+\left({9.9}
   a^2+{1.1} r^2\right) \log (N)+\left(-{4.9} a^2-{0.6} r^2\right) \log \left(9 a^2 r^4+r^6\right)+{28} a^2+{3.1} r^2\right)+{120} a^2+{14} r^2}{N_f g_s \left(-(20 a^2 + 2 r^2) \log \left(\alpha _{\theta_1}\alpha _{\theta
   _2}\right)+\left(10. a^2+r^2\right) \log (N)+\left(-5 a^2-0.5 r^2\right) \log \left(9 a^2 r^4+r^6\right)+3. r^2\right)+14. r^2}\Biggr] \nonumber\\
& & \hskip -0.5in \Omega^7_6=\frac{12.9 M \sqrt[20]{N} N_f g_s^{7/4} \log (N) \log (r)}{r^2 \alpha _{\theta_2} \left(\tilde{A}\right){}^2} \Biggl[-0.2 a^2 \log \left(\alpha _{\theta_1}\right)-0.2 a^2 \log \left(\alpha _{\theta_2}\right)+\log (N) \left(a^2 r \log (r)+0.1 a^2-0.05 r^2\right)-2. a^2 r \log (r) \log \left(\alpha _{\theta_1}\right)\nonumber\\
& & \hskip -0.5in +\left(-0.5 a^2 r \log (r)-0.05 a^2+0.02 r^2\right) \log \left(9 a^2
   r^4+r^6\right)-2. a^2 r \log (r) \log \left(\alpha _{\theta_2}\right)+0.1 r^2 \log \left(\alpha _{\theta_1}\right)+0.1 r^2 \log \left(\alpha _{\theta_2}\right)\Biggr] \nonumber\\
   & & \hskip -0.5in - \beta^{1/4}\tilde{\omega}_{76}\left(\frac{\sqrt{{\cal C}_{x^{10}x^{10}}^{(1)}} {\cal C}_q N^{17/20} \sqrt{\alpha _{\theta_1}}}{\alpha _{\theta_2}^4 \sqrt[4]{g_s} \log (N) \left(\tilde{\cal A}\right)}\right) \Biggl[{6.6} a^2 \log \left(\alpha _{\theta_1}\right)+{6.6} a^2 \log \left(\alpha _{\theta_2}\right)+\left(-{3.3}
   a^2-{0.4} r^2\right) \log (N) \nonumber\\
  & & \hskip -0.5in  +\left({1.6} a^2+{0.2} r^2\right) \log \left(9 a^2 r^4+r^6\right)+{0.7} r^2 \log \left(\alpha _{\theta_1}\right)+{0.7} r^2 \log \left(\alpha _{\theta_2}\right)\Biggr]\frac{1}{\tilde{D}}
\end{eqnarray}
}
where $\tilde{D}\equiv(a^2+0.1r^2)\tilde{A}$.

\noindent (ii) $\Omega^3_{bc}$
{\footnotesize
\begin{eqnarray}   
\label{Omega3bc}
   & & \hskip -0.5in \Omega^3_{23}=-\frac{0.5 M N^{9/20} \left(r^2-3 a^2\right) \left(r^2-2 a^2\right) \left(2 a^2+r^2\right) N_f g_s^{7/4} \log ^2(N) \log (r)}{r^6 \alpha _{\theta_1}^2 \alpha _{\theta_2}}+\beta^{1/4}\tilde{\omega}^3_{23}\Biggl[-\frac{ {\cal C}_q \sqrt{{\cal C}_{zz}^{(1)}} N^{5/4}}{r^4 \alpha _{\theta_1}^{3/2} \alpha _{\theta_2}^4 \sqrt[4]{g_s} \log (N)}\nonumber\\
   & & \hskip -0.5in \times \left(33.3 a^2 r \log (r)-0.3 a^2+r^2\right) \left(\left(r^2-3.3 a^2\right) \log (N)-87.7 a^2 r \log(r)\right)\Biggr] \nonumber\\
   & & \hskip -0.5in \Omega^3_{24}=\frac{M \left(-988.8 a^4+352.3 a^2 r^2+72.4 r^4\right) N_f g_s^{7/4} \log ^2(N) \log (r)}{\sqrt[4]{N} r^4}-\beta^{1/4}\tilde{\omega}^3_{24}\Biggl[\frac{{\cal C}_q \sqrt{{\cal C}_{zz}^{(1)}} N^{19/20} \left(r^2-3 a^2\right) \left(2 a^2+r^2\right)}{r^4 \alpha _{\theta_1}^{3/2} \alpha _{\theta_2}^3 \sqrt[4]{g_s}}\Biggr] \nonumber\\
   & & \nonumber\\
& & \hskip -.5in \Omega^3_{25}=\frac{0.3 M N^{9/20} \left(r^2-3. a^2\right) \left(r^2-2. a^2\right) \left(2. a^2+r^2\right) N_f g_s^{7/4} \log ^2(N) \log (r)}{r^6 \alpha _{\theta_1}^2 \alpha _{\theta_2}}+\beta^{1/4}\tilde{\omega}^3_{25}\Biggl[\frac{ {\cal C}_q \sqrt{{\cal C}_{zz}^{(1)}} N^{5/4} }{r^4 \alpha
   _{\theta_1}^{3/2} \alpha _{\theta_2}^4 \sqrt[4]{g_s} \log (N)} \nonumber\\
   & & \hskip-0.5in \left(100. a^2 r \log (r)-4. a^2+r^2\right) \left(\left(r^2-3.3 a^2\right) \log (N)-87.6667 a^2 r \log (r)\right)\Biggr]\nonumber\\
   & & \hskip -0.5in \Omega^3_{26}=\frac{0.3 M N^{9/20} \left(r^2-3 a^2\right) \left(r^2-2 a^2\right) \left(2 a^2+r^2\right) N_f g_s^{7/4} \log ^2(N) \log (r)}{r^6 \alpha _{\theta_1}^2 \alpha _{\theta_2}} \nonumber\\
   & & \hskip -0.5in+ \beta^{1/4}\tilde{\omega}^3_{26}\Biggl[\frac{ {\cal C}_q \sqrt{{\cal C}_{zz}^{(1)}} N^{5/4} \left(2. a^2+r^2\right) \left(\left(r^2-3. a^2\right) \log (N)-78.9 a^2 r \log (r)\right)}{r^4 \alpha _{\theta_1}^{3/2} \alpha _{\theta
   _2}^4 \sqrt[4]{g_s} \log (N)}\Biggr]
   \nonumber\\
   & & \hskip -0.5in \Omega^3_{34}=-\frac{10.4 M^3 \left(1.7 a^2-0.8 r^2\right) \left(a^2-0.3 r^2\right) \left(0.2 a^2-0.07 r^2\right) N_f^3 g_s^{21/4} \log ^4(N) \log ^3(r)}{N^{3/20} r^6 \alpha _{\theta_2}^2} \nonumber\\
   & & \hskip -0.5in-\beta^{1/4}\tilde{\omega}^3_{34}\Biggl[\frac{ {\cal C}_q \sqrt{{\cal C}_{zz}^{(1)}} N^{19/20} \left(33.3 a^2 r \log (r)+r^2\right) \left(\left(r^2-3.3
   a^2\right) \log (N)-87.7 a^2 r \log (r)\right)}{r^4 \alpha _{\theta_1}^{3/2} \alpha _{\theta_2}^3 \sqrt[4]{g_s} \log (N)}\Biggr]\nonumber
   \end{eqnarray}
   \begin{eqnarray}   
\label{Omega3bc}
& & \hskip -0.5in \Omega^3_{35}=-\frac{0.5 M N^{9/20} \left(r^2-3.3 a^2\right) \left(r^2-3. a^2\right) N_f g_s^{7/4} \log ^2(N) \log (r)}{r^4 \alpha _{\theta_1}^2 \alpha _{\theta_2}}\nonumber\\
   & & \hskip -0.5in -\beta^{1/4}\tilde{\omega}^3_{35}\Biggl[\frac{\sqrt{{\cal C}_{x^{10}x^{10}}^{(1)}} {\cal C}_q N^{5/4} \left(33.3 a^2 r \log (r)+r^2\right) \left(\left(1.
   r^2-3.3 a^2\right) \log (N)-87.7 a^2 r \log (r)\right)}{r^4 \alpha _{\theta_1}^{3/2} \alpha _{\theta_2}^4 \sqrt[4]{g_s} \log (N)}\Biggr]\nonumber\\
   & & \hskip -0.5in \Omega^3_{36}=\frac{0.4 M N^{9/20} \left(r^2-3. a^2\right) \left(r^2-2. a^2\right) N_f g_s^{7/4} \log ^2(N) \log (r)}{r^4 \alpha _{\theta_1}^2 \alpha _{\theta_2}} \nonumber\\
   & & \hskip -0.5in -\beta^{1/4}\tilde{\omega}^3_{36}\Biggl[\frac{ {\cal C}_q \sqrt{{\cal C}_{zz}^{(1)}} N^{5/4} \left(a^2-0.3 r^2\right) \left(a^2 \log (r)+0.03 r\right)}{r^3 \alpha _{\theta_1}^{3/2} \alpha _{\theta_2}^4 \sqrt[4]{g_s}}\Biggr]\nonumber\\
& & \hskip -0.5in \Omega^3_{45}=  \frac{178.2 M^3 \left(a^2-0.3 r^2\right) \left(0.2 a^2-0.07 r^2\right) \left(0.03 r^2-0.06 a^2\right) N_f^3 g_s^{21/4} \log ^4(N) \log ^3(r)}{N^{3/20} r^6 \alpha _{\theta_2}^2}\nonumber\\
   & & \hskip -0.5in -\beta^{1/4}\tilde{\omega}^3_{45}\Biggl[ \frac{ {\cal C}_q \sqrt{{\cal C}_{zz}^{(1)}} N^{19/20} \left(33.3 a^2 r \log (r)+r^2\right) \left(\left(r^2-3.3
   a^2\right) \log (N)-87.7 a^2 r \log (r)\right)}{r^4 \alpha _{\theta_1}^{3/2} \alpha _{\theta_2}^3 \sqrt[4]{g_s} \log (N)}\Biggr]\nonumber\\
   & & \hskip -0.5in \Omega^3_{46}=\frac{0.09 M^3 \left(3.3 a^2-r^2\right)^2 \left(r^2-2. a^2\right) N_f^3 g_s^{21/4} \log ^4(N) \log ^3(r)}{N^{3/20} r^6 \alpha _{\theta_2}^2}
   \nonumber\\
   & & \hskip -0.5in +\beta^{1/4}\tilde{\omega}^3_{46}\Biggl[ \frac{ {\cal C}_q \sqrt{{\cal C}_{zz}^{(1)}} M^2 N^{13/20} \sqrt{\alpha _{\theta_1}} N_f^2 g_s^{13/4} \log (N) \log ^2(r) \left(35.7143 a^2 r \log (r)-2.8 a^2+r^2\right)
   }{r^4 \alpha _{\theta_2}^5}\nonumber\\
   & & \hskip -0.5in \times \left(\left(r^2-3.3 a^2\right) \log (N)-87.7 a^2 r \log (r)\right)\Biggr]\nonumber\\
   & & \hskip -0.5in \Omega^3_{56}=-\frac{0.5 M N^{9/20} \left(r^2-3. a^2\right) \left(r^2-2. a^2\right) N_f g_s^{7/4} \log ^2(N) \log (r)}{r^4 \alpha _{\theta_1}^2 \alpha _{\theta_2}}+\beta^{1/4}\tilde{\omega}^3_{56}\Biggl[\frac{ {\cal C}_q \sqrt{{\cal C}_{zz}^{(1)}} N^{5/4} \left(a^2-0.3 r^2\right)}{r^2 \alpha _{\theta_1}^{3/2} \alpha _{\theta_2}^4 \sqrt[4]{g_s}}\Biggr]\nonumber\\
\end{eqnarray}
}
In (\ref{Omega3bc}), $\left|\tilde{\omega}^3_{bc}\right|\ll1$.

From (\ref{Omega1a}) - (\ref{Omega3bc}), one notes that in the MQGP limit (\ref{eq:MQGP}),
\begin{eqnarray}
\label{similar-Omega^1_a-beta0}
& & \left(\Omega^1_2\right)^{\beta^0} \sim \left(\Omega^1_3\right)^{\beta^0} \sim
-\left(\Omega^1_5\right)^{\beta^0},
\end{eqnarray}
\begin{eqnarray}
\label{similar-Omega^7_a-beta0}
& & \left(\Omega^7_2\right)^{\beta^0} \sim - \left(\Omega^1_2\right)^{\beta^0};\nonumber\\
& & \left(\Omega^7_3\right)^{\beta^0} \sim - \left(\Omega^7_4\right)^{\beta^0} \sim
- \left(\Omega^7_5\right)^{\beta^0},
\end{eqnarray}
and,
\begin{eqnarray}
\label{similar-Omega^a_bc-beta0}
& & \left(\Omega^3_{23}\right)^{\beta^0} \sim \left(\Omega^3_{24}\right)^{\beta^0}
\sim \left(\Omega^3_{25}\right)^{\beta^0} \sim \left(\Omega^3_{26}\right)^{\beta^0}
\sim \left(\Omega^3_{35}\right)^{\beta^0} \sim \left(\Omega^3_{36}\right)^{\beta^0}
\sim \left(\Omega^3_{56}\right)^{\beta^0},\nonumber\\
& & \left(\Omega^3_{34}\right)^{\beta^0} \sim \left(\Omega^3_{45}\right)^{\beta^0} \sim \left(\Omega^3_{46}\right)^{\beta^0}.
\nonumber\\
\end{eqnarray}

\section{$\HomG(\mathbf{27},\mathbf{7}) = 0$}
\label{HomG27270}
\setcounter{equation}{0} \seceqbb

\begin{definition}[$G_2$-equivariant maps]
Let $V$ and $W$ be finite-dimensional real representations of $G_2$, with  group homomorphisms defined via
$_V:G_2\to\mathrm{GL}(V)$ and $_W:G_2\to\mathrm{GL}(W)$.
The space of \emph{$G_2$-equivariant} linear maps is given by:
\[
  \HomG(V,W)
  \;=\;
  \bigl\{\,f:V\to W\;{linear}\;\bigm|\;
    f\bigl(_V(g)\,v\bigr) = _W(g)\,f(v)
    \;\;\forall\,g\in G_2,\;v\in V
  \,\bigr\}.
\]
\end{definition}

\noindent
The condition $f(_V(g)v)=_W(g)f(v)$ is referred to as \emph{equivariance}. If $g\cdot v$ represents the $G_2$-action, the condition is equivalent to $f(g\cdot v)=g\cdot f(v)$, i.e. the following diagram is commutative: $\forall g\in G_2$,
  \begin{tikzcd}[column sep=2.8cm, row sep=1.4cm]
    V \arrow[r,"f"] \arrow[d,"_V(g)"'] & W \arrow[d,"_W(g)"] \\
    V \arrow[r,"f"'] & W
  \end{tikzcd}
commutes.
In matrix form, if $F$ corresponds to the $(\dim W)\times(\dim V)$ matrix of $f$, equivariance reads
\[
  F\,_V(g) = _W(g)\,F
  \qquad\forall\,g\in G_2.
\]

\begin{lemma}
\label{thm:schur}
Let $G$ be a group and let $V$, $W$ be its irreducible representations.
\begin{enumerate}
\item[\normalfont(i)] {\bf (Over any field.)}
  If $V\not\cong W$, then $\mathrm{Hom}_G(V,W)=0$: every
  equivariant linear map $f:V\to W$ is identically zero.
\item[\normalfont(ii)] {\bf (Over $\mathbb{C}$ only.)}
  If $V\cong W$ and ${\cal F}=\mathbb{C}$, then
  $\mathrm{Hom}_G(V,V)=\mathbb{C}\cdot\mathrm{Id}_V$: every equivariant
  $f:V\to V$, is a scalar multiple of the identity, $f=\lambda\,\mathrm{Id}_V$
  for a unique $\lambda\in\mathbb{C}$.
\end{enumerate}
\end{lemma}

\begin{proof}[Proof of part (i)]
If $f:V\to W$ be $G$-equivariant, we show $f=0$ by analysing its kernel and image.

\medskip
\noindent
{\bf The kernel is $G$-invariant.}
If $v\in\ker(f)$, then  $\forall g\in G$:
\[
  f(g\cdot v) = g\cdot f(v) = g\cdot 0 = 0,
\]
so $g\cdot v\in\ker(f)$.
Hence $\ker(f)\subseteq V$ is a $G$-invariant subspace.

\medskip
\noindent
{\bf The image is $G$-invariant.}
If $w=f(v)\in\Im m(f)$, then $\forall g\in G$:
\[
  g\cdot w = g\cdot f(v) = f(g\cdot v)\in\Im m(f).
\]
Hence $\Im m(f)\subseteq W$ is a $G$-invariant subspace.

\medskip
\noindent
{\bf Applying irreducibility.}
Since $V$ is irreducible, its only $G$-invariant subspaces are $\{0\}$ and $V$.
Therefore:
\[
  \ker(f)=V \qquad{or}\qquad \ker(f)=\{0\}.
\]
{\bf Case 1:} $\ker(f)=V$.
Then $f=0$. 

\medskip
\noindent
{\bf Case 2:} $\ker(f)=\{0\}$.
Then $f$ is injective. Since $W$ is irreducible and $\Im m(f)\neq\{0\}$ is a $G$-invariant subspace of $W$, $\Im m(f)=W$.
Therefore $f:V\xrightarrow{\sim}W$ is an isomorphism, implying $V\cong W$.

This contradicts our assumption $V\not\cong W$. Hence Case 2 is impossible, and $f=0$ in all cases.
\end{proof}

\begin{proof}[Proof of part (ii)]
Let $f:V\to V$ be $G$-equivariant, with ${\cal F}=\mathbb{C}, {\cal F}$ being the ground field over which representations are defined, and $V$ irreducible.

\medskip
\noindent
{\bf $f$ has a complex eigenvalue.}
As ${\cal F}=\mathbb{C}$, the characteristic polynomial $\det(f-\lambda\,\mathrm{Id}_V)\in\mathbb{C}[\lambda]$
has degree $n=\dim_\mathbb{C} V\geq 1$. By the \emph{fundamental theorem of algebra}, every non-constant complex polynomial
has a root, so there exists $\lambda\in\mathbb{C}$ with $\det(f-\lambda\,\mathrm{Id}_V)=0$, meaning $f-\lambda\,\mathrm{Id}_V$ is not invertible.

\medskip
\noindent
{\bf $f-\lambda\,\mathrm{Id}_V$ is $G$-equivariant.}
So, $\forall g\in G$ and $v\in V$:
\[
  (f-\lambda\,\mathrm{Id}_V)(g\cdot v)
  = f(g\cdot v) - \lambda\,g\cdot v
  = g\cdot f(v) - \lambda\,g\cdot v
  = g\cdot(f-\lambda\,\mathrm{Id}_V)(v).
\]

\medskip
\noindent
{\bf $f=\lambda\,\mathrm{Id}_V$.}
The kernel $\ker(f-\lambda\,\mathrm{Id}_V)$ is a $G$-invariant subspace of $V$ 
(by the kernel argument of part (i)). Since $f-\lambda\,\mathrm{Id}_V$ is not invertible, $\ker(f-\lambda\,\mathrm{Id}_V)\neq\{0\}$.
By irreducibility: $\ker(f-\lambda\,\mathrm{Id}_V)=V$. Therefore, $(f-\lambda\,\mathrm{Id}_V)(v)=0$ for all $v\in V$, obtaining:
\[
  \boxed{f = \lambda\,\mathrm{Id}_V.}
\]

\end{proof}

If $V$ and $W$ are irreducible $G$-representations of \emph{different dimensions},
then $V\not\cong W$ automatically. So $\mathrm{Hom}_G(V,W)=0$ by Theorem~\ref{thm:schur}(i).

Since $\dim\mathbf{27} = 27 \neq 7 = \dim\mathbf{7}$, we have:
\[
  \mathbf{27} \;\not\cong\; \mathbf{7}.
\]
(Isomorphic representations must have equal dimension.) Therefore, 
\begin{lemma}
\label{HomG22770}
$\HomG(\mathbf{27},\mathbf{7}) = 0.$
\end{lemma}

Let $M_7$ be a 7-manifold with a $G_2$-structure $\varphi$, and let
$\Lambda^3_{\mathbf{27}}\subset\Omega^3(M_7)$ be the $\mathbf{27}$-component
of the $G_2$-representation decomposition of 3-forms.

Consider the wedge-product map
\[
  W:\Omega^3(M_7)\longrightarrow\Omega^6(M_7),
  \qquad
  \alpha\longmapsto\alpha\wedge\varphi.
\]

\begin{lemma}[$W$ is $G_2$-equivariant]
Now, $\forall g\in G_2$ and $\alpha\in\Omega^3(M_7)$:
$W(g\cdot\alpha)=g\cdot W(\alpha)$.
\end{lemma}
\begin{proof}
\[
  W(g\cdot\alpha)
  =(g\cdot\alpha)\wedge\varphi
  =g\cdot\!\bigl(\alpha\wedge(g^{-1}\cdot\varphi)\bigr)
  =g\cdot(\alpha\wedge\varphi)
  =g\cdot W(\alpha),
\]
where $g^{-1}\cdot\varphi=\varphi$ (every $g\in G_2$ preserves $\varphi$
by definition of $G_2$) has been used.
\end{proof}

\noindent
Under the identification $\Omega^6\cong\mathbf{7}$
(via the seven-dimensional Hodge dual: $\star_7:\Omega^6\xrightarrow{\sim}\Omega^1\cong\mathbf{7}$),
the restriction
\[
  W\big|_{\Lambda^3_{\mathbf{27}}}:
  \mathbf{27}\longrightarrow\mathbf{7}
\]
is a $G_2$-equivariant linear map, hence an element of
$\HomG(\mathbf{27},\mathbf{7})$.

By Lemma~\ref{HomG22770}, this space is zero:

\begin{corollary}\label{cor:main}
For every $\tau_3\in\Lambda^3_{\mathbf{27}}(M_7)$:
\[
  \tau_3\wedge\varphi = 0\;\in\;\Omega^6(M_7).
\]
\end{corollary}

\noindent
This is the identity used in deriving the contraction result
$(\tau_3)_{mnp}\,\psi_q{}^{mnp}=0$, as follows from the master identity
(proved by direct substitution of $\psi=\star_7\varphi$):
\[
  \alpha_{mnp}\,\psi_q{}^{mnp}
  =\frac{1}{6}\,\bigl[\star_7(\alpha\wedge\varphi)\bigr]_q
  \qquad\forall\;\alpha\in\Omega^3(M_7).
\]
Setting $\alpha=\tau_3$ and using Corollary~\ref{cor:main}:
\[
  (\tau_3)_{mnp}\,\psi_q{}^{mnp}
  =\frac{1}{6}\,\bigl[\star_7(\tau_3\wedge\varphi)\bigr]_q
  =\frac{1}{6}\,[\star_7(0)]_q
  = 0.
\]

\end{document}